\documentclass[reqno, 12pt]{amsart}

\usepackage{mathtools}
\usepackage[T1]{fontenc}
\usepackage{amsmath, amsthm, amssymb}
\usepackage{xcolor}
\usepackage{mathrsfs}
\usepackage{fullpage}
\usepackage{slashed}
\usepackage{graphicx} 
\usepackage[export]{adjustbox}
\usepackage[colorlinks=true,allcolors=red]{hyperref}
\usepackage{tensor}

\newtheorem{theorem}{Theorem}[section]
\newtheorem{definition}[theorem]{Definition}
\newtheorem{assumption}[theorem]{Assumption}
\newtheorem{lemma}[theorem]{Lemma}
\newtheorem{prop}[theorem]{Proposition}
\newtheorem{cor}[theorem]{Corollary}

\newtheorem*{que}{Question}

\theoremstyle{definition}

\theoremstyle{remark}
\newtheorem{remark}[theorem]{Remark}

\numberwithin{equation}{section}
\allowdisplaybreaks[4]

\newcommand{\Hquad}{\hspace{0.5em}}

\newcommand{\mf}[1]{\mathfrak{#1}}                                            
\newcommand{\mi}[1]{\mathscr{#1}}                                             
\newcommand{\mc}[1]{\mathcal{#1}}                                             
\newcommand{\ms}[1]{\mathsf{#1}}                                              

\newcommand{\R}{\mathbb{R}}                                                   
\newcommand{\Sph}{\mathbb{S}}                                                 

\DeclareMathOperator{\tr}{tr}

\let\div\undefined
\DeclareMathOperator{\div}{div}
\newcommand{\Dm}{\breve{\ms{D}}}
\newcommand{\nablam}{\breve{\nabla}}
\newcommand{\Boxm}{\breve{\Box}_g}
\newcommand{\Rm}{\breve{R}}
\newcommand{\hm}{\breve{\mf{h}}}

\begin{document}

\title{A gauge-invariant unique continuation criterion for waves in asymptotically anti-de Sitter spacetimes}

\author{Athanasios Chatzikaleas}

\address{Westf{\"a}lische Wilhelms-Universit{\"a}t M{\"u}nster, Mathematical Institute, Einsteinstrasse 62, 48149 Munster, Germany} 
\email{achatzik@uni-muenster.de}

\author{Arick Shao}
\address{School of Mathematical Sciences\\
Queen Mary University of London\\
London E1 4NS\\
United Kingdom}
\email{a.shao@qmul.ac.uk}

\begin{abstract}
We reconsider the unique continuation property for a general class of tensorial Klein-Gordon equations of the form 
\begin{align*}
\Box_{g} \phi + \sigma \phi = \mathcal{G}(\phi,\nabla \phi) \text{,} \qquad \sigma \in \mathbb{R}
\end{align*}
on a large class of asymptotically anti-de-Sitter spacetimes.
In particular, we aim to generalize the previous results of Holzegel, McGill, and the second author \cite{Arick1, Arick2, Arick3} (which established the above-mentioned unique continuation property through novel Carleman estimates near the conformal boundary) in the following ways:
\begin{enumerate}
\item We replace the so-called \emph{null convexity criterion}---the key geometric assumption on the conformal boundary needed in \cite{Arick3} to establish the unique continuation properties---by a more general criterion that is also \emph{gauge invariant}.

\item Our new unique continuation property can be applied from a larger, more general class of domains on the conformal boundary.

\item Similar to \cite{Arick3}, we connect the failure of our generalized criterion to the existence of certain null geodesics near the conformal boundary.
These geodesics are closely related to the classical Alinhac-Baouendi counterexamples to unique continuation \cite{MR1363855}.
\end{enumerate}
Finally, our gauge-invariant criterion and Carleman estimate will constitute a key ingredient in proving unique continuation results for the full nonlinear Einstein-vacuum equations, which will be addressed in a forthcoming paper of Holzegel and the second author \cite{HolzegelShaoEinsteinPreparation}.
\end{abstract}

\maketitle
\tableofcontents
\noindent

\section{Introduction} \label{sec.intro}

The Anti-de Sitter (or AdS) spacetime $\left( \mc{M}_{\text{AdS}}, g_{\text{AdS}} \right)$ is the maximally symmetric solution to the Einstein-vacuum equations
\begin{align}\label{EinsteinVacuum}
	R_{ \alpha \beta} -\frac{1}{2} g _{\alpha \beta} R + \Lambda g_{\alpha \beta} = 0
\end{align}
with negative cosmological constant $\Lambda<0$.
With respect to polar coordinates,
\begin{align*}
(r, t, \omega) \in \mc{M}_{\text{AdS}} := \mathbb{R} \times [0,\infty) \times \mathbb{S}^{n-1} \text{,}
\end{align*}
this solution reads
\begin{align*}
	g_{\text{AdS}} (r, t, \omega) = \left( 1+r^2 \right)^{-1} d r^2 - \left( 1+r^2 \right) d t^2 + r^2 d \omega^2 \text{,}
\end{align*}
where $d\omega^2$ is the standard round metric on $\mathbb{S}^{n-1}$.

For our purposes, it is more convenient to introduce the coordinate $\rho$ given by
\begin{align*}
4 r := \rho^{-1} ( 2 + \rho ) ( 2 - \rho ) \text{.}
\end{align*}
With respect to the parameters $( t, \rho, \omega )$, AdS spacetime is described as
\begin{align}
\label{AdS} (t, \rho, \omega) \in \mc{M}_{\text{AdS}} &\approx \mathbb{R} \times \left( 0, 2 \right] \times \mathbb{S}^{n-1} \text{,} \\
\notag g_{ \text{AdS} } ( \rho, t, \omega ) &= \rho^{-2} \left[ d \rho^2 + ( - dt^2 + d \omega^2 ) - \frac{1}{2} \rho^2 ( dt^2 + d \omega^2 ) + \frac{1}{16} \rho^4 ( - dt^2 + d \omega^2 ) \right] \text{,}
\end{align}
which we refer to as its \emph{Fefferman-Graham form}.
Moreover, ignoring the factor $\rho^{-2}$ in \eqref{AdS}, we can then formally attach a boundary $( \mc{I}_{\text{AdS}}, \mathfrak{g}_{\text{AdS}} )$ to AdS spacetime:
\begin{align*}
\mc{I}_{\text{AdS}} := \{ \rho = 0 \} \approx \R \times \Sph^{n-1} \text{,} \qquad \mf{g}_{\text{AdS}} := - dt^2 + d \omega^2 \text{.}
\end{align*}
We refer to $( \mc{I}_{\text{AdS}}, \mf{g}_{\text{AdS}} )$ as the \emph{conformal boundary} of AdS spacetime.

We use the term \emph{asymptotically AdS} (aAdS) to refer to a class of spacetimes $( \mc{M}, g )$ that ``have a similar structure with a conformal boundary''.
More specifically, such spacetimes can, at least near its conformal boundary, be written in the form
\begin{align}\label{FGintro}
\mc{M} := ( 0, \rho_0 ) \times \mc{I} \text{,} \qquad g (\rho, x) := \rho^{-2} \left[ d\rho^2 + \mathfrak{g} (x) + \rho^2 \bar{\mathfrak{g}} (x) + \mc{O}(\rho^3) \right] \text{,}
\end{align}
for some Lorentzian manifold $( \mc{I}, \mathfrak{g} )$ and symmetric $(0,2)$-tensor $\bar{\mathfrak{g}}$.
Again, we call the specific form \eqref{FGintro} of the metric the \textit{Fefferman-Graham} gauge, and we refer to $( \mc{I}, \mf{g} )$ as the associated \emph{conformal boundary}.
In particular, observe that the conformal boundary is allowed to have arbitrary boundary topology and geometry.

AdS spacetime \eqref{AdS} is one example of an aAdS spacetime.
The AdS-Schwarzschild and AdS-Kerr families are also aAdS spacetimes, all with conformal boundary $( \mc{I}_{\text{AdS}}, \mathfrak{g}_{\text{AdS}} )$.
\footnote{See \cite{2020arXiv200807396S} for the AdS-Schwarzschild metrics expressed in Fefferman-Graham gauge.}

Aside from the fact that aAdS solutions of \eqref{EinsteinVacuum} are of interest in the context of general relativity, they have also received considerable attention due to the celebrated AdS/CFT correspondence \cite{0501128} brought to light by Maldacena \cite{MR1730135, MR1705508, MR1633016}.
Such a duality relates string theories within a universe with a negative cosmological constant (AdS) to conformally invariant quantum field theories (CFT) on the conformal boundary of the spacetime and has important applications \cite{MR1633012, MR2551709, 09090518, 12100890, 0501128}. 

\subsection{Unique continuation and previous results}

As the conformal boundary $\mc{I}$ is timelike, it fails to be a Cauchy hypersurface for the associated aAdS spacetime.
Consequently, the problem---inspired from AdS/CFT---of solving the Klein-Gordon equations
\begin{align}\label{KleinGordon}
	\Box_{g} \phi +\sigma \phi = \mc{G}(\phi,\nabla \phi) \text{,} \qquad \sigma \in \mathbb{R}
\end{align}
on a fixed aAdS background given Cauchy data (i.e., appropriately defined Dirichlet \emph{and} Neumann traces) on $\mc{I}$ may not be well-posed.

As a result, we consider instead a different but related mathematical problem:

\begin{que}
  Are solutions (whenever they exist) of the Klein-Gordon equation \eqref{KleinGordon} uniquely determined by the (Dirichlet and Neumann) boundary data on $\mc{I}$?
\end{que}

In terms of physics, the above can be interpreted as follows:

\begin{que}
	Is there a ``one-to-one correspondence'' between solutions of \eqref{KleinGordon} in the interior and an appropriate space of boundary data on the conformal boundary $\mc{I}$?
\end{que}

The standard unique continuation results for wave equations follow from the now-classical results of H\"ormander for general linear operators---see \cite[Chapter 28]{hor:lpdo4} and \cite{ler_robb:unique}.
There, the crucial criterion required is that the hypersurface from which one uniquely continues must be \emph{pseudoconvex}.
Unfortunately, these classical results fail to apply to aAdS settings, as the conformal boundary $\mc{I}$ (barely) fails to be pseudoconvex.
Thus, a more refined understanding of the near-boundary geometry is needed to deduce unique continuation properties.

\begin{remark}
For instance, one consequence of these new difficulties is that while the classical results can be localized to a sufficiently small neighborhood of a point, one must to prescribe data on a sufficiently large part of $\mc{I}$ in order to uniquely continue the solution. 
See the introduction of \cite{Arick1} for more extensive references and discussions of these issues.
\end{remark}

The first positive answer to the above questions was provided by Holzegel and the second author in \cite{Arick1}, but only in the case of a static conformal boundary.
This was later extended by the same authors, in \cite{Arick2}, to non-static boundaries.
In both \cite{Arick1, Arick2}, the key results were proved by establishing degenerate Carleman estimates near the conformal boundary.

The most current results in this direction, which further improved upon \cite{Arick1, Arick2}, are given by McGill and the second author in \cite{Arick3}.
Its main result can be roughly stated as follows:
 
\begin{theorem}[\cite{Arick3}] \label{PreThm2}
	Let $(\mc{M},g)$ be an aAdS spacetime expressed in the Fefferman-Graham gauge \eqref{FGintro}.
  Furthermore, assume the following:
	\begin{enumerate}
    \item $t$ is a time function on $\mc{I}$.

    \item There exists a constant $p>0$ such that $\mc{G}$, from \eqref{KleinGordon}, satisfies
		\begin{align*}
			\left| \mc{G} ( \psi, \nabla \psi) \right|^2 \lesssim  \rho^{3p} \left| \psi \right|^2 + \rho^{4+p} \left| \nabla \psi \right|^2 \text{.}
		\end{align*}

    \item $(\mc{M}, g)$ satisfies the following \emph{null convexity criterion}:
		\begin{enumerate}
			\item[(3a)] There exists some constant $A > 0$ such that the following lower bound holds for all tangent vectors $\mathfrak{X} \in T \mc{I}$ such that $\mathfrak{g} (\mathfrak{X}, \mathfrak{X}) = 0$:
		\begin{align}\label{PseudoConvexityCriterionIntro3}
			- \bar{\mathfrak{g}} ( \mathfrak{X}, \mathfrak{X} )  \geq A  \left( \mathfrak{X}t \right)^2. 
		\end{align}
  		\item[(3b)] There exists a constant $B \geq 0$ such that the following upper bound holds for all tangent vectors $\mathfrak{X} \in T \mc{I}$ with $\mathfrak{g}  ( \mathfrak{X}, \mathfrak{X}) = 0$:
		\begin{align}\label{PseudoConvexityCriterionIntro4}
			\left| \mathfrak{D}^2 t (\mathfrak{X}, \mathfrak{X}) \right| \leq B \left( \mathfrak{X} t \right)^2 \text{,}
		\end{align}	
		where $\mathfrak{D}^2 t$ denotes the Hessian of $t$ with respect to $\mathfrak{g}$.
		\end{enumerate} 
	\end{enumerate} 
  Then, if $\phi$ is any scalar or tensorial solution to the Klein--Gordon equation \eqref{KleinGordon} such that
  \begin{itemize}
  \item $\phi$ is compactly supported on each level set of $(\rho, t)$, and

  \item $\rho^{-\kappa} \phi \rightarrow 0$ in $C^1$ as $\rho \rightarrow 0$ over a sufficiently large timespan $\{t_0 \leq t \leq t_{1}\}$; where $\kappa$ depends on $g$, $t$, the rank of $\phi$, and $\sigma$; and where $t_1 - t_0$ depends on $A$ and $B$,
  \end{itemize}
	then $\phi = 0$ in a neighborhood in $\mc{M}$ of the boundary region $\{ t_0 \leq t \leq t_{1} \} \subseteq \mc{I}$.
\end{theorem}

\begin{remark}
 	The assumption of sufficiently large timespan $t_1 - t_0$ is necessary in general.
  In particular, if $t_1 - t_0$ is small, and if $\bar{\mf{g}}$ satisfies some additional assumptions, \footnote{See the discussions in \cite{Arick3}.} then one can find a sequence of $g$-null geodesics in $\mc{M}$ such that
  \begin{itemize}
  \item the geodesics become arbitrarily close to the conformal boundary, and

  \item the geodesics ``fly over the boundary region $\{ t_0 \leq t \leq t_1 \}$ without  touching it''---that is, part of the geodesic projects onto this region, but the geodesic does not itself terminate at the conformal boundary in this region.
  \end{itemize}
One can then, at least in principle, apply the classical results of Alinhac-Baouendi \cite{MR1363855} to construct geometric optics counterexamples to unique continuation, using functions that are supported near these null geodesics but away from $\{ t_0 \leq t \leq t_1 \} \subseteq \mc{I}$; this will be rigorously carried out in the upcoming work of Guisset \cite{Simon}.\footnote{We do wish to emphasize that constructing these counterexamples is by no means immediate.
One issue is that the results of \cite{MR1363855} only directly applies when $\sigma$ has the value of the conformal mass.
For other values of $\sigma$, the conformal renormalization of \eqref{KleinGordon} has a singular potential at the boundary.}
This mechanism prevents one from having a general unique continuation result from sufficiently small domains in $\mc{I}$.
\end{remark}

\begin{remark}
  In \cite{Arick3}, the authors also provided a systematic study of such null geodesics near the conformal boundary.
  In particular, they showed that if the null convexity condition holds in the setting of Theorem \ref{PreThm2}, then the null geodesics described in the preceding remark cannot ``fly over'' a timespan of greater than $t_1 - t_0$.
  Thus, one can view the null convexity criterion as the ingredient needed to eliminate the Alinhac-Baoundi counterexamples.
\end{remark}

\begin{remark}
AdS spacetime satisfies the null convexity criterion, and, in this case, one requires $t_1 - t_0 > \pi$ for unique continuation.
In particular, one can find a family of null geodesics emanating from the AdS conformal boundary $\mc{I}_{\text{AdS}}$ which remain arbitrarily close to $\mc{I}_{\text{AdS}}$ and which return to $\mc{I}_{\text{AdS}}$ exactly after time $\pi$.
\end{remark}

\begin{remark}
On the other hand, this large timespan requirement is likely not needed whenever $\mc{G}$ is linear and both $g$ and $\mc{G}$ are analytic (due to Holmgren's theorem \cite{MR2108588}) or even time-analytic (due to the unique continuation results of Tataru; see \cite{hor:uc_interp, tat:uc_hh, tat:uc_hh_2}).
 \end{remark}

\begin{remark}
	Theorem \ref{PreThm2} improves upon the main results of \cite{Arick1, Arick2} in the following regards:
  \begin{itemize}
  \item Theorem \ref{PreThm2} allows for general time functions, whereas \cite{Arick1, Arick2} assumed a time foliation for which the unit normals must be geodesic.

  \item The null convexity criterion \eqref{PseudoConvexityCriterionIntro3}--\eqref{PseudoConvexityCriterionIntro4} is slightly weaker than the corresponding \emph{pseudoconvexity criterion} from \cite{Arick2}.
  In particular, the null convexity criterion pertains only to $\mf{g}$-null directions, while the pseudoconvexity criterion in \cite{Arick2} requires a corresponding bound in all directions along $\mc{I}$.
  \end{itemize}
  See the discussions in \cite{Arick3} for further comparisons among \cite{Arick1, Arick2, Arick3}.
\end{remark}

Finally, we mention the articles \cite{alex_schl_shao:uc_inf, alex_shao:nc_inf, alex_shao:uc_global, enc_shao_verga:obs_swave, peters:cpt_cauchy}, which apply Lorentzian-geometric methods to prove degenerate Carleman estimates and unique continuation results for waves in other geometric settings of interest.
Furthermore, various examples of applications of unique continuation results in relativity can be found in \cite{alex_io_kl:hawking_anal, alex_io_kl:unique_bh, alex_io_kl:rigid_bh, alex_schl:time_periodic, io_kl:unique_bh, peters:cpt_cauchy}.

\subsection{Gauge transformations and invariance}

Another important idea in the description of aAdS spacetimes and their conformal boundaries is that of \emph{gauge invariance}.
Roughly speaking, the term \emph{gauge transformation} refers to a change of parameters on $\mc{M}$,
\begin{align} \label{IntroGauge}
( \rho, x ) \longrightarrow ( \sigma, y ) \text{,} \qquad \rho, \sigma > 0 \text{,} \quad x, y \in \mc{I} \text{,}
\end{align}
that preserves the Fefferman-Graham gauge \eqref{FGintro}, that is, $g$ is also expressed as
\begin{align} \label{FGintro2}
	 g ( \sigma, y ) := \sigma^{-2} \left[ d \sigma^2 + \mathfrak{p} (y) + \sigma^2 \bar{\mathfrak{p}} (y) + \mc{O}(\sigma^3) \right] \text{.}
\end{align}

On one hand, such a gauge transformation clearly leaves the spacetime $( \mc{M}, g )$ unchanged, so that \eqref{FGintro} and \eqref{FGintro2} describe the same physical object.
However, the conformal boundary is \emph{not} left invariant by gauge transformations.
In particular, the boundary metrics $\mf{g}$ and $\mf{p}$ fail to coincide; in fact, it is well-known that they differ via a conformal factor.
Similarly, the subsequent coefficients $\bar{\mf{g}}$ and $\bar{\mf{p}}$ are also related through a more complicated formula (namely, the conformal transformation law for the Schouten tensor).
\footnote{Gauge tranformation laws can also be derived for higher-order coefficients in the Fefferman-Graham expansion not listed in \eqref{FGintro}; see, for instance, \cite{deharo_sken_solod:holog_adscft, imbim_schwim_theis_yanki:diffeo_holog} for details.}
Precise descriptions of gauge transformations and their effects on the conformal boundary are given in Section \ref{sec.gncc}.

Now, \emph{one key criticism of Theorem \ref{PreThm2} is that its crucial geometric assumption, the null convexity criterion, fails to be gauge-invariant}.
More specifically, it is possible that \eqref{PseudoConvexityCriterionIntro3}--\eqref{PseudoConvexityCriterionIntro4} hold in one gauge, but not after applying a gauge transformation \eqref{IntroGauge}.
(A similar shortcoming holds for the corresponding geometric assumptions for \cite{Arick1, Arick2}.)
Thus, according to Theorem \ref{PreThm2}, whether the unique continuation property holds depends not only on the properties of the physical system, but also on the gauge with which we choose to view it.
This is clearly an undesirable feature of the existing unique continuation results.

Therefore, a key objective of the present article is to further generalize Theorem \ref{PreThm2} by \emph{replacing the null convexity criterion by a gauge-invariant condition}.
This will be important for upcoming applications to AdS/CFT, where one again wishes to state any assumptions on the conformal boundary in a gauge-invariant---and hence physically relevant---manner.

\subsection{Statement of main results}

We now state the main results of this paper.
The first is an analogue of Theorem \ref{PreThm2}, except that the null convexity criterion \eqref{PseudoConvexityCriterionIntro3}--\eqref{PseudoConvexityCriterionIntro4} is replaced by a gauge-invariant condition, which we call the \emph{generalized null convexity criterion}.
This fulfills the goal that was discussed in the preceding subsection.

\begin{theorem}[Main Result 1: Unique continuation]\label{MainTheorem1}
	Let $(\mc{M},g)$ be an aAdS spacetime expressed in the Fefferman-Graham gauge \eqref{FGintro}.
  Furthermore, assume the following:
	\begin{enumerate}
    \item There exists a constant $p>0$ such that $\mc{G}$, from \eqref{KleinGordon}, satisfies
		\begin{equation} \label{PDEcoeff}
			\left| \mc{G} ( \psi, \nabla \psi) \right|^2 \lesssim  \rho^{3p} \left| \psi \right|^2 + \rho^{4+p} \left| \nabla \psi \right|^2 \text{.}
		\end{equation}

    \item $\mc{D} \subseteq \mc{I}$ is a domain that satisfies the following \emph{generalized null convexity criterion}:\ There exists a constant $c>0$ and a smooth function $\eta: \bar{\mc{D}} \longrightarrow \mathbb{R}$ such that for all tangent vectors $\mathfrak{X} \in T \bar{\mc{D}}$ with $\mathfrak{g} (\mathfrak{X},\mathfrak{X}) = 0$,
		\begin{align}\label{PDEintro}
			\begin{dcases}
				\left[ \mathfrak{D}^2 \eta - \eta \cdot \bar{\mathfrak{g}} \right] (\mathfrak{X}, \mathfrak{X}) > c \cdot \mathfrak{h} ( \mathfrak{X}, \mathfrak{X} ) \text{ in } \mc{D} \text{,} \\
				\eta > 0 \text{ in } \mc{D} \text{,} \\
				\eta = 0 \text{ on } \partial \mc{D} \text{,}
			\end{dcases}
		\end{align}
		where $\mathfrak{D}^2 \eta $ denotes the Hessian of $\eta$ with respect to $\mathfrak{g}$, and where $\mf{h}$ denotes an arbitrary but fixed Riemannian metric on $\mc{I}$.
	\end{enumerate} 
  Then, if $\phi$ is any scalar or tensorial solution to the Klein--Gordon equation \eqref{KleinGordon} such that
  \begin{itemize}
  \item $\phi$ is compactly supported on each level set of $(\rho, t)$, and

  \item $\rho^{-\kappa} \phi \rightarrow 0$ in $C^1$ as $\rho \rightarrow 0$ over $\mc{D}$; where $\kappa$ depends on $g$, the rank of $\phi$, and $\sigma$,
  \end{itemize}
	then $\phi = 0$ in a neighborhood in $\mc{M}$ of the boundary domain $\mc{D} \subseteq \mc{I}$.
\end{theorem}

Theorem \ref{MainTheorem1} states that solutions of \eqref{KleinGordon} can be uniquely continued from the boundary region $\mc{D}$ as long as $\mc{D}$ satisfies the generalized null convexity criterion.
Note this generalized null convexity criterion plays the same role as the null convexity criterion in Theorem \ref{PreThm2}.
Later, in Proposition \ref{Gaugeinvariance}, we will show that the generalized null convexity criterion is indeed gauge-invariant.
The precise version of Theorem \ref{MainTheorem1} is given later as Corollary \ref{CorollaryUniqueContinuation}.

\begin{remark}
We note that Theorem \ref{MainTheorem1} is a generalization of Theorem \ref{PreThm2}.
In particular, if the hypotheses of Theorem \ref{PreThm2} hold, then so does the hypotheses of Theorem \ref{MainTheorem1}, with $\mc{D} := \{ t_0 < t < t_1 \}$, and with $\eta$ being a specially chosen function of $t$.
For more details, in particular on the choice of $\eta$ here, see Proposition \ref{NCC_GNCC} and its proof.

Furthermore, note that Theorem \ref{PreThm2} is only applicable to boundary domains $\mc{D}$ of the form $\{ t_0 < t < t_1 \}$, while Theorem \ref{MainTheorem1} extends this to more general $\mc{D}$.
\end{remark}

\begin{remark}
In this article, while we restrict our attention to smooth quantities for convenience, however we require far less regularity for Theorem \ref{MainTheorem1} to hold.
For instance, we only require our metric $g$ to be $C^3$ in the proof of our Carleman estimates.
Furthermore, we only require that any lower-order coefficients decay toward $\mc{I}$ at the rates encoded in \eqref{PDEcoeff}.
\end{remark}

\begin{remark}
In general, the vanishing exponent $\kappa$ in Theorem \ref{MainTheorem1} must be sufficiently large to absorb any critical lower-order terms.
However, when $\psi$ is scalar, or when the conformal boundary is static, one can take $\kappa$ to be the optimal exponent from applying a separation of variables ansatz to $\psi$ \cite{breit_freedm:stability_sgrav}.
See Remark \ref{rmk.kappa} for further details and precise values.
\footnote{In practice, a large $\kappa$ does not cause any issues, since given sufficient regularity, one can derive that $\psi$ vanishes to arbitrarily higher order as long as $\psi$ satisfies the optimal vanishing rate in Remark \ref{rmk.kappa}.}
\end{remark}

Next, similar to \cite{Arick3}, we also establish a theorem connecting the generalized null convexity condition in Theorem \ref{MainTheorem1} to the behavior of null geodesics near the conformal boundary.

\begin{theorem}[Main result 2: Characterization of null geodesics]\label{MainTheorem2}
  Let the aAdS spacetime $( \mc{M}, g )$ be as in Theorem \ref{MainTheorem1}, and let $\mc{D} \subseteq \mc{I}$ satisfy the generalized null convexity criterion \eqref{PDEintro}.
  Let $\Lambda: ( s_-, s_+ ) \rightarrow \mc{M}$ be a complete $g$-null geodesic, written as
  \begin{align*}
  \Lambda(s) = \left( \rho(s), \lambda(s) \right) \in ( 0, \rho_0 ) \times \mc{I} \text{.}
  \end{align*}
  Then, there exists $\epsilon_0 > 0$ sufficiently small (depending on $g$ and $\mc{D}$) such that if
	\begin{align} \label{IntroNullGeodesic}
		0 < \rho (s_0) < \epsilon_0 \text{,} \qquad | \dot{\rho}(s_0) | \lesssim \rho (s_0) \text{,} \qquad \lambda(s_0) \in \mc{D} \text{,}
	\end{align}
  for some $s_0 \in ( s_-, s_+ )$, then at least one of the following statements holds:
  \begin{itemize}
    \item $\Lambda$ initiates from the conformal boundary within $\mc{D}$:
\begin{align*}
	\lim_{s \searrow s_- } \rho (s) = 0 \text{,} \qquad \lim_{ s \searrow s_- } \lambda (s) \in \mc{D} \text{.}
\end{align*}	

    \item $\Lambda$ terminates at the conformal boundary within $\mc{D}$:
\begin{align*}
	\lim_{s \nearrow s_+ } \rho (s) = 0 \text{,} \qquad \lim_{ s \nearrow s_+ } \lambda (s) \in \mc{D} \text{.}
\end{align*}
\end{itemize}
\end{theorem}

Theorem \ref{MainTheorem2} can be interpreted as follows.
The assumption \eqref{IntroNullGeodesic} tells us that the null geodesic is, at parameter $s_0$, both sufficiently close to the conformal boundary and over the region $\mc{D}$.
The conclusions of Theorem \ref{MainTheorem2} then imply that $\Lambda$ must either start from or terminate at the conformal boundary within $\mc{D}$.
In other words, there are no null geodesics that are sufficiently close to the conformal boundary and that ``fly over $\mc{D}$ without touching $\mc{D}$''.
As a result, Theorem \ref{MainTheorem2} implies that as long as $\mc{D}$ satisfies the generalized null convexity condition, then one cannot construct the geometric optics counterexamples of Alinhac-Baouendi \cite{MR1363855} to unique continuation from the region $\mc{D}$.

A more precise statement of Theorem \ref{MainTheorem2} is given later in Theorem \ref{Characterisation}.

\subsection{Main Ideas of the Proof}

We now describe some of the ideas of the proofs of Theorems \ref{MainTheorem1} and \ref{MainTheorem2}.
As most parts of the proofs have direct analogues in \cite{Arick1, Arick2, Arick3}, here we will focus primarily on how the present proofs differ from those of previous works.

First, the main ingredient in the proof of Theorem \ref{MainTheorem1} is a new Carleman estimate for the Klein--Gordon operator near the conformal boundary.
(As is common in unique continuation literature, Theorem \ref{MainTheorem1} follows from this Carleman estimate via a standard argument.)
For the Carleman estimate itself, the most noteworthy feature here is the function $f$ that defines the weight within the estimate, and for which the level sets are strongly pseudoconvex.
\footnote{Since the conformal boundary barely fails to be pseudoconvex (see the discussions in \cite{Arick1, Arick2, Arick3}), then the pseudoconvexity of the level sets of $f$ must degenerate as one approaches the conformal boundary.}

In \cite{Arick1, Arick2, Arick3}, this $f$ was defined to be of the form
\begin{align*}
f := \frac{ \rho }{ \eta_\ast (t) } \text{,}
\end{align*}
where $t$ is a given time function on the spacetime and conformal boundary.
Note in particular that $\{ f = 0 \}$ lies on the conformal boundary, while $\{ f = c \}$ for any $c > 0$ lies in the interior as long as $\eta_\ast (t) > 0$.
Moreover, all the level sets of $f$ focus at the conformal boundary at any point where $\eta_\ast (t) = 0$.
In \cite{Arick3}, it was shown that if the null convexity criterion holds, and if $\eta_\ast$ is a specific function defined in terms of the constants $A$ and $B$ in \eqref{PseudoConvexityCriterionIntro3}--\eqref{PseudoConvexityCriterionIntro4}, then the level sets of $f$ are pseudoconvex for small enough values of $f$.

For our new Carleman estimate, we take instead
\begin{align*}
f := \frac{ \rho }{ \eta } \text{,}
\end{align*}
the key difference being that $\eta_\ast (t)$ is replaced by a more general function $\eta$, defined on some subset of $\mc{I}$, that is allowed to depend on all coordinates of $\mc{I}$.
In particular, we consider a larger class of foliations $f$ that can depend also on the spatial components of $\mc{I}$.
A crucial computation in this article (see Lemma \ref{Lemmapi}) shows that if the generalized null convexity criterion \eqref{PDEintro} holds, then the level sets of this new $f$ are once again pseudoconvex near the conformal boundary.
Moreover, note that in the setting of \eqref{PDEintro}, the level set $\{ f = 0 \}$ corresponds to the boundary domain $\mc{D}$, while the level sets $\{ f = c \}$, with $c > 0$, are hypersurfaces ``flying over $\mc{D}$'' that refocus at the conformal boundary on $\partial \mc{D}$.

From this point, the proof of the Carleman estimate resembles that of \cite{Arick3}, but with our new $f$ replacing the previous.
\footnote{There are several other major subtleties and difficulties in establishing Carleman estimates near the conformal boundary.
The reader is referred to \cite{Arick1, Arick2, Arick3} for further discussions of these.}
Here, as was done in \cite{Arick3}, we will once again apply the same formalism of vertical fields---tensor fields on $\mc{M}$ that are fully tangent to level sets of $\rho$.

\begin{remark}
We also note again that the Carleman estimate behind Theorem \ref{MainTheorem1}, and in particular the generalized null convexity criterion, will serve as essential ingredients in proving unique continuation results for the full nonlinear Einstein equations, which will be addressed in a forthcoming paper \cite{HolzegelShaoEinsteinPreparation}.
\end{remark}

Next, for Theorem \ref{MainTheorem2}, the idea behind the proof is the same as in \cite{Arick3}---a detailed analysis of trajectories of null geodesics sufficiently near the conformal boundary.
However, since the generalized null convexity criterion is rather different in nature from the null convexity criterion, the process for carrying this out will differ as well.

The key observation for a null geodesic $\Lambda (s) = ( \rho (s), \lambda (s) )$, as in the statement of Theorem \ref{MainTheorem2}, is as follows.
First, reformulate the $\rho$-component of $\Lambda$ as a function on the image of $\lambda$,
\begin{align*}
\vartheta ( \lambda (s) ) := \rho (s) \text{.}
\end{align*}
Then, one can show that up to leading order, $\vartheta$ satisfies
\footnote{To make sense of the covariant derivative $\mf{D}$, one can extend $\vartheta$ arbitrarily away from the image of $\lambda$.}
\begin{align*}
\left. \left( \mf{D}^2 \vartheta - \vartheta \cdot \bar{\mf{g}} \right) \right|_{ \lambda (s) } ( \dot{\lambda} (s), \dot{\lambda} (s) ) = 0 + \text{lower-order error terms} \text{.}
\end{align*}
Moreover, as long as $\Lambda$ makes a narrow angle with the conformal boundary (i.e., $\dot{\rho}$ is sufficiently small), then $\dot{\lambda} (s)$ will be close to a $\mf{g}$-null vector.
\footnote{The equation $\mf{D}^2 \vartheta - \vartheta \mf{g} = 0$ can be viewed as a generalization of the damped harmonic oscillators that controlled near-boundary null geodesics in \cite{Arick3}.
In particular, if $\vartheta$ depends on only a single time coordinate $t$, then the above reduces to a family of damped harmonic oscillator equations.}

In other words, as long as $\Lambda$ is sufficiently close to and makes a sufficiently narrow angle with the conformal boundary, then its $\rho$-component is approximately governed by the first line of \eqref{PDEintro}, but with ``$> c \cdot \mf{h} ( \mf{X}, \mf{X} )$'' replaced by by ``$= 0$''.
In this way, $\eta$ in \eqref{PDEintro} can be thought of as strictly constraining the $\rho$-values of any such null geodesic $\Lambda$.

As a result, if $\Lambda$ satisfies the assumption \eqref{IntroNullGeodesic}, then one can (somewhat similarly to \cite{Arick3}) use the above observation, in tandem with a Sturm-type comparison $\eta$, to ensure that one of the two possibilities in the conclusion of Theorem \ref{MainTheorem2} must hold.
However, an additional complication is that the actual geodesic equations contain various nonlinear terms that, in our setting, can be viewed as error.
Therefore, like in \cite{Arick3}, one must couple the above with an elaborate continuity argument to ensure that these nonlinear error terms remain negligibly small throughout the trajectory of the geodesic.

\begin{remark}
Note this connection between $\eta$ and null geodesics also provides the intuition as to why the level sets of our $f := \rho \eta^{-1}$ in the Carleman estimates are pseudoconvex.
In particular, this geodesic constraining property of $\eta$ ensures that if a null geodesic $\Lambda$ (near the conformal boundary) hits a point of $\{ f = c \}$ tangentially, then at nearby points, $\Lambda$ must lie within $\{ f < c \}$---that is, closer to the conformal boundary.
This is precisely the geometric characterization of $\{ f = c \}$ being pseudoconvex (from the conformal boundary inward).
\end{remark}

\begin{remark}
This intuition also gives an explanation for the gauge-invariance of the generalized null convexity criterion.
Since $\eta$ is directly connected to trajectories of near-boundary null geodesics in $\mc{M}$ (in particular, to when such geodesics reach the conformal boundary), and since null geodesics are obviously invariant under coordinate changes, then this suggests that the generalized null convexity criterion should have a gauge-invariant formulation.
\end{remark}

\subsection{Organization of the Paper}

In Section \ref{sec.prelim}, we define precisely the asymptotically AdS settings on which our main results will hold.
In addition, we collect a number of computations and observations on these spaceetimes that will be useful in later sections.

In Section \ref{sec.gncc}, we discuss the generalized null convexity criterion.
First, we show that this condition is gauge invariant.
We then relate this to the null convexity criterion of \cite{Arick3}, and we study some basic examples in which this criterion holds or fails.

In Section \ref{sec.geodesic}, we give a precise statement and a proof of Theorem \ref{MainTheorem2}.

In Section \ref{sec.carleman}, we give a precise statement and a proof of Theorem \ref{MainTheorem1}.

\subsection{Acknowledgments}

The authors thank Gustav Holzegel for useful discussions. Furthermore, the first author gratefully acknowledges the support of the ERC grant 714408 GEOWAKI, under the European Union's Horizon 2020 research and innovation program.

\section{Preliminaries} \label{sec.prelim}

We begin our analysis with a precise description of the \emph{asymptotically anti-de Sitter} (abbreviated \emph{aAdS}) spacetimes we will consider in this article.
Much of this content is an abridged version of the more detailed development in \cite[Section 2]{Arick3}.

\subsection{Admissible spacetimes}

The first step is to construct our asymptotically AdS manifolds.
For this, we resort to the following assumption:

\begin{assumption} \label{ass.manifold}
Fix $n \geq 3$, let $\mc{I}$ be an $n$-dimensional manifold, and let
\begin{align}\label{manifold}
\mc{M} := ( 0, \rho_0 ] \times \mc{I} \text{,} \qquad \rho_0 > 0 \text{.}
\end{align}
In addition, we let $\rho$ denote the projection from $\mc{M}$ onto its $( 0, \rho_0 ]$-component.
\end{assumption}

Throughout, we will work with three types of tensorial objects on $\mc{M}$ and $\mc{I}$:
\begin{itemize}
\item \emph{Spacetime tensor fields:} These are simply tensor fields on $\mc{M}$.
In general, spacetime tensor fields will be denoted using standard italicized font (e.g., $g$ and $A$).

\item \emph{Boundary tensor fields:} These refer to tensor fields on $\mc{I}$.
In general, boundary tensor fields will be denoted using a fraktur font (e.g., $\mf{g}$ and $\mf{A}$).

\item \emph{Vertical tensor fields:} These are $\rho$-parametrized families of tensor fields on $\mc{I}$, which can be equivalently viewed as spacetime fields having only $\mc{I}$-components.
In general, vertical tensor fields are denoted using serif font (e.g., $\ms{g}$ and $\ms{A}$).
\end{itemize}
We will use these two viewpoints of vertical fields interchangeably, depending on context.

\begin{remark}[Identifications of fields]
Note that any boundary tensor field $\mf{A}$ can also be viewed as a vertical tensor field, by thinking of $\mf{A}$ as a $\rho$-independent family of tensor fields on $\mc{I}$.
Similarly, from the above discussion, we see that a vertical tensor field $\ms{A}$ can also be viewed as a spacetime tensor field on $\mc{M}$ that is trivial in any $\rho$-component.
\end{remark}

In addition, we let $\partial_\rho$ denote the lift to $\mc{M}$ of the vector field $\frac{d}{ d \rho }$ on $( 0, \rho_0 ]$, \footnote{In other words, the integral curves of $\partial_\rho$ are given by $\rho \mapsto ( \rho, x )$ for every $x \in \mi{M}$.} and we let $\mc{L}_\rho$ denote the Lie derivative in the $\partial_\rho$-direction.
Note in particular that for a vertical tensor field $\ms{A}$, its Lie derivative in $\rho$ is given by (see \cite[Section 2]{Arick3})
\begin{equation} \label{LieDeriv}
\mc{L}_\rho \ms{A} |_{ \rho = \sigma } = \lim_{ \sigma' \rightarrow \sigma } ( \sigma' - \sigma )^{-1} ( \ms{A} |_{ \rho = \sigma' } - \ms{A} |_{ \rho = \sigma } ) \text{.}
\end{equation}

Next, consider a local coordinate system $( U, \varphi )$ in $\mc{I}$:
\begin{itemize}
\item We say that $( U, \varphi )$ is \emph{compact} iff $\bar{U}$ is compact and $\varphi$ extends smoothly to $\bar{U}$.

\item We let $\varphi_\rho$ denote the coordinate system $( \rho, \varphi )$ on $( 0, \rho_0 ] \times U$.
\end{itemize}
In this paper, we will only consider coordinate systems of the above forms.
Moreover, similar to \cite{Arick3}, we will use lowercase Latin indices ($a, b, c, \dots$) to denote the coordinates on $\mc{I}$, and lowercase Greek indices ($\mu, \nu, \lambda, \dots$) to denote coordinates on $\mc{M}$.
We also impose Einstein summation convention, meaning that repeated indices denote summations.

We make use of such compact coordinate systems in order to make sense of asymptotic bounds and boundary limits of vertical tensor fields:

\begin{definition} \label{def.Limit}
Given a vertical tensor field $\ms{A}$ of rank $( k, l )$ and an integer $M \geq 0$:
\begin{itemize}
\item For a compact coordinate system $( U, \varphi )$, we set
\begin{equation} \label{Norm}
| \ms{A} |_{ M, \varphi } := \sum_{ m = 0 }^M \sum_{ \substack{ a_1, \dots, a_m \\ b_1, \dots, b_k \\ c_1, \dots, c_l } } | \partial^m_{ a_1 \dots a_m } \ms{A}^{ b_1 \dots b_k }_{ c_1 \dots c_l } | \text{,}
\end{equation}
where the right-hand side is expressed with respect to $\varphi$-coordinates.

\item We say that $\ms{A}$ is \emph{locally bounded in $C^M$} iff $| \ms{A} |_{ M, \varphi }$ is a bounded function for any compact coordinate system $( U, \varphi )$ on $\mc{I}$.
Furthermore, we will use the notation $\mc{O} ( \zeta )$ to refer to any vertical tensor field $\ms{B}$ such that $\zeta^{-1} \ms{B}$ is locally bounded in $C^0$.

\item Let $\mf{A}$ be a boundary tensor field of the same rank $( k, l )$.
We write $\ms{A} \rightarrow^M \mf{A}$ (i.e., \emph{$\ms{A}$ converges to $\mf{A}$ locally in $C^M$ as $\rho \searrow 0$}) iff for any compact coordinate system $( U, \varphi )$,
\begin{equation} \label{Limit}
\lim_{ \sigma \searrow 0 } \sup_{ \{ \sigma \} \times U } | \ms{A} - \mf{A} |_{ M, \varphi } = 0 \text{.}
\end{equation}
\end{itemize}
\end{definition}

With the above conventions set, we can now prescribe the geometry of our setting:

\begin{assumption} \label{ass.aads}
Let $g$ be a Lorentzian metric on $\mc{M}$, of the form
\begin{align}\label{metricdefiniton1}
	 g := \rho^{-2} ( d \rho^2 + \ms{g} ) \text{,}
\end{align}
where:
\begin{itemize}
\item $\ms{g}$ is a \emph{vertical metric}, that is, a vertical tensor field of rank $( 0, 2 )$ satisfying that $\ms{g}$ is a Lorentzian metric on every level set of $\rho$.

\item There exist a Lorentzian metric $\mf{g}$ and a rank $( 0, 2 )$ tensor field $\bar{\mf{g}}$ on $\mc{I}$ such that
\footnote{Note that $\bar{\mf{g}}$ needs not be a metric.}
\begin{equation}
\label{metricdefiniton2} \ms{g} \rightarrow^3 \mf{g} \text{,} \qquad \mc{L}_\rho \ms{g} \rightarrow^2 0 \text{,} \qquad \mc{L}_\rho^2 \ms{g} \rightarrow^1 2 \bar{\mf{g}} \text{,} \qquad \mc{L}_\rho^3 \ms{g} = \mc{O} (1) \text{.}
\end{equation}
\end{itemize}
\end{assumption}

From here on, we will remain within the setting described by Assumptions \ref{ass.manifold} and \ref{ass.aads}.

\begin{figure}[ht]
    \includegraphics[width=0.9\textwidth]{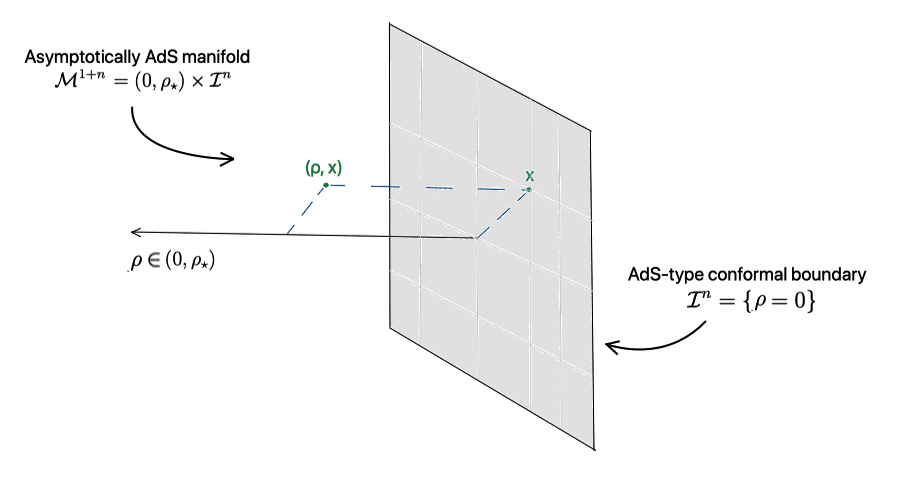}
    \caption{Illustration of an aAdS spacetime in Fefferman-Graham gauge.}
    \label{Pic1cite}
\end{figure}

Using the terminology of \cite{Arick3}, the spacetime $( \mc{M}, g )$ is called an \emph{FG-aAdS segment} and represents the near-boundary geometry of an aAdS spacetime.
Moreover, we refer to such a specific form \eqref{metricdefiniton1} of $g$ as a \emph{Fefferman-Graham} (or \emph{FG}) gauge.

\begin{remark}[Asymptotic expansion of $g$]
A more intuitive but slightly less precise formulation of \eqref{metricdefiniton1} and \eqref{metricdefiniton2} is as an asymptotic expansion for $\ms{g}$:
\begin{equation} \label{FGExpansion}
\ms{g} = \mf{g} + \rho^2 \bar{\mf{g}} + \mc{O} ( \rho^3 ) \text{.}
\end{equation}
\end{remark}

\begin{remark}
Note that no generality is lost by focusing on aAdS spacetimes given by \eqref{metricdefiniton1} and \eqref{metricdefiniton2}, as one can change coordinates such that the Fefferman-Graham gauge condition \eqref{metricdefiniton2} is satisfied.
For such an appropriate change of variables, we refer the reader to \cite{MR837196}. 
\end{remark}

The manifold $( \mc{I}, \mf{g} )$ is commonly known as the \emph{conformal boundary} of $( \mc{M}, g )$.
In particular, we can formally identify the conformal boundary $( \mc{I}, \mf{g} )$ with $\{ \rho = 0 \}$, and interpret the expansion \eqref{metricdefiniton2} as the asymptotic behaviour of $\mf{g}$ toward the conformal boundary.

Observe that our setting contains three distinct Lorentzian metrics:\ the spacetime metric $g$, the vertical metric $\ms{g}$, and the boundary metric $\mf{g}$.
Each metric comes with a number of associated operators and geometric quantities, such as the Levi-Civita connection and the curvature.
The following table summarizes the notations we will use for these objects.

\begin{table}[ht]
 \begin{tabular}[ht]{|c|c|c|c|c|c|c|c|c|} 
 \hline
 Type & Tensors & Metric & Dual & Trace & Christoffel & Connection & Gradient & Curvature \\
 \hline 
 Spacetime & $A, X, Y$ & $g$ & $g^{-1}$ & $\tr_g$ & $\Gamma$ & $\nabla$ & $\nabla^\sharp$ & $R$ \\ 
 Vertical & $\ms{A}, \ms{X}, \ms{Y}$ & $\ms{g}$ & $\ms{g}^{-1}$ & $\tr_{ \ms{g} }$ & $\ms{\Gamma}$ & $\ms{D}$ & $\ms{D}^\sharp$ & $\ms{R}$ \\
 Boundary & $\mf{A}, \mf{X}, \mf{Y}$ & $\mf{g}$ & $\mf{g}^{-1}$ & $\tr_{ \mf{g} }$ & $\mf{T}$ & $\mf{D}$ & $\mf{D}^\sharp$ & $\mf{R}$ \\
 \hline 
 \end{tabular}
\end{table}

Also, as is standard, we will omit the ``$-1$'' from metric duals when in index notation.

\begin{remark}
The more accurate statement is that $\ms{g}$ gives a Lorentzian metric on each level set of $\rho$.
Thus, the vertical connection $\ms{D}$ and curvature $\ms{R}$ are more precisely defined on each level set $\{ \rho = \sigma \}$ as the connection and curvature associated with $\ms{g} |_{ \rho = \sigma }$.
\end{remark}

The following lemmas list the metric and Christoffel symbol components in the FG gauge:

\begin{lemma}[Metric] \label{LemmaMetric}
Let $( U, \varphi )$ be any coordinate system in $\mc{I}$.
Then:
\begin{itemize}
\item The following hold with respect to $\varphi_\rho$-coordinates:
	\begin{align}
		\label{expansiong} g_{\rho \rho} = \rho^{-2} \text{,} \qquad g_{\rho b} &= 0 \text{,} \qquad g_{bc} = \rho^{-2} \ms{g}_{bc} \text{,} \\
		\notag g^{\rho \rho} = \rho^{2} \text{,} \qquad g^{\rho b} &= 0 \text{,} \qquad g^{bc} = \rho^{2} \ms{g}^{bc} \text{.}
	\end{align}

\item The following hold with respect to $\varphi$-coordinates:
\begin{align}
  \label{expansiongbc} \ms{g}_{bc} &= \mf{g}_{bc} + \rho^2 \bar{\mf{g}}_{bc} + \mc{O} (\rho^3)_{bc} \text{,} \qquad \ms{g}^{bc} = \mf{g}^{bc} - \rho^2 \mf{g}^{bd} \bar{\mf{g}}_{de} \mf{g}^{ec} + \mc{O} (\rho^3)^{bc} \text{.}
\end{align}
\end{itemize}
\end{lemma}

\begin{proof}
	These follow immediately from \eqref{metricdefiniton1} and \eqref{FGExpansion}.
\end{proof}

\begin{lemma}[Christoffel symbols] \label{LemmaChristoffel}
Let $( U, \varphi )$ denote any coordinate system in $\mc{I}$.
Then, the following identities hold with respect to $\varphi_\rho$-coordinates:
\begin{align}
  \label{expansionGamma} &\Gamma_{\rho \rho}^{\rho} = -\rho^{-1} \text{,} \qquad \Gamma_{\rho b}^{\rho} = 0 \text{,} \qquad \Gamma_{bc}^{\rho} = \rho^{-1} \ms{g}_{bc} - \frac{1}{2} \mc{L}_\rho \ms{g}_{bc} \text{,} \\
  \notag &\Gamma_{\rho \rho}^{b} = 0 \text{,} \qquad \Gamma_{\rho b}^{c} = - \rho^{-1} \delta_b^c + \frac{1}{2} \ms{g}^{ce} \mc{L}_\rho \ms{g}_{be} \text{,} \qquad \Gamma_{bc}^{d} = \ms{\Gamma}_{bc}^{d} \text{.}
\end{align}
\end{lemma}

\begin{proof}
  These follow from direct computations using \eqref{expansiong}.
\end{proof}

\subsection{Near-Boundary Foliations}

A crucial step in our unique continuation result is the construction of a pseudoconvex foliation of hypersurfaces near the conformal boundary $\mc{I}$ that also terminate at $\mc{I}$.
In this section, we first study a general class of foliating functions $f$ near $\mc{I}$, constructed by ``bending'' the level sets of $\rho$ toward $\mc{I}$.

More specifically, given an open $\mc{D} \subseteq \mc{I}$ and a smooth $\eta: \mc{D} \longrightarrow ( 0, \infty )$, we define
\begin{align} \label{Definitionf}
	f_\eta := f: ( 0, \rho_0 ] \times \mc{D} \longrightarrow \R \text{,} \qquad f ( \rho, x ) := \frac{ \rho }{ \eta(x) } \text{.}
\end{align}
Observe that wherever $\eta \searrow 0$, the level sets of $f$ terminate at the conformal boundary.
Figure \ref{Pic2cite} illustrates a generic foliation obtained from the level sets of $f$.

\begin{figure}[ht]
    \includegraphics[width=0.9\textwidth]{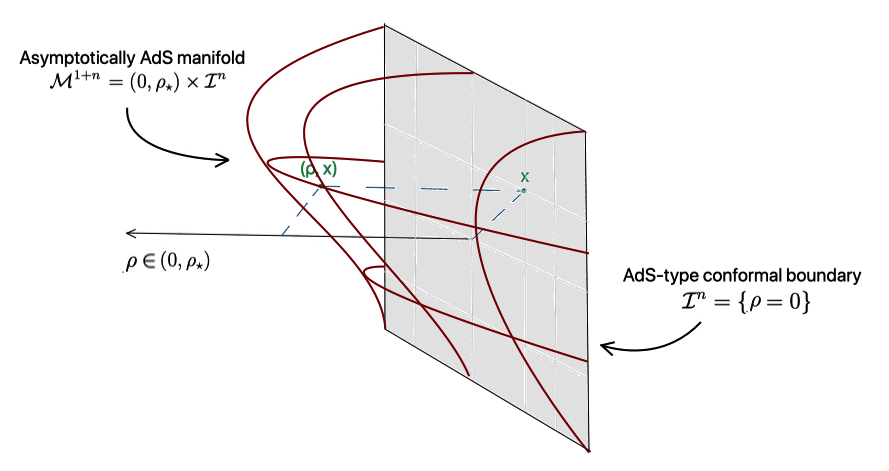}
    \caption{The red grid illustrates a level set of $f$. At the edges of this grid, this level set terminates at the conformal boundary.}
    \label{Pic2cite}
\end{figure}

In the following, we compute the derivatives, the gradient, and the Hessian of $f$:

\begin{lemma}[Derivatives of $f$] \label{LemmaDerivativesf}
Let $( U, \varphi )$ be an arbitrary coordinate system in $\mc{I}$.
Then, the following identities hold with respect to $\varphi_\rho$-coordinates:
\begin{align}
  \label{Derivativesf} \partial_\rho f = \rho^{-1} f \text{,} &\qquad \partial_b f = - \rho^{-1} f^2 \partial_b \eta \text{,} \\
  \notag \partial_{\rho \rho}^2 f = 0 \text{,} \qquad	\partial_{\rho b}^2 f = - \rho^{-2} f^{2} \partial_{b} \eta \text{,} &\qquad \partial_{bc}^2 f = - \rho^{-1} f^{2} \partial_{bc}^2 \eta + 2 \rho^{-2} f^{3} \partial_{b} \eta \partial_{c} \eta \text{.}
\end{align}
\end{lemma}

\begin{proof}
	These are immediate consequences of \eqref{Definitionf}.
\end{proof}

\begin{lemma}[Gradient of $f$] \label{LemmaGradientf}
The following identities hold:
\begin{align}
  \label{Gradientf}	\nabla^\sharp f = \rho f \partial_\rho - \rho f^2 \ms{D}^\sharp \eta \text{,} &\qquad g ( \nabla^\sharp f, \nabla^\sharp f ) = f^2 + f^4 \, \ms{g} ( \ms{D}^\sharp \eta, \ms{D}^\sharp \eta ) \text{.}
\end{align}
If $\eta$ is also locally bounded in $C^1$, then
\begin{align}
  \label{AsympGradientf} \nabla^\sharp f = \rho f \partial_{\rho} - \rho f^2 \mf{D}^\sharp \eta + \mc{O} ( \rho^3 f^2 ) \text{,} \qquad g ( \nabla^\sharp f, \nabla^\sharp f ) = f^2 + \mc{O} ( f^4 ) \text{.}
\end{align}
\end{lemma}

\begin{proof}
Note that by \eqref{Definitionf}, we have
\begin{align*}
\nabla^\sharp f = g^{ \mu \nu } \partial_\mu f \partial_\nu \text{,} \qquad g ( \nabla^\sharp f, \nabla^\sharp f ) = g^{ \mu \nu } \partial_\mu f \partial_\nu f \text{.}
\end{align*}
The identities \eqref{Gradientf} then follow from \eqref{expansiong}, Lemma \ref{LemmaDerivativesf}, and the above.
Furthermore, the asymptotics \eqref{AsympGradientf} follow from \eqref{Gradientf} and the metric expansions \eqref{expansiongbc}.
\end{proof}

\begin{lemma}[Hessian of $f$] \label{LemmaHessianf}
Let $( U, \varphi )$ be an arbitrary coordinate system in $\mc{I}$.
Then, the following identities hold with respect to $\varphi_\rho$-coordinates:
\begin{align}
\label{Hessianf} \nabla_{ \rho \rho }^2 f &= \rho^{-2} f \text{,} \\
\notag \nabla_{ \rho b }^2 f &= -2 \rho^{-2} f^2 \ms{D}_b \eta + \frac{1}{2} \rho^{-1} f^2 \mc{L}_\rho \ms{g}_{bd} ( \ms{D}^\sharp \eta )^d \text{,} \\
\notag \nabla_{ b c }^2 f &= - \rho^{-1} f^2 \ms{D}_{bc}^2 \eta - \rho^{-2} f \ms{g}_{bc} + \frac{1}{2} \rho^{-1} f \mc{L}_\rho \ms{g}_{bc} + 2 \rho^{-2} f^3 \ms{D}_b \eta \ms{D}_c \eta \text{.}
\end{align}
\end{lemma}

\begin{proof}
These results follow from Lemma \ref{LemmaChristoffel}, Lemma \ref{LemmaDerivativesf}, and the standard identity
\begin{align*}
\nabla_{\mu \nu}^2 f = \partial_{\mu \nu}^2 f - \Gamma_{\mu \nu}^{\lambda} \partial_{\lambda} f \text{.}
\end{align*}
Indeed, from the above, we compute the following:
\begin{align*}
\nabla_{\rho \rho}^2 f &= \partial _{\rho \rho}^2 f - \Gamma_{\rho\rho}^{\rho} \partial _{\rho} f - \Gamma_{\rho\rho}^{b} \partial_{b} f \\
&= \rho^{-2} f \text{,} \\
\nabla_{\rho b}^2 f &= \partial _{\rho b}^2 f - \Gamma_{\rho b}^{\rho} \partial _{\rho} f - \Gamma_{\rho b}^{c} \partial_{c} f \\
&= - \rho^{-2} f^2 \partial_{b} \eta - \left( \frac{1}{2} \ms{g}^{cd} \mc{L}_\rho \ms{g}_{bd} - \rho^{-1} \delta_{b}^{c} \right) ( - \rho^{-1} f^2 \partial_c \eta ) \\
&= - 2 \rho^{-2} f^2 \ms{D}_b \eta + \frac{1}{2} \rho^{-1} f^2 \mc{L}_\rho \ms{g}_{bd} ( \ms{D}^\sharp \eta )^d \text{,} \\
\nabla_{bc}^2 f &= \partial_{bc}^2 f - \Gamma_{bc}^{\rho} \partial_\rho f - \Gamma_{bc}^{d} \partial_d f \\
&= -\rho^{-1} f^2 \partial_{bc}^2 \eta + 2 \rho^{-2} f^3 \partial_b \eta \partial_c \eta + \frac{1}{2} \rho^{-1} f \mc{L}_\rho \ms{g}_{bc} - \rho^{-2} f \ms{g}_{bc}  + \rho^{-1} f^2 \ms{\Gamma}_{bc}^{d} \partial_d \eta \\
&= -\rho^{-1} f^2 \ms{D}^2 \eta + 2 \rho^{-2} f^3 \ms{D}_b \eta \ms{D}_c \eta + \frac{1}{2} \rho^{-1} f \mc{L}_\rho \ms{g}_{bc} - \rho^{-2} f \ms{g}_{bc} \text{.} \qedhere
\end{align*}
\end{proof}

\begin{lemma}[Wave operator of $f$] \label{LemmaWaveOperatorf}
The following holds for $\Box_g f := \tr_g ( \nabla^2 f )$:
\begin{align} \label{WaveOperatorf}
\Box_g f = - (n-1) f - \rho f^2 \tr_{ \ms{g} } ( \ms{D}^2 \eta ) + \frac{1}{2} \rho f \tr_{ \ms{g} } \mc{L}_\rho \ms{g}	+ 2 f^3 \, \ms{g} ( \ms{D}^\sharp \eta, \ms{D}^\sharp \eta ) \text{.}
\end{align}
If $\eta$ is also locally bounded in $C^2$, then
\begin{align} \label{AsympWaveOperatorf}
\Box_g f &= - (n-1) f + \mc{O} ( f^3 ) \text{.}
\end{align}
\end{lemma}

\begin{proof}
First, \eqref{WaveOperatorf} follows from a direct computation using Lemmas \ref{LemmaMetric} and \ref{LemmaHessianf}:
\begin{align*}
\Box_g f &= g^{\rho \rho} \nabla_{\rho \rho}^2 f + g^{b c} \nabla_{bc}^2 f \\
&= \rho^{2} ( \rho^{-2} f ) + \rho^2 \ms{g}^{bc} \left(	- \rho^{-1} f^2 \ms{D}^2_{bc} \eta - \rho^{-2} f \ms{g}_{bc} + \frac{1}{2} \rho^{-1} f \mc{L}_\rho \ms{g}_{bc} + 2 \rho^{-2} f^3 \ms{D}_b \eta \ms{D}_c \eta \right) \\
&= - (n-1) f - \rho f^2 \tr_{ \ms{g} } ( \ms{D}^2 \eta ) + \frac{1}{2} \rho f \tr_{ \ms{g} } \mc{L}_\rho \ms{g}	+ 2 f^3 \, \ms{g} ( \ms{D}^\sharp \eta, \ms{D}^\sharp \eta ) \text{.}
\end{align*}
For \eqref{AsympWaveOperatorf}, we apply Lemma \ref{LemmaWaveOperatorf}, along with the asymptotics
\begin{align*}
	\tr_{ \ms{g} } ( \ms{D}^2 \eta ) = \mc{O} (1) \text{,} \qquad \tr_{ \ms{g} } \mc{L}_\rho \ms{g} = \mc{O} (\rho) \text{,} \qquad \ms{g} ( \ms{D}^\sharp \eta, \ms{D}^\sharp \eta ) = \mc{O} (1) \text{,}
\end{align*}
which follow from Assumption \ref{ass.aads}, Lemmas \ref{LemmaMetric} and \ref{LemmaChristoffel}, and our assumptions on $\eta$.
\end{proof}

Next, given any $c > 0$, we denote the corresponding level set of $f$ by
\begin{align} \label{levetsetsSigma}
	\Sigma_c := \{ f = c \} =	\{ ( \sigma, x ) \in ( 0, \rho ] \times \mc{D} : \sigma = c \eta(x) \} \text{.}
\end{align}
Observe that $\nabla^\sharp f$ is an inward-pointing normal to each $\Sigma_c$.
\footnote{By inward-pointing, we mean that $\nabla^\sharp f$ points from the boundary into the bulk, that is, $( \nabla^\sharp f ) \rho > 0$.}

\begin{definition}[Unit normal] \label{DefinitionNormal}
Whenever $\nabla^\sharp f$ is spacelike, we define the vector field
\footnote{Note the second equality in \eqref{Normal} is a consequence of Lemma \ref{LemmaGradientf}.}
\begin{equation} \label{Normal}
N := \frac{ \nabla^\sharp f }{ | g ( \nabla^\sharp f, \nabla^\sharp f ) |^\frac{1}{2} } = \rho [ 1 + f^2 \, g ( \ms{D}^\sharp \eta, \ms{D}^\sharp \eta ) ]^{ -\frac{1}{2} } ( \partial_\rho - f \ms{D}^\sharp \eta ) \text{,}
\end{equation}
\end{definition}

In addition, since $\Sigma_c$ can be viewed as the graph of $c \eta$, then one can naturally identify $T \mc{D}$ with $T \Sigma_c$.
This can be done for each $c > 0$ via the following formula:

\begin{definition}[Boundary-foliation isomorphism] \label{DefinitionVa}
Given a vector field $\mf{X}$ on some open $\mc{U} \subseteq \mc{D}$, we define $\mc{P} \mf{X}$ to be the vector field on $T ( ( 0, \rho_0 ] \times \mc{U} )$ given by
\footnote{The factor $\rho$ on the right-hand side of \eqref{Va} arises from the scale factor $\rho^{-2}$ associated with $g$.}
\begin{equation} \label{Va}
\mc{P} \mf{X} = \rho [ f ( \mathfrak{X} \eta ) \, \partial_\rho + \mf{X} ] \text{.}
\end{equation}
\end{definition}

\begin{lemma}[Properties of $\mc{P}$] \label{LemmaVa}
Let $\mf{X}$, $\mf{Y}$ be vector fields on (a subset of) $\mc{D}$, and let
\[
X := \mc{P} \mf{X} \text{,} \qquad Y := \mc{P} \mf{Y} \text{.}
\]
\begin{itemize}
\item $\mc{P}$ is linear and invertible.

\item $X$ and $Y$ are tangent to each level set of $f$.

\item The following identity holds:
\begin{align} \label{VaVb}
g ( X, Y ) &= f^2 ( \mf{X} \eta ) ( \mf{Y} \eta ) + \ms{g} ( \mf{X}, \mf{Y} ) \text{.}
\end{align}

\item If $\eta$ is also locally bounded in $C^1$, then
\begin{align} \label{AsympVaVb}
g ( X, Y ) &= \mf{g} ( \mf{X}, \mf{Y} ) + \mc{O} ( f^2 ) ( \mf{X}, \mf{Y} ) \text{.}
\end{align}
\end{itemize}
\end{lemma}

\begin{proof}
That $\mc{P}$ is linear and invertible follows from the formula \eqref{Va}.
That $X$ and $Y$ are tangent to level sets of $f$ is a direct computation using \eqref{Va}; we show this for $X$:
\begin{align*}
X f &= \rho [ f ( \mf{X} \eta ) \partial_\rho f + \mf{X} f ] \\
&= \rho [ \rho \eta^{-2} ( \mf{X} \eta ) + \rho \, \mf{X} ( \eta^{-1} ) ] \\
&= 0 \text{.}
\end{align*}

Finally, the identity \eqref{VaVb} is an immediate consequence of Lemma \ref{LemmaMetric} and \eqref{Va}, while the asymptotics \eqref{AsympVaVb} follows from \eqref{expansiongbc} and \eqref{VaVb}.
\end{proof}

\subsection{Pseudoconvexity}

In this section, we examine when the level sets $\Sigma_c$ of $f$ are pseudoconvex.
More specifically, using the tools from above, we connect the pseudoconvexity of these level sets to a partial differential inequality on the conformal boundary.

First, we recall the definition of pseudoconvexity from \cite{ler_robb:unique}; see also \cite[Section 2.3]{Arick1}:

\begin{definition}[Pseudoconvexity]
Let $f$ be as before.
Then, the level set $\Sigma_c := \{ f = c \}$ is pseudoconvex (with respect to the wave operator $\Box_g$ and the direction of increasing $f$) iff $-f$ is convex along all $g$-null directions tangent to $\Sigma_c$, namely,
\begin{align} \label{DefintionPseudoConvexity}
- \nabla^2 f ( U, U ) > 0 \text{,} \quad \text{ for all } U \in T \Sigma_c \text{, with } g ( U, U ) = 0 \text{.}
\end{align}
\end{definition}

However, rather than examine \eqref{DefintionPseudoConvexity}, it suffices to instead study the following condition: there exists a smooth function $w: ( 0, \rho_0 ] \times \mc{D} \rightarrow \R$ such that
\begin{align} \label{fconvex}
- ( \nabla^2 f + w \, g ) ( X, X ) > 0 \text{,} \quad \text{ for all } X \in T \Sigma_c \text{.}
\end{align}
Clearly, \eqref{fconvex} implies \eqref{DefintionPseudoConvexity}, though the two conditions are in fact equivalent.
\footnote{See \cite{Arick3}.}
In other words, we can view \eqref{fconvex} as the criterion for $\Sigma_c$ being pseudoconvex.

Thus, our goal is now to obtain conditions on which \eqref{conditiononpi} holds.
For this, we must find the Hessian of $f$ with respect to the vector fields tangent to level sets of $f$:

\begin{lemma}[Hessian of $f$, revisited] \label{LemmaHessian}
Let $\mf{X}$, $\mf{Y}$ be vector fields on (a subset of) $\mc{D}$, and let $X := \mc{P} \mf{X}$ and $Y := \mc{P} \mf{Y}$ (see Definition \ref{DefinitionVa}).
Then, the following identities hold:
\begin{align}
\label{Hessian} \nabla^2 f ( X, Y ) &= - f \, g ( X, Y ) - \rho f^2 \, \ms{D}^2 \eta ( \mf{X}, \mf{Y} ) + \frac{1}{2} \rho f \, \mc{L}_\rho \ms{g} ( \mf{X}, \mf{Y} ) \\
\notag &\qquad + \frac{1}{2} \rho f^3 [ \mf{Y} \eta \, \mc{L}_\rho \ms{g} ( \ms{D}^\sharp \eta, \mf{X} ) + \mf{X} \eta \, \mc{L}_\rho \ms{g} ( \ms{D}^\sharp \eta, \mf{Y} ) ] \text{,} \\
\notag \nabla^2 f ( \nabla^\sharp f, X ) &= \rho f^4 \, \ms{D}^2 \eta ( \ms{D}^\sharp \eta, \mf{X} ) - \frac{1}{2} \rho f^5 \mf{X} \eta \, \mc{L}_\rho \ms{g} ( \ms{D}^\sharp \eta, \ms{D}^\sharp \eta ) \text{,} \\
\notag \nabla^2 f ( \nabla^\sharp f, \nabla^\sharp f ) &= f \, g ( \nabla^\sharp f, \nabla^\sharp f ) + 2 f^5 \, \ms{g} ( \ms{D}^\sharp \eta, \ms{D}^\sharp \eta ) - \frac{1}{2} \rho f^5 \, \mc{L}_\rho \ms{g} ( \ms{D}^\sharp \eta, \ms{D}^\sharp \eta ) \\
\notag &\qquad - \rho f^6 \, \ms{D}^2 \eta ( \ms{D}^\sharp \eta, \ms{D}^\sharp \eta ) + 2 f^7 [ \ms{g} ( \ms{D}^\sharp \eta, \ms{D}^\sharp \eta ) ]^2 \text{.}
\end{align}
Moreover, if $\eta$ is locally bounded in $C^2$, then
\begin{align} \label{AsympHessian}
( \nabla^2 f + f g ) ( X, Y ) &= - \rho f^2 \, ( \mf{D}^2 \eta - \eta \bar{\mf{g}} ) ( \mf{X}, \mf{Y} ) + \mc{O} ( \rho f^3 ) ( \mf{X}, \mf{Y} ) \text{,} \\
\notag ( \nabla^2 f + f g ) ( \nabla^\sharp f, X ) &= \mc{O} ( \rho f^4 ) ( \mf{X} ) \text{,} \\
\notag ( \nabla^2 f + f g ) ( \nabla^\sharp f, \nabla^\sharp f ) &= 2 f^3 + \mc{O} ( f^5 ) \text{.}
\end{align}
\end{lemma}

\begin{proof}
First, observe that given $X$, $Y$ as above, we have from \eqref{Hessianf} that
\begin{align*}
\nabla^2 f ( X, Y ) &= X^\rho Y^\rho \, \nabla_{ \rho \rho } f + ( X^\rho Y^b + X^b Y^\rho ) \nabla_{ \rho b } f + X^b Y^c \, \nabla_{ b c } f \\
&= X^\rho Y^\rho \, \rho^{-2} f + ( X^\rho Y^b + X^b Y^\rho ) \left[ -2 \rho^{-2} f^2 \ms{D}_b \eta +  \frac{1}{2} \rho^{-1} f^2 \mc{L}_\rho \ms{g}_{bd} ( \ms{D}^\sharp \eta )^d \right] \\
&\qquad + X^b Y^c \left[ - \rho^{-1} f^2 \ms{D}_{bc}^2 \eta - \rho^{-2} f \ms{g}_{bc} + \frac{1}{2} \rho^{-1} f \mc{L}_\rho \ms{g}_{bc} + 2 \rho^{-2} f^3 \ms{D}_b \eta \ms{D}_c \eta \right] \text{.}
\end{align*}
Second, the values $X^\mu = ( \mc{P} \mf{X} )^\mu$ and $Y^\mu = ( \mc{P} \mf{Y} )^\mu$ can be expanded using \eqref{Va}:
\begin{align*}
\nabla^2 f ( X, Y ) &= \mf{X} \eta \mf{Y} \eta \, f^3 + ( \mf{X} \eta \, \mf{Y}^b + \mf{Y} \eta \mf{X}^b ) \left[ -2 f^3 \ms{D}_b \eta + \frac{1}{2} \rho f^3 \mc{L}_\rho \ms{g}_{bd} ( \ms{D}^\sharp \eta )^d \right] \\
&\qquad + \mf{X}^b \mf{Y}^c \left[ - \rho f^2 \ms{D}_{bc}^2 \eta - f \ms{g}_{bc} + \frac{1}{2} \rho f \mc{L}_\rho \ms{g}_{bc} + 2 f^3 \ms{D}_b \eta \ms{D}_c \eta \right] \\
&= - f \, \ms{g} ( \mf{X}, \mf{Y} ) - f^3 \, \mf{X} \eta \mf{Y} \eta + \frac{1}{2} \rho f^3 ( \mf{X} \eta \, \mf{Y}^b + \mf{Y} \eta \mf{X}^b ) \mc{L}_\rho \ms{g}_{bd} ( \ms{D}^\sharp \eta )^d \\
&\qquad + \mf{X}^b \mf{Y}^c \left( - \rho f^2 \ms{D}_{bc}^2 \eta + \frac{1}{2} \rho f \mc{L}_\rho \ms{g}_{bc} \right) \text{.} 
\end{align*}
The first part of \eqref{Hessian} now follows from the above and from \eqref{VaVb}.

Next, applying \eqref{Gradientf} along with \eqref{Hessianf} and \eqref{Va} yields
\begin{align*}
\nabla^2 f ( \nabla^\sharp f, X ) &= \rho^2 f^2 \mf{X} \eta \, \nabla_{ \rho \rho } f + [ \rho^2 f \, \mf{X}^b - \rho^2 f^3 ( \ms{D}^\sharp \eta )^b \, \mf{X} \eta ] \nabla_{ \rho b } f - \rho^2 f^2 ( \ms{D}^\sharp \eta )^b \mf{X}^c \, \nabla_{ b c } f \\
&= f^3 \, \mf{X} \eta + [ \mf{X}^b - f^2 ( \ms{D}^\sharp \eta )^b \, \mf{X} \eta ] \left[ -2 f^3 \ms{D}_b \eta + \frac{1}{2} \rho f^3 \mc{L}_\rho \ms{g}_{bd} ( \ms{D}^\sharp \eta )^d \right] \\
&\qquad + ( \ms{D}^\sharp \eta )^b \mf{X}^c \left[ \rho f^4 \ms{D}_{bc}^2 \eta + f^3 \ms{g}_{bc} - \frac{1}{2} \rho f^3 \mc{L}_\rho \ms{g}_{bc} - 2 f^5 \ms{D}_b \eta \ms{D}_c \eta \right] \text{.}
\end{align*}
Further expanding the right-hand side of the above and then observing that several pairs of terms cancel results in the second identity in \eqref{Hessian}.

Similarly, for the last part of \eqref{Hessian}, we expand
\begin{align*}
\nabla^2 f ( \nabla^\sharp f, \nabla^\sharp f ) &= \rho^2 f^2 \, \nabla_{ \rho \rho } f - 2 \rho^2 f^3 ( \ms{D}^\sharp \eta )^b \, \nabla_{ \rho b } f + \rho^2 f^4 ( \ms{D}^\sharp \eta )^b ( \ms{D}^\sharp \eta )^c \, \nabla_{ b c } f \\
&= f^3 + ( \ms{D}^\sharp \eta )^b [ 4 f^5 \ms{D}_b \eta - \rho f^5 \mc{L}_\rho \ms{g}_{bd} ( \ms{D}^\sharp \eta )^d ] \\
&\qquad + ( \ms{D}^\sharp \eta )^b ( \ms{D}^\sharp \eta )^c \left[ - \rho f^6 \ms{D}_{bc}^2 \eta - f^5 \ms{g}_{bc} + \frac{1}{2} \rho f^5 \mc{L}_\rho \ms{g}_{bc} + 2 f^7 \ms{D}_b \eta \ms{D}_c \eta \right] \\
&= f^3 + 3 f^5 \, \ms{g} ( \ms{D}^\sharp \eta, \ms{D}^\sharp \eta ) - \frac{1}{2} \rho f^5 \, \mc{L}_\rho \ms{g} ( \ms{D}^\sharp \eta, \ms{D}^\sharp \eta ) - \rho f^6 \, \ms{D}^2 \eta ( \ms{D}^\sharp \eta, \ms{D}^\sharp \eta ) \\
&\qquad - \rho f^6 \, \ms{D}^2 \eta ( \ms{D}^\sharp \eta, \ms{D}^\sharp \eta ) + 2 f^7 [ \ms{g} ( \ms{D}^\sharp \eta, \ms{D}^\sharp \eta ) ]^2 \text{.}
\end{align*}
The desired identity now follows from the above, after applying the second part of \eqref{Gradientf} to the first two terms on the right-hand side.
This completes the proof of \eqref{Hessian}.

Finally, \eqref{AsympHessian} follow from combining \eqref{Hessian} with the following asymptotics:
\begin{align*}
\ms{D}^2 \eta ( \mf{X}, \mf{Y} ) &= \mf{D}^2 \eta ( \mf{X}, \mf{Y} ) + \mc{O} ( \rho^2 ) ( \mf{X}, \mf{Y} ) \text{,} \\
\mc{L}_\rho \ms{g} ( \mf{X}, \mf{Y} ) &= 2 \rho \, \bar{\mf{g}} ( \mf{X}, \mf{Y} ) + \mc{O} ( \rho^2 ) ( \mf{X}, \mf{Y} ) \\
&= 2 f \eta \, \bar{\mf{g}} ( \mf{X}, \mf{Y} ) + \mc{O} ( \rho^2 ) ( \mf{X}, \mf{Y} ) \text{,} \\
\mc{L}_\rho \ms{g} ( \ms{D}^\sharp \eta, \mf{X} ) &= \mc{O} (\rho) ( \mf{X} ) \text{,} \\
\ms{D}^2 \eta ( \ms{D}^\sharp \eta, \mf{X} ) &= \mc{O} ( 1 ) ( \mf{X} ) \text{,} \\
\mc{L}_\rho \ms{g} ( \ms{D}^\sharp \eta, \ms{D}^\sharp \eta ) &= \mc{O} (\rho) \text{,} \\
\ms{D}^2 \eta ( \ms{D}^\sharp \eta, \ms{D}^\sharp \eta ) &= \mc{O} ( 1 ) \text{,} \\
\ms{g} ( \ms{D}^\sharp \eta, \ms{D}^\sharp \eta ) &= \mc{O} (1) \text{.} 
\end{align*}
These in turn follow from the various asymptotic expansions derived in \eqref{metricdefiniton2}, \eqref{expansiongbc}, \eqref{expansionGamma}, and \eqref{AsympGradientf}, along with the assumption that $\eta$ is locally bounded in $C^2$.
\end{proof}

In particular, the first identity of \eqref{AsympHessian} shows the leading order behavior of $\nabla^2 f + f \, g$---and hence the pseudoconvexity of $\Sigma_c$ for small $c$---is determined by $\eta$ and geometric properties of the conformal boundary.
Similar to \cite{Arick1, Arick2, Arick3}, for our Carleman estimates, it will be more convenient to express this in terms of a related modified deformation tensor:

\begin{definition}[Modified deformation tensor] \label{Defnmodifieddeformationtensor}
Given a function $\zeta: ( 0, \rho_0 ] \times \mc{D} \rightarrow \R$, we define the following modified deformation tensor $\pi_\zeta$ by
\begin{align} \label{modifieddeformationtensor}
\pi_\zeta := - [ \nabla ( f^{n-3} \nabla f ) + f^{n-3} w_\zeta \cdot g ] \text{,} \qquad w_\zeta := f + \rho f^2 \zeta \text{.}
\end{align}
\end{definition}

Observe that \eqref{fconvex} can be equivalently expressed in terms of $\pi$ as
\begin{align} \label{conditiononpi}
\pi_\zeta ( X, X ) > 0 \text{,} \quad \text{ for all } X \in T \Sigma_c \text{,}
\end{align}
for some particular chosen function $\zeta$, and with $w := w_\zeta$ in \eqref{fconvex}.
Lastly, we reformulate \eqref{AsympHessian} in terms of the modified deformation tensor \eqref{modifieddeformationtensor}:

\begin{lemma}[Asymptotic pseudoconvexity] \label{Lemmapi}
Assume the setting of Lemma \ref{LemmaHessian}, and let $\zeta$ and $\pi_\zeta$ be as in Definition \ref{Defnmodifieddeformationtensor}.
Moreover, suppose $\eta$ and $\zeta$ are locally bounded in $C^2$ and $C^0$, respectively.
Then, the following hold wherever $\nabla^\sharp f$ is spacelike:
\begin{align}
\label{pi} \pi_\zeta ( X, Y ) &=  \rho f^{n-1} ( \mf{D}^2 \eta - \eta \bar{\mf{g}} - \zeta \mf{g} ) ( \mf{X}, \mf{Y} ) + \mc{O} ( \rho^2 f^{n-1} ) ( \mf{X}, \mf{Y} ) \text{,} \\
\notag \pi_\zeta ( N, X ) &= \mc{O} ( \rho f^n ) ( \mf{X} ) \text{,} \\
\notag \pi_\zeta ( N, N ) &= - (n-1) f^{n-2} + \mc{O} ( f^n ) \text{.}
\end{align}
\end{lemma}

\begin{proof}
First, notice that \eqref{modifieddeformationtensor}, along with the fact that $g ( N, X ) = 0$, implies
\begin{align*}
\pi_\zeta ( X, Y ) &= - f^{ n - 3 } [ ( \nabla^2 f + f \, g ) ( X, Y ) + \rho f^2 \zeta \, g ( X, Y ) ] \text{,} \\
\pi_\zeta ( N, X ) &= - f^{ n - 3 } \, ( \nabla^2 f + f \, g ) ( N, X ) \text{.}
\end{align*}
Applying \eqref{AsympGradientf}, \eqref{Normal}, \eqref{AsympVaVb}, and \eqref{AsympHessian} to the above, we then obtain
\begin{align*}
\pi_\zeta ( X, Y ) &= f^{n-3} [ \rho f^2 ( \mf{D}^2 \eta - \eta \bar{\mf{g}} ) + \mc{O} ( \rho f^3 ) + \rho f^2 \zeta  \, \mf{g} + \mc{O} ( \rho^3 f^2 ) ] ( \mf{X}, \mf{Y} ) \\
&= \rho f^{n-1} ( \mf{D}^2 \eta - \eta \bar{\mf{g}} - \zeta \mf{g} ) ( \mf{X}, \mf{Y} ) + \mathcal{O} ( \rho f^n ) ( \mf{X}, \mf{Y} ) \text{,} \\
\pi_\zeta ( N, X ) &= - f^{n-3} ( \nabla^2 f + f \, g ) ( \nabla^\sharp f, X ) \cdot [ g ( \nabla^\sharp f, \nabla^\sharp f ) ]^{ - \frac{1}{2} } \\
&= \mc{O} ( \rho f^n ) \text{.}
\end{align*}
For the remaining part, we again apply \eqref{AsympGradientf}, \eqref{Normal}, and \eqref{AsympHessian} to conclude that
\begin{align*}
\pi_\zeta ( N, N ) &= - f^{n-3} ( \nabla^2 f + f \, g ) ( \nabla^\sharp f, \nabla^\sharp f ) \cdot [ g ( \nabla^\sharp f, \nabla^\sharp f ) ]^{-1} \\
&\qquad - (n-3) f^{n-4} \, g ( \nabla^\sharp f, \nabla^\sharp f ) - \rho f^{n-1} \zeta \\
&= [ - 2 f^{n-2} + \mc{O} ( f^n ) ] + [ - (n - 3) f^{n-2} + \mc{O} ( f^n ) ] + \mc{O} ( f^n ) \\
&= - (n - 1) f^{n-2} + \mc{O} ( f^n ) \text{.} \qedhere
\end{align*}
\end{proof}

\section{The Generalized Null Convexity Criterion} \label{sec.gncc}

Lemma \ref{Lemmapi}, in particular the first identity in \eqref{pi}, implies that in order for the level sets of $f$ to be pseudoconvex (namely, the condition \eqref{conditiononpi} holds), one requires positivity for the leading-order coefficient $\mf{D}^2 \eta - \eta \bar{\mf{g}} - \zeta \mf{g}$ at the conformal boundary.
This observation motivates our key \emph{generalized null convexity criterion} (abbreviated \emph{GNCC}):

\begin{definition}[Generalized Null Convexity Criterion] \label{DefAdmissibleDomains1}
Let $\mf{p}$ be a smooth Riemannian metric on $\mc{I}$, and let $\mc{D} \subseteq \mc{I} $ be open with compact closure.
We say $\mc{D}$ satisfies the generalized null convexity criterion iff there exists $\eta \in C^4 ( \bar{\mc{D}} )$ satisfying
\footnote{The assumption $\eta \in C^4 ( \bar{\mc{D}} )$ arises from the fact that one must take four derivatives of $\eta$ at one point in the proof of the upcoming Carleman estimates.}
\begin{align} \label{BVP}
\begin{cases}
( \mf{D}^2 \eta - \eta \, \bar{\mf{g}} ) ( \mf{Z}, \mf{Z} ) > c \eta \, \mf{p} ( \mf{Z}, \mf{Z} ) &\text{in } \mc{D} \text{,} \\
\eta > 0 &\text{in } \mc{D} \text{,} \\
\eta = 0 &\text{on } \partial \mc{D} \text{,}
\end{cases} 
\end{align}
for some constant $c > 0$, and for all tangent vectors $\mf{Z} \in T \mc{D}$ with $\mf{g} ( \mf{Z}, \mf{Z} ) = 0$.
\end{definition}

Applying \cite[Corollary 3.5]{Arick3}, one obtains a useful equivalent formulation of the GNCC:

\begin{prop}[Generalized Pseudoconvexity Criterion] \label{DefAdmissibleDomains0}
Let $\mf{p}$ and $\mc{D}$ be as in Definition \ref{DefAdmissibleDomains1}.
Then, $\mc{D}$ satisfies the GNCC if and only if there exists $\eta \in C^4 ( \bar{\mc{D}} )$ satisfying 
\begin{align} \label{BVP0}
\begin{cases}
( \mf{D}^2 \eta - \eta \, \bar{\mf{g}} - \zeta \, \mf{g} ) ( \mf{X}, \mf{X} ) > c \eta \, \mf{p} ( \mf{X}, \mf{X} ) &\text{in } \mc{D} \text{,} \\
\eta > 0 &\text{in } \mc{D} \text{,} \\
\eta = 0 &\text{on } \partial \mc{D} \text{,}
\end{cases} 
\end{align}	  
for some $c > 0$ and $\zeta \in C^2 ( \bar{\mc{D}} )$, and for all $\mf{X} \in T \mc{D}$.
\end{prop}

\begin{remark}
Both Definition \ref{DefAdmissibleDomains1} and Proposition \ref{DefAdmissibleDomains0} are independent of the choice of metric $\mf{p}$.
Since $\bar{\mc{D}}$ is compact, a change in $\mf{p}$ would only result in a different choice of constant $c$.
\end{remark}

In this section, we further investigate the equivalent conditions \eqref{BVP} and \eqref{BVP0}:
\begin{itemize}
\item To begin with, we show that the GNCC is gauge invariant (Proposition \ref{Gaugeinvariance}).

\item We then show that the GNCC implies the null convexity criterion of \cite{Arick3}.
\end{itemize}
Finally, to make the discussion more concrete, we look at some examples of domains where the GNCC holds, and where the GNCC fails to hold.

\subsection{Gauge Invariance}

In this subsection, we demonstrate that the GNCC \eqref{BVP} is gauge-invariant.
The first step in this process is to make precise the notion of gauge transformations.

Let $( \mc{M}, g )$ be as usual, and assume that $g$ can be written as
\footnote{The more accurate statements would be \eqref{metricdefiniton2} and their analogues $\ms{h} \rightarrow^3 \mf{h}$, $\mc{L}_\sigma \ms{h} \rightarrow^2 0$, $\mc{L}_\sigma^2 \ms{h} \rightarrow^1 2 \bar{\mf{h}}$, and $\mc{L}_\sigma^3 \ms{h} = \mc{O} ( 1 )$ for $( \sigma, \ms{h} )$.
However, we write this merely as expansions here to keep notations simpler.}
\begin{align}
\label{gauges} g = \rho^{-2} ( d\rho^2 + \ms{g} ) \text{,} &\qquad \ms{g} = \mf{g} + \rho^2 \bar{\mf{g}} + \mc{O} ( \rho^3 ) \text{,} \\
\notag g = \sigma^{-2} ( d\sigma^2 + \ms{h} ) \text{,} &\qquad \ms{h} = \mf{h} + \sigma^2 \bar{\mf{h}} + \mc{O} ( \sigma^3 ) \text{,}
\end{align}
where $\rho, \sigma \in C^\infty ( \mc{M} )$ both determine foliations of $\mc{M}$ and vanish at the conformal boundary, and where $\ms{g}$ and $\ms{h}$ are vertical metrics on level sets of $\rho$ and $\sigma$, respectively.
In other words, $( \mc{M}, g )$ is expressed in Fefferman-Graham gauge in two different ways, using $( \rho, \ms{g} )$ and $( \sigma, \ms{h} )$.

In addition, we assume some regularity between $\rho$ and $\sigma$ at the conformal boundary:
\footnote{More precisely, $\rho \rightarrow^4 0$, $\mc{L}_\sigma \rho \rightarrow^3 \mf{r}_1$, $\mc{L}_\sigma^2 \rho \rightarrow^2 2 \mf{r}_2$, $\mc{L}_\sigma^3 \rho \rightarrow^1 6 \mf{r}_3$, and $\mc{L}_\sigma^4 \rho = \mc{O} (1)$.}
\begin{align} \label{rhosigma}
\rho = \mf{r}_1 \sigma + \mf{r}_2 \sigma^2 + \mf{r}_3 \sigma^3 + \mc{O} ( \sigma^4 ) \text{,} \qquad \mf{r}_1 > 0 \text{.}
\end{align}
Here, the coefficients $\mf{r}_1$, $\mf{r}_2$, $\mf{r}_3$ are real-valued functions on $\mc{I}$.
Under these conditions, we can derive relations between the coefficients of $\ms{g}$ and $\ms{h}$:

\begin{prop}[Gauge transformations] \label{PropGaugeTransform}
Let $( \mc{M}, g)$ be as above, and assume the two FG-gauges \eqref{gauges} and the comparison \eqref{rhosigma}.
Then, we have the identities
\begin{equation} \label{GaugeTransform}
\mf{h} = \mf{a}^2 \mf{g} \text{,} \qquad \bar{\mf{h}} = \bar{\mf{g}} + \mf{a}^{-1} \, \mf{D}^2 \mf{a} - 2 \mf{a}^{-2} ( \mf{D} \mf{a} \otimes \mf{D} \mf{a} ) + \frac{1}{2} \mf{a}^{-2} \, \mf{g} ( \mf{D}^\sharp \mf{a}, \mf{D}^\sharp \mf{a} ) \cdot \mf{g} \text{,}
\end{equation}
where $\mf{a} := \mf{r}_1^{-1}$, and where $\mf{D}$ denotes the Levi-Civita connection with respect to $\mf{g}$.
\end{prop}

\begin{remark}
The relations \eqref{GaugeTransform} are well-known in the physics literature---see, e.g., \cite{deharo_sken_solod:holog_adscft, imbim_schwim_theis_yanki:diffeo_holog}.
However, here we express them rigorously in finite regularity settings.
\end{remark}

\begin{proof}[Proof of Proposition \ref{PropGaugeTransform}.]
Fix any compact coordinate system $( U, \varphi )$ on $\mc{I}$, and let
\begin{align*}
( \rho, x^1, \dots x^n ) := \varphi_\rho \text{,} \qquad ( \sigma, y^1 \dots, y^n ) := \varphi_\sigma \text{.}
\end{align*}
In other words, the $x^a$'s and $y^b$'s are defined to be constant along the $\partial_\rho$ and $\partial_\sigma$ directions, respectively.
Since $( x^1, \dots, x^n )$ and $( y^1, \dots, y^n )$ coincide on $\mc{I}$, then we have the expansion
\begin{align} \label{gauge_xy}
x^a &= y^a + \mf{q}_1^a \sigma + \mf{q}_2^a \sigma^2 + \mf{q}_3^a \sigma^3 + \mc{O} ( \sigma^4 ) \text{,}
\end{align}
where the $\mf{q}_1^a$'s, $\mf{q}_2^a$'s, and $\mf{q}_3^a$'s are scalar functions on $U$.

We begin with some preliminary computations.
From \eqref{rhosigma} and \eqref{gauge_xy}, we deduce that
\begin{align}
\label{gauge_exp_1} \partial_\sigma \rho &= \mf{r}_1 + 2 \mf{r}_2 \, \sigma + 3 \mf{r}_3 \, \sigma^2 + \mc{O} ( \sigma^3 ) \text{,} \\
\notag \partial_a \rho &= \partial_a \mf{r}_1 \, \sigma + \partial_a \mf{r}_2 \, \sigma^2 + \mc{O} ( \sigma^3 )_a \text{,} \\
\notag \sigma^{-2} \rho^2 &= \mf{r}_1^2 + 2 \mf{r}_1 \mf{r}_2 \, \sigma + ( \mf{r}_2^2 + 2 \mf{r}_1 \mf{r}_3 ) \, \sigma^2 + \mc{O} ( \sigma^3 ) \text{,} \\
\notag \partial_\sigma x^a &= \mf{q}^a_1 + 2 \mf{q}^a_2 \, \sigma + 3 \mf{q}^a_3 \, \sigma^2 + \mc{O} ( \sigma^3 ) \text{,} \\
\notag \partial_a x^b &= \delta_a^b + \partial_a \mf{q}^b_1 \, \sigma + \partial_a \mf{q}^b_2 \, \sigma^2 + \mc{O} ( \sigma^3 )_a \text{,}
\end{align}
whereas \eqref{gauges} and \eqref{gauge_exp_1} yield
\begin{align}
\label{gauge_exp_2} \ms{g}_{ab} (\rho, x) &= \mf{g}_{ab} (x) + \mf{r}_1^2 \bar{\mf{g}}_{ab} (x) \, \sigma^2 + \mc{O}( \sigma^3 )_{ab} (x) \text{.}
\end{align}
Moreover, applying Taylor's theorem around $y^a$-coordinates and recalling \eqref{gauge_xy}, we have
\begin{align*}
\mf{g}_{ab} (x) &= \mf{g}_{ab} (y) + ( x^c - y^c ) \partial_c \mf{g}_{ab} (y) + \frac{1}{2} ( x^c - y^c ) ( x^d - y^d ) \partial_{cd}^2 \mf{g}_{ab} (y) + \mc{O} ( \sigma^3 )_{ab} (y) \\
&= \mf{g}_{ab} (y) + ( \mf{q}_1^c \partial_c \mf{g}_{ab} ) (y) \, \sigma + [ \mf{q}_2^c \partial_c \mf{g}_{ab} + \mf{q}_1^c \mf{q}_2^d \partial_{cd}^2 \mf{g}_{ab} ] (y) \, \sigma^2 + \mc{O} ( \sigma^3 )_{ab} (y) \text{,} \\ 
\bar{\mf{g}}_{ab} (x) &= \bar{\mf{g}}_{ab} (y) + \mc{O} ( \sigma )_{ab} (y) \text{.}
\end{align*}
Thus, evaluating $\ms{g}$ in terms of the $\sigma$-foliation, and using \eqref{gauge_exp_2} and the above, we obtain
\begin{align}
\label{gauge_exp_4} \ms{g}_{ab} = \mf{g}_{ab} + \mf{q}_1^c \partial_c \mf{g}_{ab} \, \sigma + ( \mf{r}_1^2 \bar{\mf{g}}_{ab} + \mf{q}_2^c \partial_c \mf{g}_{ab} + \mf{q}_1^c \mf{q}_1^d \partial_{cd}^2 \mf{g}_{ab} ) \sigma^2 + \mc{O} ( \sigma^3 ) \text{.}
\end{align}

Now, by \eqref{gauges} and the chain rule, we observe that
\begin{align*}
g &= \sigma^{-2} ( d \sigma^2 + \ms{h}_{ab} dy^a dy^b ) \text{,} \\
g &= \rho^{-2} ( d \rho^2 + \ms{g}_{ab} dx^a dx^b ) \\
&= \rho^{-2} [ ( \partial_\sigma \rho )^2 + \ms{g}_{ab} \partial_\sigma x^a \partial_\sigma x^b ] d \sigma^2 + 2 \rho^{-2} [\partial_\sigma \rho \partial_a \rho + \ms{g}_{bc} \partial_a x^c \partial_\sigma x^b ] d \sigma d y^a \\
&\qquad + \rho^{-2} [ \partial_a \rho \partial_b \rho + \ms{g}_{cd} \partial_a x^c \partial_b x^d ] d y^a d y^b \text{.}
\end{align*}
Comparing the two expansions for $g$, we infer
\begin{align}
\label{DefinitionhabGaugeInvariance} 0 &= ( \partial_\sigma \rho )^2 + \ms{g}_{ab} \partial_\sigma x^a \partial_\sigma x^b - \sigma^{-2} \rho^2 \text{,} \\
\notag 0 &= \partial_\sigma \rho \partial_a \rho + \ms{g}_{bc} \partial_a x^c \partial_\sigma x^b \text{,} \\
\notag \sigma^{-2} \ms{h}_{ab} &= \rho^{-2} [ \partial_a \rho \partial_b \rho + \ms{g}_{cd} \partial_a x^c \partial_b x^d ] \text{.}
\end{align}
Expanding the first two parts of \eqref{DefinitionhabGaugeInvariance} and recalling \eqref{gauge_exp_1} and \eqref{gauge_exp_4} then yields
\begin{align*}
0 &= [ \mf{g}_{ab} + \mf{q}_1^c \partial_c \mf{g}_{ab} \, \sigma + ( \mf{r}_1^2 \bar{\mf{g}}_{ab} + \mf{q}_2^c \partial_c \mf{g}_{ab} + \mf{q}_1^c \mf{q}_1^d \partial_{cd}^2 \mf{g}_{ab} ) \sigma^2 ] \\
&\qquad\qquad \cdot ( \mf{q}_1^a + 2 \mf{q}_2^a \, \sigma + 3 \mf{q}_3^a \, \sigma^2 )( \mf{q}_1^b + 2 \mf{q}_2^b \, \sigma + 3 \mf{q}_3^b \, \sigma^2 ) \\
&\qquad + ( \mf{r}_1 + 2 \mf{r}_2 \, \sigma + 3 \mf{r}_3 \, \sigma^2 )^2 - [ \mf{r}_1^2 + 2 \mf{r}_1 \mf{r}_2 \, \sigma + ( \mf{r}_2^2 + 2 \mf{r}_1 \mf{r}_3 ) \sigma^2 ] + \mc{O} ( \sigma^3 ) \\
&= \mf{A}_0 + \mf{A}_1 \, \sigma + \mf{A}_2 \, \sigma^2 + \mc{O} ( \sigma^3 ) \text{,} \\
0 &= [ \mf{g}_{ b c } + \mf{q}_1^d \partial_d \mf{g}_{ b c } \, \sigma + ( \mf{r}_1^2 \bar{\mf{g}}_{ b c } + \mf{q}_2^d \partial_d \mf{g}_{ b c } + \mf{q}^d_1 \mf{q}^e_1 \partial_{ d e } \mf{g}_{ b c } ) \, \sigma^2 ] \\
\notag &\qquad \quad \cdot ( \delta_a^c + \partial_a \mf{q}^c_1 \, \sigma + \partial_a \mf{q}^c_2 \, \sigma^2 ) ( \mf{q}^b_1 + 2 \mf{q}^b_2 \, \sigma + 3 \mf{q}^b_3 \, \sigma^2 ) \\
\notag &\qquad + ( \mf{r}_1 + 2 \mf{r}_2 \, \sigma + 3 \mf{r}_3 \, \sigma^2 ) ( \partial_a \mf{r}_1 \, \sigma + \partial_a \mf{r}_2 \, \sigma^2 ) + \mc{O} ( \sigma^3 ) \\
\notag &:= \mf{A}_{ a, 0 } + \mf{A}_{ a, 1 } \, \sigma + \mf{A}_{ a, 2 } \, \sigma^2 + \mc{O} ( \sigma^3 )_a \text{.} \\
\end{align*}
From the above, we deduce the following values for our coefficients:
\begin{align}
\label{gauge_coeff} 0 = \mf{A}_{a, 0} = \mf{g}_{ab} \mf{q}_1^b \quad &\Longrightarrow \quad \mf{q}_{1}^b = 0 \text{,} \\
\notag 0 = \mf{A}_1 = 2 \mf{r}_1 \mf{r}_2 \quad &\Longrightarrow \quad \mf{r}_2 = 0 \text{,} \\
\notag 0 = \mf{A}_{a, 1} = 2 \mf{g}_{ab} \mf{q}_2^b + \mf{r}_1 \partial_a \mf{r}_1 \quad &\Longrightarrow \quad \mf{q}_2^a = - \frac{1}{2} \mf{r}_1 \mf{g}^{ab} \partial_b \mf{r}_1 \text{,} \\
\notag 0 = \mf{A}_{a, 2} = 3 \mf{g}_{ab} \mf{q}_3^b \quad &\Longrightarrow \quad \mf{q}_3^a = 0 \text{,} \\
\notag 0 = \mf{A}_2 = 4 \mf{g}_{ab} \mf{q}_2^a \mf{q}_2^b + 4 \mf{r}_1 \mf{r}_3 \quad &\Longrightarrow \quad \mf{r}_3 = - \frac{1}{4} \mf{r}_1 \mf{g}^{cd} \partial_c \mf{r}_1 \partial_d \mf{r}_1 \text{.}
\end{align}

Finally, from the last equation of \eqref{DefinitionhabGaugeInvariance}, along with \eqref{gauge_exp_1}, \eqref{gauge_exp_4}, and \eqref{gauge_coeff}, we have
\begin{align}
\label{gauge_final_1} \sigma^{-2} \rho^2 \ms{h}_{ab} &= \ms{g}_{cd} \partial_a x^c \partial_b x^d + \partial_a \rho \partial_b \rho \\
\notag &= \mf{g}_{ab} + ( \mf{r}_1^2 \bar{\mf{g}}_{ab} + \mf{q}_2^c \partial_c \mf{g}_{ab} + \mf{g}_{ad} \partial_b \mf{q}_2^d + \mf{g}_{bd} \partial_a \mf{q}_2^d ) \sigma^2 + \mc{O} ( \sigma^3 ) \\
\notag &\qquad + \partial_a \mf{r}_1 \partial_b \mf{r}_1 \, \sigma^2 + \mc{O} ( \sigma^3 )
\end{align}
The non-tensorial terms in \eqref{gauge_final_1} can be treated using \eqref{gauge_coeff}:
\begin{align*}
\mf{q}_2^c \partial_c \mf{g}_{ab} + \mf{g}_{ad} \partial_b \mf{q}_2^d + \mf{g}_{bd} \partial_a \mf{q}_2^d &= ( \partial_e \mf{g}_{ab} - \mf{g}_{ad} \mf{T}_{be}^d - \mf{g}_{bd} \mf{T}_{ae}^d ) \mf{q}_{2}^e + \mf{D}_{b} ( \mf{g}_{ad} \mf{q}_2^d ) + \mf{D}_{a} ( \mf{g}_{bd} \mf{q}_2^d ) \\
&= -\frac{1}{2} \mf{D}_b ( \mf{r}_1 \partial_a \mf{r}_1 ) - \frac{1}{2} \mf{D}_a ( \mf{r}_1 \partial_b \mf{r}_1 ) \\
&= - \mf{D}_a \mf{r}_1 \mf{D}_b \mf{r}_1  - \mf{r}_1 \mf{D}_{ab}^2 \mf{r}_1 \text{.}
\end{align*}
Combining \eqref{gauge_exp_1}, \eqref{gauge_final_1}, and the above yields
\begin{align*}
\ms{h}_{ab} &= \frac{ \ms{g}_{ab} + ( \mf{r}_1 \bar{\mf{g}}_{ab} - \mf{r}_1 \mf{D}_{ab}^2 \mf{r}_1 ) \sigma^2 + \mc{O} ( \sigma^3 ) }{ \mf{r}_1^2 \left[ 1 - \frac{1}{2} \mf{g}^{cd} \mf{D}_c \mf{r}_1 \mf{D}_d \mf{r}_1 \, \sigma^2 + \mc{O} ( \sigma^3 ) \right] } \\
&= \mf{r}_1 ^{-2} \mf{g}_{ab} + \left( \bar{\mf{g}}_{ab} - \mf{r}_1^{-1} \mf{D}_{ab}^2 \mf{r}_1 + \frac{1}{2} \mf{r}_{1}^{-2} \mf{g}^{cd} \mf{D}_c \mf{r}_1 \mf{D}_d \mf{r}_1 \mf{g}_{ab} \right) \sigma^2  + \mc{O} ( \sigma^3 ) \text{.}
\end{align*}
The desired \eqref{GaugeTransform} follows from the above after setting $\mf{r}_1 = \mf{a}^{-1}$.
\end{proof}

Using the transformations \eqref{GaugeTransform}, we can show that the GNCC is indeed gauge-invariant:

\begin{prop}[Gauge invariance] \label{Gaugeinvariance}
Let $( \mc{M}, g )$ be as above, and assume the two FG-gauges \eqref{gauges} and the comparison \eqref{rhosigma}.
In addition, let $\mc{D}$ and $\mf{h}$ be as in Definition \ref{DefAdmissibleDomains1}.
Then, $\mc{D}$ satisfies the generalized null convexity condition with respect to $( \rho, \ms{g} )$-gauge if and only if it satisfies the generalized null convexity condition with respect to $( \sigma, \ms{h} )$-gauge.
\end{prop}

\begin{proof}
Suppose $\mc{D}$ satisfies the GNCC in the $( \rho, \ms{g} )$-gauge, so there exists $\eta \in C^4 ( \bar{\mc{D}} )$ with
\begin{align*} 
\begin{cases}
( \mf{D}^2 \eta - \eta \bar{\mf{g}} ) ( \mf{Y}, \mf{Y} ) > c \eta \, \mf{p} ( \mf{Y}, \mf{Y} ) &\text{in } \mc{D} \text{,} \\
\eta > 0 &\text{in } \mc{D} \text{,} \\
\eta = 0 &\text{on } \partial \mc{D} \text{,}
\end{cases}
\end{align*}
with $c > 0$, and with $\mf{g}$-null $\mf{Y} \in T \mc{D}$.
Now, let $\mf{a} := \mf{r}^{-1} > 0$, and let $\xi := \mf{a} \eta$.
Furthermore, we denote by $\mf{O}$ the Levi-Civita connection with respect to $\mf{h}$.

Clearly, $\xi > 0$ on $\mc{D}$ and $\xi = 0$ on $\partial \mc{D}$ by definition.
Let $\mf{Z} \in T \mc{D}$ be $\mf{h}$-null; since $\mf{h} = \mf{a}^2 \mf{g}$ by \eqref{GaugeTransform}, then $\mf{Z}$ is $\mf{g}$-null as well.
Moreover, using standard formulas for conformal transformation to relate $\mf{D}$ and $\mf{O}$, along with \eqref{GaugeTransform}, we compute that
\begin{align*}
\mf{O}^2 \xi ( \mf{Z}, \mf{Z} ) &= \mf{D}^2 \xi ( \mf{Z}, \mf{Z} ) - \mf{a}^{-1} [ 2 \mf{D}_{ \mf{Z} } \mf{a} \mf{D}_{ \mf{Z} } \xi - \mf{g} ( \mf{Z}, \mf{Z} ) \mf{g} ( \mf{D}^\sharp \mf{a}, \mf{D}^\sharp \xi ) ] \\
&= \mf{D}^2 ( \mf{a} \eta ) ( \mf{Z}, \mf{Z} ) - 2 \mf{a}^{-1} \mf{D}_{ \mf{Z} } \mf{a} \mf{D}_{ \mf{Z} }( \mf{a} \eta ) \\
&= \mf{a} \mf{D}^2 \eta ( \mf{Z}, \mf{Z} ) + \eta \mf{D}^2 \mf{a} ( \mf{Z}, \mf{Z} ) - 2 \eta \mf{a}^{-1} ( \mf{D}_{ \mf{Z} } \mf{a} )^2 \text{,} \\
\xi \, \bar{\mf{h}} ( \mf{Z}, \mf{Z} ) &= \mf{a} \eta \, \bar{\mf{g}} ( \mf{Z}, \mf{Z} ) + \eta \mf{D}^2 \mf{a} ( \mf{Z}, \mf{Z} ) - 2 \eta \mf{a}^{-1} ( \mf{D}_{ \mf{Z} } \mf{a} )^2 \text{.}
\end{align*}
In particular, the above implies that
\begin{align*}
( \mf{O}^2 \xi - \xi \, \bar{\mf{h}} ) ( \mf{Z}, \mf{Z} ) &= \mf{a} ( \mf{D}^2 \eta - \eta \, \bar{\mf{g}} ) ( \mf{Z}, \mf{Z} ) > c \xi \, \mf{p} ( \mf{Z}, \mf{Z} ) \text{,}
\end{align*}
since $\mf{a} > 0$.
This shows that GNCC indeed holds with respect to the $( \sigma, \ms{h} )$-gauge.
Finally, the converse statement can be proved using a symmetric argument.
\end{proof}

\subsection{Examples and Counterexamples}

In this subsection, we study some conditions under which the GNCC either is satisfied or fails to hold.
We begin with counterexamples; for this, our first tool is to restrict \eqref{BVP} to individual null geodesics.

\begin{lemma}[Necessary size of $\bar{\mf{g}}$] \label{ComparisonPrinciple2}
Let $\mc{D} \subseteq \mc{I}$ be open and with compact closure, and assume the GNCC \eqref{BVP} holds on $\mc{D}$, for some $\eta$, $c$, $\mf{p}$.
In addition, let $\lambda: [ 0, \ell ] \longrightarrow \mc{I}$ be any $\mf{g}$-null geodesic on $\mc{I}$ that also satisfies the following conditions:
\begin{equation} \label{ComparisonPrincipleAss}
\lambda (0), \lambda ( \ell ) \in \partial \mc{D} \text{,} \qquad \lambda (s) \in \mc{D} \text{,} \quad 0 < s < \ell \text{.}
\end{equation}
Then, $\bar{\mf{g}}$ satisfies the following inequality along $\lambda$:
\begin{equation} \label{ComparisonPrincipleEq}
\sup_{ s \in [0, \ell] } [ - \bar{\mf{g}} ( \dot{\lambda} (s), \dot{\lambda} (s) ) ] > \frac{ \pi^2 }{ \ell^2 } \text{.}
\end{equation}
\end{lemma}

\begin{proof}
For convenience, we define the following along $\lambda$:
\begin{equation} \label{Comparison_setup_1}
\theta (s) := \eta ( \lambda (s) ) \text{,} \qquad \mf{c} (s) := \bar{\mf{g}} ( \dot{\lambda} (s), \dot{\lambda} (s) ) \text{,} \qquad 0 \leq s \leq \ell \text{.}
\end{equation}
Observe that by compactness, there is some constant $0 < p < \frac{\pi}{\ell}$ such that
\begin{equation} \label{Comparison_setup_2}
c \, \mf{p} ( \dot{\lambda} (s), \dot{\lambda} (s) ) \geq p^2 \text{,} \qquad 0 \leq s \leq \ell \text{.}
\end{equation}
In addition, we consider the quantity
\begin{equation} \label{Comparison_setup_3}
\varphi (s) := \sin \left( \sqrt{ \frac{ \pi^2 }{ \ell^2 } - p^2 } \cdot s \right) \text{,} \qquad 0 \leq s \leq \ell \text{.}
\end{equation}
Then, \eqref{BVP} and \eqref{Comparison_setup_1}--\eqref{Comparison_setup_3} imply that the following hold:
\begin{align} \label{BVP_geodesic}
\begin{dcases}
\ddot{\theta} - ( \mf{c} + p^2 ) \theta > 0 &\quad \text{on $[0, \ell]$,} \\
\theta > 0 &\quad \text{on $(0, \ell)$,} \\
\theta (0) = \theta ( \ell ) = 0 \text{,}
\end{dcases} \qquad
\begin{dcases}
\ddot{\varphi} + \left( \frac{ \pi^2 }{ \ell^2 } - p^2 \right) \varphi = 0 &\quad \text{on $[0, \ell]$,} \\
\varphi > 0 &\quad \text{on $(0, \ell]$,} \\
\varphi (0) = 0 \text{,} \quad \varphi (\ell) > 0 \text{.}
\end{dcases}
\end{align}

Assume, for a contradiction, that the desired \eqref{ComparisonPrincipleEq} does not hold, that is,
\begin{equation} \label{eql.ComparisonPrinciple2_0}
- \mf{c} (s) \leq \frac{ \pi^2 }{ \ell^2 } \text{,} \qquad 0 \leq s \leq \ell \text{.}
\end{equation}
We now apply a Sturm comparison argument---from \eqref{BVP_geodesic}, the assumption \eqref{eql.ComparisonPrinciple2_0}, and two integrations by parts, we obtain for each $0 < s < \ell$ that
\begin{align*}
0 &< \int_0^s [ \ddot{\theta} - \mf{c} \, \theta - p^2 \, \theta ] \varphi - \int_0^s \theta \left( \ddot{\varphi} + \frac{ \pi^2 }{ \ell^2 } \, \varphi - p^2 \, \varphi \right) \\
&\leq \dot{\theta} (s) \varphi (s) - \theta (s) \dot{\varphi} (s) - \int_0^s \left( \mf{c} - \frac{ \pi^2 }{ \ell^2 } \right) \theta \varphi \\
&\leq \varphi^2 (s) \left( \frac{ \theta }{ \varphi } \right)' (s) \text{.}
\end{align*}

Dividing the above by $\varphi^2 (s) > 0$ and integrating from a fixed $0 < s_0 < \ell$ yields
\[
\frac{ \theta (s) }{ \varphi (s) } > \frac{ \theta ( s_0 ) }{ \varphi (s_0) } \text{,} \quad s_0 < s < \ell \text{,} \qquad \theta (\ell) \geq \frac{ \theta ( s_0 ) }{ \varphi (s_0) } \cdot \varphi (\ell) > 0 \text{,}
\]
where the second part follows from taking $s \nearrow \ell$.
On the other hand, since $\lambda (\ell) \in \partial \mc{D}$, then \eqref{Comparison_setup_1} implies $\theta (\ell) = 0$, contradicting the above.
As a result, \eqref{ComparisonPrincipleEq} must hold.
\end{proof}

\begin{remark}
An equivalent formulation of Lemma \ref{ComparisonPrinciple2} is that if, in the same setting,
\[
- \bar{\mf{g}} ( \dot{\lambda} (s), \dot{\lambda} (s) ) \leq B^2 \text{,} \qquad 0 \leq s \leq \ell
\]
for some constant $B > 0$, then $\ell > \frac{\pi}{B}$.
\end{remark}

Lemma \ref{ComparisonPrinciple2} is useful for identifying situations in which the GNCC cannot hold:

\begin{cor}[Failure of GNCC] \label{NoGNCC}
Suppose $- \bar{\mf{g}} ( \mf{X}, \mf{X} ) \leq 0$ for all vectors $\mf{X} \in T \mc{I}$ satisfying $\mf{g} ( \mf{X}, \mf{X} ) = 0$.
Then, the GNCC is not satisfied by any domain $\mc{D} \subseteq \mc{I}$.
\end{cor}

\begin{proof}
Applying Lemma \ref{ComparisonPrinciple2} to any $\mf{g}$-null geodesic $\lambda$ passing through $\mc{D}$ yields a contradiction, since \eqref{ComparisonPrincipleEq} implies that $- \bar{\mf{g}} ( \dot{\lambda}, \dot{\lambda} )$ must be positive somewhere.
\end{proof}

Recall (see \cite{2020arXiv200807396S}) that when $( \mc{M}, g )$ is Einstein-vacuum,
\[
\operatorname{Ric} (g) = \Lambda g \text{,} \qquad \Lambda := - \frac{n(n-1)}{2} < 0 \text{,}
\]
and has boundary dimension $n \geq 3$, then one also has
\[
- ( n - 2 ) \bar{\mf{g}} = \mf{Ric} ( \mf{g} ) \text{,}
\]
with $\mf{Ric}$ being the Ricci curvature for $\mf{g}$.
As a result, Corollary \ref{NoGNCC} implies that \emph{in vacuum spacetimes, no conformal boundary that is non-positively curved in null directions},
\[
\mf{Ric} ( \mf{X}, \mf{X} ) \leq 0 \text{,} \qquad \mf{X} \in T \mc{I} \text{,} \quad \mf{g} ( \mf{X}, \mf{X} ) = 0 \text{,}
\]
\emph{can have a subdomain $\mc{D} \subseteq \mc{I}$ satisfying the GNCC}.
In particular, the above applies to planar and toric Schwarzschild-AdS spacetimes (which satisfy $\mf{Ric} = 0$):

\begin{cor}[Counterexamples to GNCC] \label{NoGNCC_pAdS}
Given any planar or toric Schwarzschild-AdS spacetime, no subdomain of its conformal boundary can satisfy the GNCC.
\end{cor}

Our next objective is to connect the GNCC to the null convexity criterion (abbreviated \emph{NCC}) of \cite{Arick3}.
This will allow us to generate some common examples for which the GNCC is satisfied.
First, let us recall a precise formulation of the NCC:

\begin{definition}[Null Convexity Criterion, \cite{Arick3}] \label{NCC}
Let $t$ be a time function on $\mc{I}$, satisfying
\[
K^{-1} \leq - g ( \mf{D}^\sharp t, \mf{D}^\sharp t ) \leq K \text{,} \qquad K \geq 1 \text{,}
\]
and assume that the level sets of $t$ are compact Cauchy hypersurfaces.
Then, we say that the \emph{null convexity criterion} holds on an open subset $\mc{D} \subseteq \mc{I}$, with associated constants $0 \leq B < C$, iff the following inequalities hold for any vector $\mf{Z} \in \mc{T} \mc{D}$ with $\mf{g} ( \mf{Z}, \mf{Z} ) = 0$:
\begin{equation} \label{NCC_ineq}
- \bar{\mf{g}} ( \mf{Z}, \mf{Z} ) \geq C^2 \cdot ( \mf{Z} t )^2 \text{,} \qquad | \mf{D}^2 t ( \mf{Z}, \mf{Z} ) | \leq 2 B \cdot ( \mf{Z} t )^2 \text{.}
\end{equation}
\end{definition}

\begin{remark}
The compactness assumptions in Definition \ref{NCC}, and elsewhere in this paper, are imposed to simplify discussions.
Otherwise, one encounters additional technical issues in proving unique continuation results for \eqref{KleinGordon}, due to integrability issues and the need for uniform bounds on the geometry.
For further discussions, see \cite{Arick2, Arick3}, which sidestep these issues by assuming compact support properties on their solutions of \eqref{KleinGordon}.
\end{remark}

Our main observation relating the GNCC and the NCC can be roughly characterized as the NCC implying the GNCC on a sufficiently large timespan:

\begin{prop}[NCC implies GNCC] \label{NCC_GNCC}
Let $t$ be a time function on $\mc{I}$ satisfying the same assumptions as in Definition \ref{NCC}, and consider a boundary domain of the form
\begin{equation} \label{D_t}
\mc{D}_{ t_-, t_+ } := \{ t_- < t < t_+ \} \subseteq \mc{I} \text{,} \qquad \inf_{ \mc{I} } t < t_- < t_+ < \sup_{ \mc{I} } t \text{.}
\end{equation}
Then, if the NCC holds on $\mc{D}_{ t_-, t_+ }$, with constants $0 \leq B < C$, and if
\footnote{Here, the value of $\tan^{-1}$ in \eqref{NCC_timespan} is chosen to lie in $[ \frac{\pi}{2}, \pi )$.}
\begin{equation} \label{NCC_timespan}
t_+ - t_- > \mc{T}_+ ( B, C ) := \frac{2}{ \sqrt{ C^2 - B^2 } } \, \tan^{-1} \left( - \frac{ \sqrt{ C^2 - B^2 } }{ B } \right) \text{,}
\end{equation}
then the GNCC is also satisfied on $\mc{D}_{ t_-, t_+ }$.
\end{prop}

\begin{proof}
Observe that there exist constants $B < b < c < C$ such that
\[
t_+ - t_- = \mc{T}_+ ( b, c ) > \mc{T}_+ ( B, C ) \text{.}
\]
By time translation, we can also assume without loss of generality $t_\pm = \pm \frac{1}{2} \mc{T}_+ ( b, c )$.
Let
\begin{equation} \label{NCC_GNCC_phi}
\varphi (s) := e^{ -b |s| } \, \sin \left( \frac{ \sqrt{ c^2 - b^2 } }{2} ( \mc{T}_+ ( b, c ) - 2 | s | ) \right) \text{,} \qquad s \in \R \text{.}
\end{equation}
From direct computations (see also \cite[Proposition 5.2]{Arick3}), we see that $\varphi$ satisfies
\begin{equation} \label{NCC_GNCC_ODE}
\begin{dcases}
\ddot{\varphi} (s) - 2 b \, \dot{\varphi} (s) + c^2 \varphi (s) = 0 \text{,} &\qquad s \leq 0 \text{,} \\
\ddot{\varphi} (s) + 2 b \, \dot{\varphi} (s) + c^2 \varphi (s) = 0 \text{,} &\qquad s \geq 0 \text{.}
\end{dcases}
\end{equation}

Next, we define on $\bar{\mc{D}}_{ t_-, t_+ }$ the function
\begin{equation} \label{NCC_GNCC_eta}
\eta := \varphi (t) \text{.}
\end{equation}
By \eqref{NCC_GNCC_phi} and \eqref{NCC_GNCC_ODE}, we see that $\eta > 0$ on $\mc{D}_{ t_-, t_+ }$, and that $\eta = 0$ on $\partial \mc{D}_{ t_-, t_+ }$.
Furthermore, given any $\mf{g}$-null vector $\mf{Z} \in T \mc{D}_{ t_-, t_+ }$, a direct computation yields that
\begin{align*}
\mf{D}^2 \eta ( \mf{Z}, \mf{Z} ) &= \mf{Z}^a \mf{Z}^b [ \ddot{\varphi} (t) \, \partial_a t \partial_b t + \dot{\varphi} (t) ( \partial_{ a b } t - \mf{T}^d_{ a b } \partial_d t ) ] \\
&= \ddot{\varphi} (t) \, ( \mf{Z} t )^2 + \mf{D}^2 t ( \mf{Z}, \mf{Z} ) \, \dot{\varphi} (t) \text{.}
\end{align*}
From the above and \eqref{NCC_GNCC_ODE}, we then obtain
\begin{align*}
( \mf{D}^2 \eta - \bar{\mf{g}} \, \eta ) ( \mf{Z}, \mf{Z} ) &= \ddot{\varphi} (t) \, ( \mf{Z} t )^2 + \mf{D}^2 t ( \mf{Z}, \mf{Z} ) \, \dot{\varphi} (t) - \bar{\mf{g}} ( \mf{Z}, \mf{Z} ) \, \varphi (t) \\
&= \begin{cases} [ \mf{D}^2 t ( \mf{Z}, \mf{Z} ) + 2 b ( \mf{Z} t )^2 ] \dot{\varphi} (t) + [ - \bar{\mf{g}} ( \mf{Z}, \mf{Z} ) - c^2 ( \mf{Z} t )^2 ] \varphi (t) \text{,} &\quad t \leq 0 \text{,} \\ [ \mf{D}^2 t ( \mf{Z}, \mf{Z} ) - 2 b ( \mf{Z} t )^2 ] \dot{\varphi} (t) + [ - \bar{\mf{g}} ( \mf{Z}, \mf{Z} ) - c^2 ( \mf{Z} t )^2 ] \varphi (t) \text{,} &\quad t \geq 0 \text{.} \end{cases}
\end{align*}

Note in particular \eqref{NCC_ineq} implies that in either case, both terms in the right-hand side of the above are positive on $\mc{D}_{ t_-, t_+ }$.
(Here, we used that $\dot{\varphi} \geq 0$ when $- \frac{1}{2} \mc{T}_+ ( b, c ) < s \leq 0$, and that $\dot{\varphi} \leq 0$ when $0 \leq s < \frac{1}{2} \mc{T}_+ ( b, c )$.)
As a result, we obtain that
\[
( \mf{D}^2 \eta - \bar{\mf{g}} \, \eta ) ( \mf{Z}, \mf{Z} ) \geq ( C^2 - c^2 ) ( \mf{Z} t )^2 \, \eta \text{,}
\]
which implies \eqref{BVP} indeed holds, with $\mc{D} := \mc{D}_{ t_-, t_+ }$ and $\eta$ as in \eqref{NCC_GNCC_eta}.

One final issue is that $\eta$ only lies in $C^2 ( \bar{\mc{D}}_{ t_-, t_+ } )$; see \cite[Proposition 5.2]{Arick3}.\footnote{In particular, $\mf{D}^3 \eta$ becomes discontinuous at $t = 0$.}
The proof is completed by perturbing $\eta$ to a smoother function in a way that preserves \eqref{BVP}.
\end{proof}

For example, the Kerr-AdS (and hence the Schwarzschild-AdS) spacetimes satisfy
\[
\mf{g} = - dt^2 + d \omega^2 \text{,} \qquad \bar{\mf{g}} := -\frac{1}{2} ( dt^2 + d \omega^2 ) \text{,}
\]
with $d \omega^2$ being the unit round metric on $\Sph^{n-1}$.
It follows that the NCC holds on the conformal boundary, with $( B, C ) = ( 0, 1 )$.
From the above, we can establish the following:

\begin{cor}[Examples of GNCC] \label{GNCC_AdS}
On any Kerr-AdS spacetime, the GNCC is satisfied on a timeslab $\mc{D}_{ t_-, t_+ } := \{ t_- < t < t_+ \} \subseteq \mc{I}$ if and only if $t_+ - t_- > \pi$.
\end{cor}

\begin{proof}
That $t_+ - t_- > \pi$ implies the GNCC follows from Proposition \ref{NCC_GNCC}, in the special case $( B, C ) = ( 0, 1 )$.
The converse statement can be proved using Lemma \ref{ComparisonPrinciple2}.
\end{proof}

\subsection{Causal Diamonds}

We conclude this section with some examples of boundary domains that satisfy or violate the GNCC.
Here, we focus our attention on causal diamond domains, which are often considered in the physics literature.
Moreover, we separate the cases $n = 2$ and $n \geq 3$, as the conclusions are quite different in these two settings.

For simplicity, we restrict our attention to the case of Minkowski boundary geometry,
\begin{equation} \label{AdSmetrics}
\mc{I} := \R^n \text{,} \qquad \mf{g} := - dt^2 + d ( \omega^1 )^2 + \dots + d ( \omega^{n-1} )^2 \text{.}
\end{equation}
Though this particular setting is somewhat contrived, it does allow for some explicit formulas and computations.
It will be apparent from the upcoming discussions that similar qualitative results can also be derived for more general curved boundary geometries.
 
\subsubsection{The Case $n = 2$}

First, we study the case $n = 2$, that is,
\begin{equation} \label{AdSmetrics2d}
( \mc{I}, \mf{g} ) := ( \R^2, - dt^2 + d \omega^2 ) \text{.}
\end{equation}
For our boundary region, we fix any $\alpha > 0$, and we consider the causal diamond
\begin{equation} \label{D2d}
\mc{D}_\alpha := \{ (t, \omega) \in \mc{I} \mid -\alpha < t \pm \omega < \alpha \} \text{.}
\end{equation}
In particular, $\mc{D}_\alpha$ is the interior of the red diamond drawn in Figure \ref{Pic7cite}. 

Next, assume an arbitrary $\bar{\mf{g}}$ in our setting that satisfies the following on $\mc{I}$:
\begin{equation} \label{AdS_gbar_2d}
- \bar{\mf{g}} ( \mf{Z}, \mf{Z} ) \geq \beta^2 \text{,} \qquad \mf{Z} = \partial_t \pm \partial_\omega \text{,} \quad \beta > 0 \text{.}
\end{equation}
Consider now the function $\eta \in C^\infty ( \bar{\mc{D}}_\alpha )$ given by
\begin{equation} \label{eta_2d}
\eta (t, \omega) := \sin \left( \frac{ \pi (t - \omega + \alpha ) }{ 2 \alpha } \right) \sin \left( \frac{ \pi ( t + \omega + \alpha ) }{ 2 \alpha } \right) \text{.}
\end{equation}
By definition, $\eta$ is positive within $\mc{D}_\alpha$ and vanishes on $\partial \mc{D}_\alpha$, so the second and third properties of \eqref{BVP} hold.
For the remaining first part of \eqref{BVP}, it suffices to check this for $\mf{Z} := \partial_t \pm \partial_\omega$:
\begin{align*}
\mf{D}^2 \eta ( \mf{Z}, \mf{Z} ) |_{ ( t, \omega ) } &= \partial_{tt} \eta ( t, \omega ) \pm 2 \partial_{ t \omega } \eta ( t, \omega ) + \partial_{\omega \omega} \eta ( t, \omega ) \\
&= - \left( \frac{ \pi }{ \alpha } \right)^2 \sin \left( \frac{ \pi (t - \omega + \alpha ) }{ 2 \alpha } \right) \sin \left( \frac{ \pi ( t + \omega + \alpha ) }{ 2 \alpha } \right) \\
&= - \left( \frac{ \pi }{ \alpha } \right)^2 \eta ( t, \omega ) \text{.}
\end{align*}

From the above, we see that the first part of \eqref{BVP} holds, with $\eta$ as in \eqref{eta_2d}, if $\beta > \frac{ \pi }{ \alpha }$.
In other words, the causal diamond $\mc{D}_\alpha$ satisfies the GNCC as long as it is sufficiently large:

\begin{prop} \label{AdS_2d}
Assume the conformal boundary \eqref{AdSmetrics2d}, and assume the condition \eqref{AdS_gbar_2d} holds.
Then, the causal diamond $\mc{D}_\alpha$ from \eqref{D2d} satisfies the GNCC whenever $\alpha > \frac{ \pi }{ \beta }$.
\end{prop}

\begin{figure}[ht]
    \centering
  \includegraphics[width=0.65\textwidth]{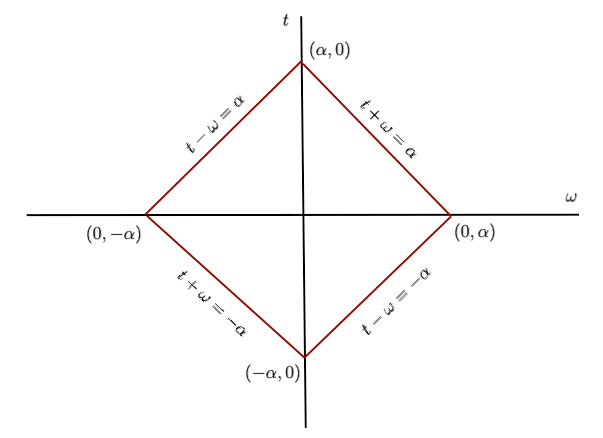}
    \caption{The causal diamond region $\mc{D}_\alpha$ in $(t,\omega)$-coordinates.}
    \label{Pic7cite}
\end{figure}

\subsubsection{The Cases $n \geq 3$}

Next, we consider higher-dimensional cases---i.e., the setting \eqref{AdSmetrics} with $n \geq 3$.
Here, the analogues of the causal diamonds \eqref{D2d} are given by
\begin{equation} \label{Dhd}
\mc{D}_\alpha := \{ (t, \omega) \in \mc{I} = \R \times \R^{n-1} \mid t + | \omega | < \alpha \text{, } t - | \omega | > -\alpha \} \text{,} \qquad \alpha > 0 \text{.}
\end{equation}

However, the situation here differs significantly from the $n = 2$ case.
The key observation is that there are many more null geodesics to consider in higher dimensions.
In particular, one can find null geodesics near the sphere $S_\alpha := \{ t = 0 \text{, } | \omega | = \alpha \} \subseteq \partial \mc{D}_\alpha$ that pass through $\mc{D}_\alpha$ but spend an arbitrarily small amount of time within $\mc{D}_\alpha$.

To be more explicit, one can consider the null geodesics
\begin{equation} \label{AdS_geod_hd}
\lambda (s) := ( \lambda^t (s), \lambda^\omega (s) ) := \left( s, \alpha - \delta, s, 0, \dots, 0 \right) \text{,} \qquad 0 < \delta < \alpha \text{.}
\end{equation}
Note that by taking $\delta$ as small as needed, one can have $\lambda$ satisfy
\[
\lambda ( \pm \ell ) \in \partial \mc{D}_\alpha \text{,} \qquad \lambda (s) \in \mc{D}_\alpha \text{,} \quad - \ell < s < \ell \text{,}
\]
for $\ell > 0$ arbitrarily small.
If $\mc{D}_\alpha$ were to satisfy the GCC, then Lemma \ref{ComparisonPrinciple2} would yield that some components of $- \bar{\mf{g}}$ must be positive and arbitrarily large near $S_\alpha$.
This produces a contradiction, and it follows that $\mc{D}_\alpha$ cannot satisfy the GCC:

\begin{prop} \label{AdS_hd}
Assume the conformal boundary \eqref{AdSmetrics}, with $n \geq 3$.
Then, the causal diamond $\mc{D}_\alpha$ from \eqref{Dhd} fails to satisfy the GNCC for any $\alpha > 0$.
\end{prop}

Consequently, no causal diamond of any size satisfies the GNCC in the higher-dimensional cases.
Furthermore, Proposition \ref{AdS_hd} continues to hold for more general conformal boundaries, since the above intuitions carry over to curved settings.

\begin{figure}[ht]
    \centering
  \includegraphics[width=0.6\textwidth]{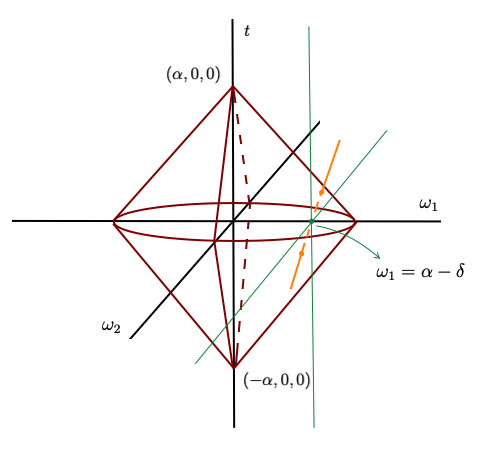}
    \caption{The region bounded by the two red cones is the causal diamond $\mc{D}_\alpha$ from \eqref{Dhd}, in the case $n = 3$.
    The null geodesic line $\lambda$ from \eqref{AdS_geod_hd}, which passes through $\mc{D}_\alpha$ for an arbitrarily short time, is drawn in orange.}
    \label{Pic12cite}
\end{figure}

\section{Characterization of Null Geodesics} \label{sec.geodesic}

In this section, we give a precise statement and proof of our second main result (informally stated in Theorem \ref{MainTheorem2}), which shows that the generalized null convexity criterion of Definition \ref{DefAdmissibleDomains1} governs the trajectories of (spacetime) null geodesics near the conformal boundary.

Roughly, the result states that \emph{if $\mc{D} \subseteq \mc{I}$ satisfies the GNCC, then any $g$-null geodesic that is close enough to the conformal boundary and that passes over $\mc{D}$ must either initiate from or terminate at $\mc{D}$}.
As discussed in Section \ref{sec.intro}, this property precisely rules out the standard geometric optics counterexamples, constructed in \cite{MR1363855}, to unique continuation for waves from $\mc{D}$.
(Later, in Section \ref{sec.carleman}, we prove the GNCC in fact implies unique continuation from $\mc{D}$.)

The precise statement of our second main result is as follows:

\begin{theorem}[Main result 2: Characterization of null geodesics] \label{Characterisation}
Let $( \mc{M}, g )$ be a spacetime satisfying Assumptions \ref{ass.manifold} and \ref{ass.aads}.
In addition, let $\Lambda: ( s_-, s_+ ) \longrightarrow \mc{M}$ be a complete null geodesic with respect to the metric $\rho^2 g = d \rho^2 + \ms{g}$, written as
\footnote{Since $g$ and $\rho^2 g$ have the same null geodesics, we can more usefully interpret $\Lambda$ as being a $g$-null geodesic.
However, we use $\rho^2 g$ in Theorem \ref{Characterisation}, since the $\rho^2 g$-affine parametrization of $\Lambda$ is more convenient.}
\[
\Lambda (s) := ( \rho (s), \lambda(s) ) \in ( 0, \rho_0 ) \times \mc{I} \text{,} \qquad s \in ( s_-, s_+ ) \text{,}
\]
with $s$ being an affine parameter for $\Lambda$.
In addition:
\begin{itemize}
\item Let $\mc{D} \subseteq \mc{I}$ be open with compact closure, and suppose $\mc{D}$ satisfies the GNCC.

\item Let $\epsilon_0 > 0$ be sufficiently small, depending on the metric $\ms{g}$ on $\mc{D}$.

\item Suppose there exists $s_0 \in ( s_-, s_+ )$ such that
\footnote{In other words, $\Lambda$ is both hovering over $\mc{D}$ and ``$\epsilon$-close" to $\mc{D}$.}
\begin{equation} \label{Characterisation_ass}
0 < \rho (s_0) < \epsilon_0 \text{,} \qquad | \dot{\rho} (s_0) | \lesssim \rho (s_0) \text{,} \qquad \lambda (s_0) \in \mc{D} \text{.}
\end{equation}
\end{itemize}
Then, at least one of the following holds:
\begin{itemize}
\item $\Lambda$ initiates from the conformal boundary within $\mc{D}$:
\begin{equation} \label{Characterisation_init}
\lim_{ s \searrow s_- } \rho (s) = 0 \text{,} \qquad \lim_{ s \searrow s_- } \lambda (s) \in \mc{D} \text{.}
\end{equation}

\item $\Lambda$ terminates at the conformal boundary within $\mc{D}$:
\begin{equation} \label{Characterisation_term}
\lim_{ s \nearrow s_+ } \rho (s) = 0 \text{,} \qquad \Hquad \lim_{ s \nearrow s_+ } \lambda (s) \in \mc{D} \text{.}
\end{equation}
\end{itemize}
\end{theorem}

\subsection{Proof of Theorem \ref{Characterisation}}

Assume the hypotheses of Theorem \ref{Characterisation}.
The first step is to describe the behaviour of null geodesics near the conformal boundary:

\begin{lemma}[Null geodesic equations] \label{NullGeodesicEquation}
The following hold for each $s \in ( s_-, s_+ )$:
\begin{align} \label{NullGeodesicAsymp}
\ddot{\rho} (s) &= \rho(s) \cdot \bar{\mf{g}} ( \dot{\lambda} (s), \dot{\lambda}(s) ) + \rho^2 (s) \cdot \mc{O} (1) ( \dot{\lambda} (s), \dot{\lambda} (s) ) \text{,} \\
\notag ( \mf{D}_{ \dot{\lambda} } \dot{\lambda} ) (s) &= \rho^2 (s) \cdot \mc{O} (1) ( \dot{\lambda} (s), \dot{\lambda} (s) ) - \rho (s) \dot{\rho} (s) \cdot \mc{O} (1) ( \dot{\lambda} (s) ) \text{,} \\
\notag \dot{\rho}^2 (s) + \mf{g} ( \dot{\lambda} (s), \dot{\lambda} (s) ) &= \rho^2 (s) \cdot \mc{O} (1) ( \dot{\lambda} (s), \dot{\lambda} (s) ) \text{.}
\end{align}
\end{lemma}

\begin{proof}
Fix an arbitrary coordinate system $( U, \varphi )$ of $\mc{I}$ along $\lambda$.
A direct computation shows that the Christoffel symbols $\hat{\Gamma}^\alpha_{ \mu \nu }$ for $\rho^2 g$ in $\varphi_\rho$-coordinates satisfy
\begin{align*}
\hat{\Gamma}_{\rho \rho}^\rho = 0 \text{,} \qquad \hat{\Gamma}_{\rho b}^\rho &= 0 \text{,} \qquad \hat{\Gamma}_{\rho \rho}^c = 0 \text{,} \\
\hat{\Gamma}_{b c}^\rho = -\frac{1}{2} \mc{L}_\rho \ms{g}_{bc} \text{,} \qquad \hat{\Gamma}_{\rho b}^c &= \frac{1}{2} \ms{g}^{c e} \mc{L}_\rho \ms{g}_{eb} \text{,} \qquad \hat{\Gamma}_{ab}^c = \ms{\Gamma}_{ab}^c \text{.}
\end{align*}
Then, the above and the geodesic equations imply
\begin{align}
\label{eql.NullGeodesicEquation_1} 0 &= \ddot{\rho} (s) + \Gamma_{\rho \rho}^\rho \cdot \dot{\rho}^2 (s) + 2 \hat{\Gamma}_{\rho b}^\rho \cdot \dot{\rho} (s) \dot{\lambda}^b (s) + \hat{\Gamma}_{b c}^\rho \cdot \dot{\lambda}^b (s) \dot{\lambda}^c (s) \\
\notag &= \ddot{\rho} (s) - \frac{1}{2} \mc{L}_\rho \ms{g} ( \dot{\lambda} (s), \dot{\lambda} (s) ) \text{,} \\
\notag 0 &= \ddot{\lambda}^c (s) + \hat{\Gamma}_{\rho \rho}^c \cdot \dot{\rho}^2 (s) + 2 \hat{\Gamma}_{\rho b}^c \cdot \dot{\rho} (s) \dot{\lambda}^b (s) + \hat{\Gamma}_{a b}^c \cdot \dot{\lambda}^a (s) \dot{\lambda}^b (s) \\
\notag &= ( \ms{D}_{ \dot{\lambda} } \dot{\lambda} )^c (s) + \dot{\rho} (s) \dot{\lambda}^b (s) \cdot \ms{g}^{ce} \mc{L}_\rho \ms{g}_{be} \text{.}
\end{align}
Furthermore, since $\Lambda$ is $\rho^2 g$-null, \eqref{metricdefiniton1} yields
\begin{equation} \label{eql.NullGeodesicEquation_2}
0 = \dot{\rho}^2 (s) + \ms{g} ( \dot{\lambda} (s), \dot{\lambda} (s) ) \text{.}
\end{equation}

Note that the asymptotics of Lemma \ref{LemmaMetric} imply
\begin{align*}
\ms{g}^{-1} &= \mf{g}^{-1} + \mc{O} ( \rho^2 ) \text{,} \qquad \mc{L}_\rho \ms{g} = 2 \rho \cdot \bar{\mf{g}} + \mc{O} ( \rho^2 ) \text{,} \qquad \ms{D}_{ \mf{X} } \mf{Z} = \mf{D}_{ \mf{X} } \mf{Z} + \mc{O} ( \rho^2 ) ( \mf{X}, \mf{Z} ) \text{,}
\end{align*}
for any boundary tangent vectors $\mf{X}, \mf{Z} \in T \mc{I}$.
The desired relations \eqref{NullGeodesicAsymp} now follow from \eqref{eql.NullGeodesicEquation_1}, \eqref{eql.NullGeodesicEquation_2}, and the above, since all the equations are evaluated at $\rho$-value $\rho (s)$.
\end{proof}

For convenience, we can assume, without any loss of generality, that $s_0 := 0$.
In addition, we fix a Riemannian metric $\mf{p}$ on $\mc{I}$.
Since $\mc{D}$ satisfies the GNCC, there exists $\eta \in \mc{C}^4 ( \bar{\mc{D}} )$ and $c > 0$ such that \eqref{BVP} holds.
Note that if $\eta$ is multiplied by a positive constant, then the new function still satisfies \eqref{BVP}.
Thus, we can also assume without loss of generality that
\begin{equation} \label{eta_normalize}
\eta ( \lambda ( 0 ) ) = \rho ( 0 ) \text{.}
\end{equation}
Applying \cite[Corollary 3.5]{Arick3} with a partition of unity, we see that \eqref{BVP} can be equivalently stated as follows:\ there exists $\zeta \in C^2 ( \bar{\mc{D}} )$ such that for any tangent vector $\mf{X} \in T \mc{D}$,
\begin{align}\label{etaequation}
\begin{dcases}
[ \mf{D}^2 \eta - \eta \, \bar{\mf{g}} - \zeta \, \mf{g} ] ( \mf{X}, \mf{X} ) > c \eta \, \mf{p} ( \mf{X}, \mf{X} ) &\quad \text{in $\mc{D}$,} \\
\eta > 0 &\quad \text{in $\mc{D}$,} \\
\eta = 0 &\quad \text{on $\partial \mc{D}$.}
\end{dcases} 
\end{align}	  

Fix now a constant $0 < \delta \ll 1$, whose value is to be determined later, and define
\begin{equation} \label{DefinitionTheta}
\theta: \mc{J}_\delta := \{ s \in ( s_-, s_+ ) \mid \lambda ( ( 1 + \delta ) s ) \in \mc{D} \} \rightarrow ( 0, \infty ) \text{,} \qquad \theta (s) := \eta ( \lambda ( (1 + \delta) s) ) \text{.}
\end{equation}
Note that by our normalization \eqref{eta_normalize}, we have
\begin{equation} \label{theta_bound_0}
0 < \sup_{ s \in \mc{J}_\delta } \theta (s) \leq \rho (0) \sup_{ s \in \mc{J}_\delta } \frac{ \eta ( \lambda ( (1 + \delta) s) ) }{ \eta ( \lambda (0) ) } \lesssim_{ \ms{g}, \mc{D} } \rho (0) \text{,}
\end{equation} 
since an upper bound for the ratio in \eqref{theta_bound_0} is independent of the normalization for $\eta$ and is a property of $\mc{D}$ and $\lambda (0)$.
Moreover, we use \eqref{etaequation} to compute, for any $s \in \mc{J}_\delta$,
\begin{align*}
\ddot{\theta} (s) &= ( 1 + \delta )^2 \cdot [ ( \partial_{ab}^2 \eta \circ \lambda ) \dot{\lambda}^a \dot{\lambda}^b + ( \partial_a \eta \circ \lambda ) \ddot{\lambda}^a ] ( ( 1 + \delta ) s ) \\
&= ( 1 + \delta )^2 \cdot [ \mf{D}^2 \eta ( \dot{\lambda}, \dot{\lambda} ) + \mf{D} \eta ( \mf{D}_{ \dot{\lambda} } \dot{\lambda} ) ] ( ( 1 + \delta ) s ) \text{.}
\end{align*}

Combining the above with \eqref{eta_normalize} and \eqref{etaequation}, we see that $\theta$ satisfies
\begin{align} \label{thetaequation}	
\begin{dcases}
\ddot{\theta} (s) + \mf{c}_+ (s) \, \theta (s) > \mc{N}_\theta (s) \text{,} &\quad s \in \mc{J}_\delta \text{,} \\
\theta(s) > 0 \text{,} &\quad s \in \mc{J}_\delta \text{,} \\
0 < \theta (0) = \rho (0) \text{,}
\end{dcases}
\end{align}
where $\mf{c}_+$ and $\mc{N}_\theta$ are given, for any $s \in \mc{J}_\delta$, by
\begin{align} \label{bound_N_theta}
\mf{c}_+ (s) &:= ( 1 + \delta )^2 \cdot [ ( \bar{\mf{g}} + c \mf{p} ) ( \dot{\lambda}, \dot{\lambda} ) ] ( ( 1 + \delta ) s ) \text{,} \\
\notag \mc{N}_\theta (s) &:= ( 1 + \delta )^2 \cdot [ \zeta ( \lambda ) \cdot \mf{g} ( \dot{\lambda}, \dot{\lambda} ) + \mf{D} \eta ( \mf{D}_{ \dot{\lambda} } \dot{\lambda} ) ] ( ( 1 + \delta ) s ) \text{.}
\end{align}
Moreover, since $\mf{p}$ is positive-definite, then by continuity, compactness, and \eqref{bound_N_theta}, we can choose $\delta$ sufficiently small and find a constant $p > 0$ such that
\begin{equation} \label{theta_positive}
\mf{c}_+ - \bar{\mf{g}} ( \dot{\lambda}, \dot{\lambda} ) \geq p^2 \text{.}
\end{equation}

Now, to prove Theorem \ref{Characterisation}, we split into two cases, depending on the relation between $\dot{\theta} (0)$ and $\dot{\rho} (0)$.
In particular, the conclusions of Theorem \ref{Characterisation} are consequences of the following:

\begin{lemma} \label{Characterisation_case_init}
If $\dot{\theta} (0) \leq \dot{\rho} (0)$, then \eqref{Characterisation_init} holds.
\end{lemma}

\begin{lemma} \label{Characterisation_case_term}
If $\dot{\theta} (0) \geq \dot{\rho} (0)$, then \eqref{Characterisation_term} holds.
\end{lemma}

Lemma \ref{Characterisation_case_term} is proved in Section \ref{sec.geod_subproof}; the proof of Lemma \ref{Characterisation_case_init} is completely analogous and is omitted.
Thus, the proof of Theorem \ref{Characterisation} will be complete after the next subsection.

The intuitions behind the proofs of Lemmas \ref{Characterisation_case_init} and \ref{Characterisation_case_term} are illustrated below in Figure \ref{Pic1111AandBcite}. 

\subsection{Proof of Lemma \ref{Characterisation_case_term}} \label{sec.geod_subproof}

Observe that at least one of the following scenarios must hold:
\begin{enumerate}
\item $\Lambda$ escapes from $\mc{I}$:\ $\lim_{ s \nearrow s_+ } \rho (s) = \rho_0$.

\item $\Lambda$ terminates at $\mc{I}$:\ $\lim_{ s \nearrow s_+ } \rho (s) = 0$.

\item $\lambda$ exits $\mc{D}$:\ there exists $\tau_+ \in ( 0, s_+ ]$ such that $\lim_{ s \nearrow \tau_+ } \lambda (s) \in \partial \mc{D}$.
\end{enumerate}
The goal is to show (2) and rule out (1) and (3), as well as show $\lim_{ s \nearrow s_+ } \lambda (s) \in \mc{D}$.

The proof of this is based on a Sturm comparison argument, combined with a continuity argument to control nonlinear error terms.
The key step is the following:

\begin{lemma} \label{Step1a}
$( 0, s_+ ) \subseteq \mc{J}_\delta$, and $\rho (s) < \theta (s)$ for all $s \in ( 0, s_+ )$. 
\end{lemma}

\begin{proof}
We start with the following bootstrap assumption for an arbitrary $s_1 \in ( 0, s_+ )$:
\begin{itemize}
\item (\textbf{BA}) $( 0, s_1 ) \subseteq \mc{J}_\delta$, and $\rho (s) < 2 \, \theta (s)$ for all $s \in ( 0, s_1 )$.
\end{itemize}
(Note that (\textbf{BA}) holds for $s_1$ sufficiently close to $0$, by \eqref{eta_normalize}.)
Then, by a standard continuity argument, it suffices to establish that (\textbf{BA}) implies the strictly stronger property
\footnote{More precisely, we consider the set $\mc{A} := \{ s_1 \in ( 0, s_+ ) \mid s \in \mc{J}_\delta \text{ and } \rho (s) < 2 \, \theta (s) \text{ for all } s \in ( 0, s_1 ) \}$, which is clearly closed in $( 0, s_+ )$.
If (\textbf{BA}) implies \eqref{bootstrap_conclusion}, then $\mc{A}$ is also open and hence is all of $( 0, s_+ )$.}
\begin{equation} \label{bootstrap_conclusion}
s_1 \in \mc{J}_\delta \text{,} \qquad \rho (s) < \theta (s) \text{,} \quad s \in ( 0, s_1 ) \text{.}
\end{equation}

Consider now the Wronskian,
\begin{equation} \label{Wronskian}
W (s) := \dot{\theta} (s) \rho (s) - \dot{\rho} (s) \theta (s) \text{,} \qquad s \in \mc{J}_\delta \text{.}
\end{equation}
Note that \eqref{eta_normalize} and the above imply that
\begin{equation} \label{Wronskian_s0}
W (0) = \rho (0) [ \dot{\theta} (0) - \dot{\rho} (0) ] \geq 0 \text{.}
\end{equation}
Moreover, differenting $W$ and recalling \eqref{NullGeodesicAsymp}, \eqref{thetaequation}, and \eqref{theta_positive}, we obtain that
\begin{align*}
\dot{W} &= \ddot{\theta} \rho - \ddot{\rho} \theta \\ 
&> ( \mf{c}_+ \theta + \mc{N}_\theta ) \rho - [ \bar{\mf{g}} ( \dot{\lambda}, \dot{\lambda} ) \, \rho + \mc{O} (1) ( \dot{\lambda}, \dot{\lambda} ) \, \rho^2 ] \theta \\
&\geq p^2 \, \rho \theta + \mc{N}_\theta \, \rho + \mc{O} (1) ( \dot{\lambda}, \dot{\lambda} ) \, \theta \rho^2 \text{.}
\end{align*}
Integrating the above leads, for each $s \in ( 0, s_1 )$, to
\begin{align}\label{EstimateW}
W (s) &= W (0) + \int_0^s \dot{W} ( \tau ) \, d \tau \\
\notag &> p^2 \int_0^s \rho \theta + \int_0^s \rho \mc{N}_\theta + \int_0^s \theta \rho^2 \, \mc{O} (1) ( \dot{\lambda}, \dot{\lambda} ) \\
\notag &:= \mc{W}_0 (s) + \mc{W}_1 (s) + \mc{W}_2 (s) \text{.}
\end{align}

We now control the terms in the right-hand side of \eqref{EstimateW}.
First, since $\rho (s) \simeq \theta (s) \simeq \rho (0)$ whenever $s$ is sufficiently close to $0$ (see \eqref{eta_normalize}), and since $\rho, \theta > 0$ on $( 0, s_1 )$, then
\begin{equation} \label{EstimateW0}
\mc{W}_0 (s) \geq \begin{cases} C_{ 0, 1 } s \cdot \rho^2 (0) & 0 < s \ll 1 \text{,} \\ C_{ 0, 2 } \cdot \rho^2 (0) & \text{otherwise,} \end{cases} \qquad s \in ( 0, s_1 ) \text{,}
\end{equation}
for some positive constants $C_{ 0, 1 }$, $C_{ 0, 2 }$.
Next, note that (\textbf{BA}) and \eqref{theta_bound_0} imply
\begin{equation} \label{EstimateWW}
0 \leq \rho (s) \lesssim \theta (s) \lesssim \rho (0) \text{,} \qquad s \in ( 0, s_1 ) \text{.}
\end{equation}
From the above, we obtain the following bound for some constant $C_2 > 0$:
\begin{equation} \label{EstimateW1}
\mc{W}_2 (s) \geq - C_2 \cdot \rho^3 (0) \text{,} \qquad s \in ( 0, s_1 ) \text{.}
\end{equation}

Next, we integrate the first part of \eqref{NullGeodesicAsymp} and apply \eqref{EstimateWW} to obtain
\begin{align*}
| \dot{\rho} (s) | \leq | \dot{\rho} (0) | + \int_0^s [ | \bar{\mf{g}} ( \dot{\lambda}, \dot{\lambda} ) | \rho + | \mc{O} (1) ( \dot{\lambda}, \dot{\lambda} ) | \rho^2 ] \lesssim \rho (0) \text{,} \qquad s \in ( 0, s_1 ) \text{,}
\end{align*}
as long as $\epsilon_0 > \rho (0)$ is sufficiently small.
Then, by \eqref{NullGeodesicAsymp}, \eqref{thetaequation}, \eqref{EstimateWW}, and the above,
\begin{align*}
| \mc{N}_\theta (s) | &\lesssim | \mf{g} ( \dot{\lambda}, \dot{\lambda} ) ( ( 1 + \delta ) s ) | + | \mf{D} \eta ( \mf{D}_{ \dot{\lambda} } \dot{\lambda} ) ( ( 1 + \delta ) s ) | \\
&\lesssim ( [ | \dot{\rho}^2 ( (1 + \delta) s) | + \rho^2 ( (1 + \delta) s) ] + [ | \rho^2 ( (1 + \delta) s ) | + | ( \rho \dot{\rho} ) ( (1 + \delta) s ) | ] \\
&\lesssim  \rho^2 (0) \text{,}
\end{align*}
for any $s \in ( 0, s_1 )$.
By \eqref{EstimateWW} and the above, there is some $C_1 > 0$ such that
\begin{equation} \label{EstimateW2}
\mc{W}_1 (s) \geq - C_1 \cdot \rho^3 (0) \text{,} \qquad s \in ( 0, s_1 ) \text{.}
\end{equation}

Observe that as long as $\epsilon_0$ is sufficiently small, depending on the constants $C_{ 0, 1 }$, $C_{ 0, 2 }$, $C_1$, $C_2$ (which arise from $\ms{g}$ and $\mc{D}$), then \eqref{EstimateW}, \eqref{EstimateW0}, \eqref{EstimateW1}, and \eqref{EstimateW2} yield
\[
\rho^2 (s) \left( \frac{ \theta }{ \rho } \right)' (s) = W (s) > \mc{W}_0 (s) + \mc{W}_1 (s) + \mc{W}_2 (s) > 0 \text{,} \qquad s \in ( 0, s_1 ) \text{.}
\]
Integrating the above from $s_0 = 0$ and recalling \eqref{eta_normalize} yields the second part of \eqref{bootstrap_conclusion}:
\[
\rho (s) < \theta (s) \text{,} \qquad s \in ( 0, s_1 ) \text{.}
\]
Moreover, since $\rho$ is positive on $( 0, s_1 )$, then \eqref{DefinitionTheta} and the above imply
\[
\eta ( \lambda ( ( 1 + \delta ) s_1 ) ) = \lim_{ s \nearrow s_1 } \theta ( s ) \geq \rho ( s_1 ) > 0 \text{,}
\]
which, along with \eqref{etaequation}, implies that $s_1 \in \mc{J}_\delta$.
Therefore, we have established the improved properties \eqref{bootstrap_conclusion}, which completes the bootstrap argument and hence the proof itself.
\end{proof}

Finally, by \eqref{theta_bound_0} and Lemma \ref{Step1a}, we have that
\[
\rho (s) < \theta (s) \lesssim \rho (0) \text{,} \qquad s \in ( 0, s_+ ) \text{.}
\]
Thus, by taking $\epsilon_0$ to be sufficiently small, the above rules out scenario (1).

Next, suppose (3) holds, and let $\tau_+ \in ( 0, s_+ ]$ be the smallest parameter with
\[
\lim_{ s \nearrow \tau_+ } \lambda (s) \in \partial \mc{D} \text{.}
\]
Then, the above implies $\theta ( ( 1 + \delta )^{-1} \tau_+ ) = 0$; this results in a contradiction, since Lemma \ref{Step1a} then yields $\rho ( ( 1 + \delta )^{-1} \tau_+ ) < 0$.
Therefore, (3) cannot hold, and it follows that
\[
\lim_{ s \nearrow s_+ } \lambda (s) \in \mc{D} \text{.}
\]

Furthermore, since scenarios (1) and (3) are ruled out, then (2) holds.
In particular, the above yields \eqref{Characterisation_term}, which completes the proof of Lemma \ref{Characterisation_case_term}.

\begin{figure}[ht]
\centering
\includegraphics[width=0.84\textwidth]{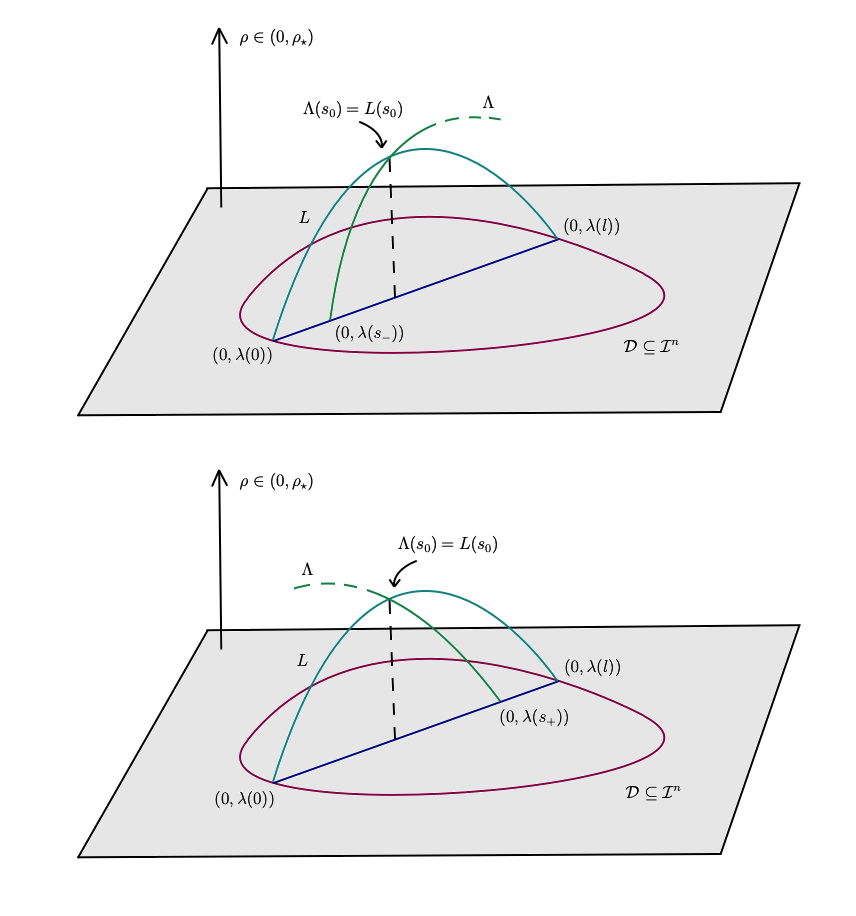}
\caption{If $\dot{\theta}(s_0) \leq \dot{\rho}(s_0)$, then $L := ( \theta, \lambda )$ lies strictly above $\Lambda := ( \rho, \lambda )$ backwards in time, so that $\rho$ must vanish before $\theta$; see the first graphic.
On the other hand, if $\dot{\theta} (s_0) \geq \dot{\rho} (s_0)$, then $L$ lies strictly above $\Lambda$ forward in time, so again $\rho$ must vanish before $\theta$; see the second illustration.}
\label{Pic1111AandBcite}
\end{figure}

\section{Carleman estimate}\label{sec.carleman}
 
In this section, we precisely state and then prove the first main theorem of this paper, Theorem \ref{MainTheorem1}, which establishes unique continuation for solutions of wave equations from the conformal boundary.
The main focus of this section, however, will be on the corresponding Carleman estimate, which is the key tool for proving unique continuation.

\subsection{Preliminaries}

The key applications of our Carleman estimates will require that they also apply to general vertical tensor fields.
As a result, we will need to define some additional concepts and notations concerning vertical and mixed tensorial quantities.
Here, we give an abridged version of the development given in \cite[Sections 2.3 and 2.4]{Arick3}.

First, it will be useful to define additional objects for treating vertical tensor fields:

\begin{definition}[Analysis of vertical tensors] \label{VerticalTensor}
Let $\mf{h}$ be a Riemannian metric on $\mc{I}$, which we also view as a $\rho$-independent vertical tensor field.
In addition, fix two integers $k, l \geq 0$.
\begin{itemize}
\item (Full contractions) Given two vertical tensor fields $\ms{A}$ and $\ms{B}$ of dual ranks $(k,l)$ and $(l,k)$, respectively, we let $\langle \ms{A}, \ms{B} \rangle$ denote the full contraction of $\ms{A}$ and $\ms{B}$, that is, the scalar field obtained by contracting all corresponding components of $\ms{A}$ and $\ms{B}$.

\item (Full duals) Given a vertical tensor field $\ms{A}$ of rank $(k, l)$, we let $\mf{h}^\star \ms{A}$ denote the full $\mf{h}$-dual of $\ms{A}$---the rank $(l, k)$ vertical tensor field obtained by raising and lowering all indices of $\ms{A}$ using (the vertical Riemannian metric) $\mf{h}$.

\item (Bundle metrics) Given vertical tensor fields $\ms{A}$ and $\ms{B}$ of the same rank, we write $\mf{h} ( \ms{A}, \ms{B} )$ to denote the full metric contraction of $\ms{A}$ and $\ms{B}$ using $\mf{h}$:
\begin{equation} \label{vertical_rmetric}
\hm ( \ms{A}, \ms{B} ) := \langle \mf{h}^\star \ms{A}, \ms{B} \rangle = \langle \ms{A}, \mf{h}^\star \ms{B} \rangle \text{.}
\end{equation}

\item (Vertical norm) Given a vertical tensor field $\ms{A}$, we define its $\mf{h}$-norm by
\begin{equation} \label{vertical_norm}
| \ms{A} |_{ \mf{h} }^2 := \hm ( \ms{A}, \ms{A} ) \text{.}
\end{equation}
\end{itemize}
\end{definition}

\begin{remark}
One difference between the present setting and \cite{Arick3} is that the latter assumed a time function $t$ on $\mc{I}$, which is naturally extended into $\mc{M}$, whereas here we have no need of such a function.
Furthermore, in \cite{Arick3}, the vertical Riemannian metric was defined using this $t$, whereas here we simply take any arbitrary $\rho$-independent metric.
\end{remark}

\begin{remark}
Direct computations (see \cite[Proposition 2.21]{Arick3}) yield the relations
\begin{equation} \label{vertical_cauchy}
| \mf{h}^\ast \ms{A} |_{ \mf{h} } = | \ms{A} |_{ \mf{h} } \text{,} \qquad | \ms{A} \otimes \ms{B} |_{ \mf{h} } \leq | \ms{A} |_{ \mf{h} } | \ms{B} |_{ \mf{h} } \text{,} \qquad | \langle \ms{A}, \ms{C} \rangle | \leq | \ms{A} |_{ \mf{h} } | \ms{C} |_{ \mf{h} } \text{,}
\end{equation}
for any vertical tensor fields $\ms{A}$, $\ms{B}$, $\ms{C}$ of appropriate ranks, in the setting of Definition \ref{VerticalTensor}.
\end{remark}

\begin{definition}[Extended vertical connection] \label{VerticalConnection}
We extend the vertical connection $\ms{D}$ to also apply in the $\rho$-direction in the following manner:\ given any vertical tensor field $\ms{A}$ of rank $( k, l )$ and any coordinate system $( U, \varphi )$ on $\mc{I}$, we define, with respect to $\varphi_\rho$-coordinates,
\[
\breve{\ms{D}}_\rho	\ms{A}_{ b_1 \dots b_l }^{ a_1 \dots a_k } := \mc{L}_\rho \ms{A}_{ b_1 \dots b_l }^{ a_1 \dots a_k } + \frac{1}{2} \sum_{i=1}^k \ms{g}^{ a_i c } \mc{L}_\rho \ms{g}_{ d c } \ms{A}_{ b_1 \dots b_l }^{ a_1 \hat{d}_i a_k } - \frac{1}{2} \sum_{j=1}^l \ms{g}^{ d c } \mc{L}_\rho \ms{g}_{ b_j c } \ms{A}_{ b_1 \hat{d}_j b_l }^{ a_1 \dots a_k } \text{,}
\]
where the multi-index notations $\smash{a_1 \hat{d}_i a_k}$ and $\smash{b_1 \hat{d}_j b_l}$ denote the sequences $a_1 \dots a_k$ and $b_1 \dots b_l$ of indices, respectively, except with $a_i$ and $b_j$ replaced by $d$.

Then, $\ms{D}$ and the above formula define a unique connection $\smash{\Dm}$ on vertical tensor fields that extends $\ms{D}$-covariant derivatives to all directions along $\mc{M}$.
\end{definition}

\begin{remark}
See \cite[Definition 2.22, Proposition 2.23]{Arick3} for more precise statements on the extended vertical connection $\smash{\Dm}$.
In practice, the following properties of $\smash{\Dm}$ are most useful:
\begin{itemize}
\item For any vertical vector field $\ms{Y}$ and vertical tensor field $\ms{A}$,
\begin{equation} \label{vertical_connection_extend}
\Dm_{ \ms{Y} } \ms{A} = \ms{D}_{ \ms{Y} } \ms{A} \text{.} 
\end{equation}

\item For any $a \in C^\infty ( \mc{M} )$ and any vector field $X$ on $\mc{M}$,
\begin{equation} \label{vertical_connection_basic}
\Dm_X a = X a \text{,} \qquad \Dm_X \ms{g} = 0 \text{,} \qquad \Dm_X \ms{g}^{-1} = 0 \text{.} 
\end{equation}

\item For any vector field $X$ on $\mc{M}$, vertical tensor fields $\ms{A}$ and $\ms{B}$, and contraction $\mc{C}$,
\begin{equation} \label{vertical_connection_product}
\Dm_X ( \ms{A} \otimes \ms{B} ) = \Dm_X \ms{A} \otimes \ms{B} + \ms{A} \otimes \Dm_X \ms{B} \text{,} \qquad \Dm_X ( \mc{C} \ms{A} ) = \mc{C} ( \Dm_X \ms{A} ) \text{.}
\end{equation}
\end{itemize}
\end{remark}

Next, we widen our scope to \emph{mixed} tensor fields, which contain both spacetime and vertical components.
In the following, we give minimally technical definitions of these objects; the reader is referred to \cite[Section 2.4]{Arick3} for more precise statements.

\begin{definition}[Mixed tensor fields] \label{MixedTensor}
Fix integers $\kappa, \lambda, k, l \geq 0$.
\begin{itemize}	
\item (Mixed tensor field) A \emph{mixed tensor field} of rank $(\kappa,\lambda;k,l)$ is, roughly, a tensor field object containing $\kappa$ contravariant and $\lambda$ covariant spacetime components, as well as $k$ contravariant and $l$ covariant vertical components.
\footnote{See \cite[Definition 2.25]{Arick3} for a more precise definition in terms of sections of vector bundles over $\mc{M}$.}

\item (Mixed connection) Let $\smash{\nablam}$ be the \emph{mixed connection}---the connection on mixed tensor fields that acts like $\nabla$ on spacetime components and $\smash{\Dm}$ on vertical components.
\footnote{More precisely, $\smash{\nablam}$ is the tensor product connection of $\nabla$ and $\smash{\Dm}$.}
\end{itemize}
\end{definition}

\begin{remark}
As a general convention, we will use bond font to denote mixed tensor fields (e.g., $\mathbf{A}$, $\mathbf{B}$).
In addition, we recall the following properties (see \cite[Proposition 2.28]{Arick3}):
\begin{itemize}
\item Any spacetime tensor field $A$ and vertical tensor field $\ms{B}$ can be viewed as a mixed tensor field.
In particular, for any vector field $X$ on $\mc{M}$,
\begin{equation} \label{mixed_connection_extend}
\nablam_X A = \nabla_X A \text{,} \qquad \nablam_X \ms{B} = \Dm_X \ms{B} \text{.}
\end{equation}

\item For any vector field $X$ on $\mc{M}$ and any mixed tensor fields $\mathbf{A}$ and $\mathbf{B}$,
\begin{equation} \label{mixed_connection_product}
\nablam_X ( \mathbf{A} \otimes \mathbf{B} ) = \nablam_X \mathbf{A} \otimes \mathbf{B} + \mathbf{A} \otimes \nablam_X \mathbf{B} \text{.}
\end{equation}

\item $\smash{\nablam}$ annihilates both $g$ and $\ms{g}$---for any vector field $X$ on $\mathcal{M}$,
\begin{equation} \label{mixed_connection_metric}
\nablam_X g = 0 \text{,} \qquad \nablam_X g^{-1} = 0 \text{,} \qquad \nablam_X \ms{g} = 0 \text{,} \qquad \nablam_X \ms{g}^{-1} = 0 \text{.}
\end{equation}
\end{itemize}
\end{remark}

\begin{definition}[Mixed operators] \label{MixedTensorOps}
Let $\mathbf{A}$ be a mixed tensor field of rank $( \kappa, \lambda; k, l )$.
\begin{itemize}
\item (Mixed differential) We can view $\smash{\nablam} \mathbf{A}$ as a mixed tensor field of rank $( \kappa, \lambda + 1; k, l )$:
\[
( \nablam \mathbf{A} ) (X) := \nablam_X \mathbf{A} \text{.}
\]

\item (Mixed Hessian) We define the \emph{mixed Hessian} $\smash{\nablam}^2 \mathbf{A}$ to be $\smash{\nablam} ( \smash{\nablam} \mathbf{A} )$, i.e., two applications of $\smash{\nablam}$ to $\mathbf{A}$.
Note this is a mixed tensor field of rank $( \kappa, \lambda + 2; k, l )$.

\item (Mixed wave operator) We define $\Boxm \mathbf{A}$ over $\mathcal{M}$ to be the $g$-trace of $\smash{\nablam}^2 \mathbf{A}$:
\begin{equation} \label{mixed_wave}
\Boxm \mathbf{A} := \tr_g ( \nablam^2 \mathbf{A} ) = g^{\mu \nu} \nablam^2_{ \mu \nu } \mathbf{A} \text{.}
\end{equation}  

\item (Mixed curvature) The \emph{mixed curvature} applied to $\mathbf{A}$ is the rank $( \kappa, \lambda + 2; k, l )$ mixed tensor field $\smash{\Rm} [ \mathbf{A} ]$ defined as the commutator of two mixed differentials of $\mathbf{A}$:
\begin{equation} \label{mixed_curv}
\Rm [\mathbf{A}] (X, Y) := \Rm_{XY} [\mathbf{A}] := \nablam_{ X Y } \mathbf{A} - \nablam_{ Y X } \mathbf{A} \text{.}
\end{equation}  
\end{itemize}
\end{definition}

\begin{remark}
Most importantly, Definition \ref{MixedTensorOps} makes sense of the wave operator applied to vertical tensor fields.
An additional benefit of using the mixed connection $\smash{\nablam}$ in our analysis is that its covariant structure allows us to apply product rule and integration by parts formulas to mixed tensor fields in the same way that we would for scalar fields.
\end{remark}

We summarize our notations in the table below: 

\begin{table}[ht]
\centering
\begin{tabular}{|c|c|c|c|c|c|c|} 
\hline
Name & Font & Connection & Hessian & Wave operator & Curvature        \\
\hline 
Spacetime tensor field & $A$ & $\nabla A$ & $\nabla^2 A$ & $\Box A$ & $R [A]$ \\ 
Vertical tensor field & $\ms{A}$ & $\smash{\Dm} \ms{A}$ & $\smash{\Dm}^2 \ms{A}$ & $\smash{\Boxm} \ms{A}$ & $ \smash{\Rm} [\ms{A}]$ \\
Mixed tensor fields & $\mathbf{A}$ & $\nablam \mathbf{A}$ & $\smash{\nablam}^2 \mathbf{A} $ & $\smash{\Boxm} \mathbf{A}$ & $\smash{\Rm} [ \mathbf{A}] $ \\ \hline 
\end{tabular}
\end{table} 

\vspace{-0.8pc}

\begin{definition}[Integration measures]
We write $d \mu_g$ to denote the volume measure induced by $g$ on $\mc{M}$.
Similarly, we write $d \mu_{ \ms{g} }$ for the volume measure induced by $\ms{g}$ on level sets of $\rho$.
\end{definition}

\subsection{The Main Estimate}

We now give a precise statement of our main Carleman estimate:

\begin{theorem}[Carleman estimate] \label{TheoremCarlemanEstimate}
Let $( \mc{M}, g )$ be a spacetime satisfying Assumptions \ref{ass.manifold} and \ref{ass.aads}, and fix a Riemannian metric $\mf{h}$ on $\mc{I}$.
Furthermore, we assume the following: 
\begin{itemize}
\item Let $\mc{D} \subseteq \mc{I} $ be open with compact closure, and assume the GNCC is satisfied on $\mc{D}$.

\item Let $\eta \in C^4 ( \bar{\mc{D}} )$ be the function satisfying \eqref{BVP}, with $\mf{h}$ as above.
In addition, let the function $f$ be defined as in \eqref{Definitionf}, with the above $\eta$ and $\mc{D}$.

\item Fix integers $k, l \geq 0$.

\item Let $\sigma \in \R$, fix $\mf{X}^\rho \in C^\infty ( \mc{I} )$ and a vector field $\mf{X}$ on $\mc{I}$, and set $X:= \mf{X}^{\rho} \partial_\rho + \mf{X}$.
\end{itemize}
Then, there exist constants $\mc{C}_0 \geq 0$ and $\mc{C}_b > 0$ (depending on $\ms{g}$, $\mc{D}$, $X$, $k$, $l$) such that
\begin{itemize}
\item for any $\kappa \in \mathbb{R}$ with 
\begin{equation} \label{carleman_ass_kappa}
2 \kappa \geq n - 1 + \mc{C}_0 \text{,} \qquad \kappa^2 - (n-2) \kappa + \sigma - (n-1) - \mc{C}_0 \geq 0 \text{,}
\end{equation}
\item and for any constants $f_\star, \lambda, p > 0$ with 
\begin{equation} \label{carleman_ass_params}
0 < f_\star \ll_{ \ms{g}, \mc{D}, X, k, l} 1 \text{,} \qquad \lambda \gg_{ \ms{g}, \mc{D}, X, k, l } |\kappa| + |\sigma| \text{,} \qquad 0 < 2 p < 1 \text{,}
\end{equation}
\end{itemize}
the following Carleman estimate holds for any vertical tensor field $\ms{\Phi}$ on $\mc{M}$ of rank $( k, l )$ such that both $\ms{\Phi}$ and $\smash{\nablam} \ms{\Phi}$ vanish on $\mc{M} \cap \{ f = f_\star \}$,
\begin{align} \label{CarlemanEstimate}
&\int_{ \mc{D} (f_\star) } e^{ - 2 \lambda p^{-1} f^{p} } f^{ n - 2 - p - 2 \kappa } | ( \Boxm + \sigma + \rho^2 \nablam_X ) \ms{\Phi} |_{ \mf{h} }^2 d \mu_g \\
\notag &\quad \qquad + \mc{C}_b \lambda^3 \limsup_{ \rho_\star \searrow 0 } \int_{ \mc{D} (f_\star) \cap \{ \rho = \rho_\star \} } [ | \Dm_{ \partial_\rho } ( \rho^{-\kappa} \ms{\Phi} ) |_{ \mf{h} }^2 + | \ms{D} ( \rho^{-\kappa} \ms{\Phi} ) |_{ \mf{h} }^2 + | \rho^{-\kappa-1} \ms{\Phi} |_{ \mf{h} }^2 ] d \mu_{ \ms{g} } \\
\notag &\quad \geq \lambda \int_{ \mc{D} (f_\star) } e^{-2 \lambda p^{-1} f^{p}} f^{ n - 2 - 2 \kappa } ( \rho^4 | \Dm_{ \partial_\rho } \ms{\Phi} |_{ \mf{h} }^2 + \rho^4 | \ms{D} \ms{\Phi} |_{ \mf{h} }^2 + f^{2p} | \ms{\Phi} |_{ \mf{h} }^2 ) d\mu_g \text{,}
\end{align}
where $\mc{D} ( f_\star )$ is the spacetime region
\begin{equation} \label{Df}
\mc{D} (f_\star) := [ ( 0, \rho_0 ] \times \mc{D} ] \cap \{ f < f_\star \} \text{.}
\end{equation}
Furthermore, if $X = 0$ and $k = l = 0$ ($\mathsf{\Phi}$ is scalar), then one can take $\mc{C}_0 = 0$ in the above.
\end{theorem}

The proof of Theorem \ref{TheoremCarlemanEstimate} is given in Section \ref{ProofCarlemanEstimate} below.
The following unique continuation property for wave equations then follows as a consequence of Theorem \ref{TheoremCarlemanEstimate}:

\begin{cor}[Unique continuation] \label{CorollaryUniqueContinuation}
Let $( \mc{M}, g )$ be a spacetime satisfying Assumptions \ref{ass.manifold} and \ref{ass.aads}.
Furthermore, we assume the following: 
\begin{itemize}
\item Let $\mc{D} \subseteq \mc{I} $ be open with compact closure, and assume the GNCC is satisfied on $\mc{D}$.

\item let $k$, $l$, $\sigma$, $X$ be as in the statement of Theorem \ref{TheoremCarlemanEstimate}.
\end{itemize}
Then, there exists $\mc{C}_0 \geq 0$ (depending on $\ms{g}$, $\mc{D}$, $X$, $k$, $l$) so that the following unique continuation property holds:\ if $\ms{\Phi}$ is a vertical tensor field on $\mc{M}$ of rank $(k,l)$ such that
\begin{itemize}
\item there exist constants $p > 0$ and $\mc{C} > 0$ satisfying
\begin{equation} \label{UC_wave_ineq}
| ( \Boxm + \sigma + \rho^2 \nablam_X ) \ms{\Phi} |_{ \mf{h} }^2 \leq \mc{C} ( \rho^{4+p} | \Dm_{ \partial_\rho } \ms{\Phi} |_{ \mf{h} }^2 + \rho^{4+p} | \ms{D} \ms{\Phi} |_{ \mf{h} }^2 + \rho^{3p} | \ms{\Phi} |_{ \mf{h} }^2 ) \text{,}
\end{equation}

\item and $\ms{\Phi}$ satisfies, for some $\kappa \in \R$ satisfying \eqref{carleman_ass_kappa}, the vanishing condition
\begin{equation} \label{UC_vanishing}
\lim_{ \rho_\star \searrow 0 } \int_{ \{ \rho_\star \} \times \mc{D} } [ | \Dm_{ \partial_\rho } ( \rho^{-\kappa} \ms{\Phi} ) |_{ \mf{h} }^2 + | \ms{D} ( \rho^{-\kappa} \ms{\Phi} ) |_{ \mf{h} }^2 + | \rho^{-\kappa-1} \ms{\Phi} |_{ \mf{h} }^2 ] d \mu_{ \ms{g} } = 0 \text{,}
\end{equation}
\end{itemize}
then there exists $f_\star > 0$ (depending on $\ms{g}$, $\mc{D}$, $X$, $k$, $l$) such that $\ms{\Phi} = 0$ on the domain $\mc{D} ( f_\star )$ defined in \eqref{Df}.
Furthermore, if $X = 0$ and $k = l = 0$, then the above holds with $\mc{C}_0 = 0$.
\end{cor}

We omit the proof of Corollary \ref{CorollaryUniqueContinuation}, as this follows from a standard argument using the Carleman estimate \eqref{CarlemanEstimate}; for details, see the proof of \cite[Theorem 3.11]{Arick2}.

\begin{remark} \label{rmk.kappa}
We note that if $X = 0$ and $k = l = 0$ (namely, $\ms{\Phi}$ is scalar), then Corollary \ref{CorollaryUniqueContinuation} holds with the optimal $\kappa$ (like in \cite{Arick1, Arick2, Arick3}) for a certain range of $\sigma$.
In particular, here we can take $\mc{C}_0$ in Corollary \ref{CorollaryUniqueContinuation}, and the smallest allowed value of $\kappa$ is given by
\begin{align*}
\kappa = \begin{cases} \frac{n-2}{2} + \sqrt{ \frac{n^2}{4} - \sigma } \text{,} &\quad \sigma \leq \frac{n^2 - 1}{4} \text{,} \\ \frac{n-1}{2} \text{,} &\quad \sigma \geq \frac{n^2 - 1}{4} \text{.} \end{cases}
\end{align*}
\end{remark}

\subsection{Proof of Theorem \ref{TheoremCarlemanEstimate}} \label{ProofCarlemanEstimate}

Throughout this subsection, we assume the hypotheses of Theorem \ref{TheoremCarlemanEstimate}.
The proof is, to a large extent, analogous to that of the Carleman estimate from \cite[Theorem 5.11]{Arick3}.
As a result, we will omit details from the parts that are close to \cite{Arick3} and focus our attention primarily on the arguments that differ.
\footnote{The main differences with \cite{Arick3} are that we use a different function $\eta$ in our definition of $f$, and that we use a different Riemannian metric $\mf{h}$ to measure the sizes of our vertical tensor fields.}

By \cite[Corollary 3.5]{Arick3} and a partition of unity argument, there is a $\zeta \in C^2 ( \bar{\mc{D}} )$ such that
\begin{equation} \label{zeta}
( \mf{D}^2 \eta - \eta \, \bar{\mf{g}} - \zeta \, \mf{g} ) ( \mf{X}, \mf{X} ) > c \eta \, \mf{h} ( \mf{X}, \mf{X} )
\end{equation}
for any tangent vector $\mf{X} \in T \mc{D}$.
We then define $\pi_\zeta$ and $w_\zeta$ as in Definition \ref{Defnmodifieddeformationtensor}, using the above $\zeta$.
In addition, we let $N$ be as in Definition \ref{DefinitionNormal}, and we set the following:
\begin{itemize}
\item We define the vector field $S$ and the real-valued function $v_\zeta$ by
\footnote{Here, $\div_g$ denotes the divergence with respect to $g$.}
\begin{equation} \label{DefinitionS}
S := f^{n-3} \nabla^\sharp f \text{,} \qquad v_\zeta := f^{n-3} w_\zeta + \frac{1}{2} \div_g S \text{.}
\end{equation}

\item We define the auxiliary quantities
\begin{equation} \label{DefinitionauxiliaryF}
F := \kappa \log f + \lambda p^{-1} f^p \text{,} \qquad \ms{\Psi} := e^{-F} \ms{\Phi} \text{.}
\end{equation}

\item Throughout the proof, we will view $F := F (f)$ as a function of $f$.
In particular, we write $F', F'', \dots$ to denote derivatives of $F$ with respect to $f$.

\item We also define the following differential operators:
\begin{equation} \label{DefinitionOperators}
\mc{L} := e^{-F} ( \Boxm + \sigma ) e^F \text{,} \qquad \mc{L}^\dagger := e^{-F} ( \Boxm + \sigma + \rho^2 \nablam_X ) e^F \text{,} \qquad \breve{S}_\zeta := \nablam_S + v_\zeta \text{.}
\end{equation}
\end{itemize}
Finally, we assume that for all $\mc{O}$-notations, the constants can depend on $\ms{g}$, $\mc{D}$, $X$, $k$, $l$.

\subsubsection{Preliminary Computations} \label{PreliminaryEstimates}

Here, we collect a number of computations and asymptotic properties that we will require later in the proof.

\begin{lemma}[Asymptotics for $S$ and $v_\zeta$] \label{LemmaAsymptoticsS}
The following hold for $S$:
\begin{equation} \label{AsymptoticsS}
S = [ f^{n-2} + \mc{O} ( f^n ) ] N \text{,} \qquad \div_g S = -2 f^{n-2} + \mc{O} ( f^n ) \text{.}
\end{equation}
Moreover, the following hold for $v_\zeta$:
\begin{align} \label{Asymptotics_vzeta}
v_\zeta = \mc{O} ( f^n ) \text{,} &\qquad \nabla_\rho v_\zeta = \mc{O} ( \rho^{-1} f^n ) \text{,} \\
\notag \ms{D} v_\zeta = \mc{O} ( \rho^{-1} f^{n+1} ) \text{,} &\qquad \Box_g v_\zeta = \mc{O} ( f^n ) \text{.}
\end{align}
\end{lemma}

\begin{proof}
The first two parts of \eqref{AsymptoticsS} follow immediately from Lemma \ref{LemmaGradientf}, \eqref{Normal}, and \eqref{DefinitionS}.
For the last part of \eqref{AsymptoticsS}, we first apply Lemmas \ref{LemmaGradientf} and \ref{LemmaWaveOperatorf} to obtain
\begin{align} \label{eql.AsymptoticsS_1}
\div_g S &= (n-3) f^{n-4} \, g ( \nabla^\sharp f, \nabla^\sharp ) + f^{n-3} \Box_g f \\
\notag &= -2 f^{n-2} + (n-1) f^n \, g ( \ms{D}^\sharp \eta, \ms{D}^\sharp \eta ) - \rho f^{n-1} \tr_{ \ms{g} } \ms{D}^2 \eta + \frac{1}{2} \rho f^{n-2} \tr_{ \ms{g} } \mc{L}_\rho \ms{g} \text{.}
\end{align}
The asymptotics for $\div_g S$ then follow from the above, along with the asymptotics
\[
\ms{D} \eta = \mc{O} (1) \text{,} \qquad \ms{D}^2 \eta = \mc{O} (1) \text{,} \qquad \mc{L}_\rho \ms{g} = \mc{O} ( \rho ) \text{,}
\]
and the observation that $\rho \lesssim f$.

Next, from \eqref{modifieddeformationtensor}, \eqref{DefinitionS}, and \eqref{AsymptoticsS}, we compute
\begin{align} \label{DefinitionauxiliarlyvzetaBETTER}
v_\zeta &= \rho f^{n-1} \zeta + \frac{1}{2} (n-1) f^n \, \ms{g} ( \ms{D}^\sharp \eta, \ms{D}^\sharp \eta ) - \frac{1}{2} \rho f^{n-1} \tr_{ \ms{g} } \ms{D}^2 \eta + \frac{1}{4} \rho f^{n-2} \tr_{ \ms{g} } \mc{L}_\rho \ms{g} \text{.}
\end{align}
The first part of \eqref{Asymptotics_vzeta} follows immediately from the above.
The remaining parts of \eqref{Asymptotics_vzeta} then follow from differentiating \eqref{DefinitionauxiliarlyvzetaBETTER} as needed, along with some tedious computations; see \cite[Lemma 5.23]{Arick3} for further details from an analogous computation.
\end{proof}

\begin{lemma}[Asymptotics for $\smash{\Rm}$] \label{LemmaVerticalCurvatureEstimates}
Let $\ms{A}$ be a vertical tensor field of rank $(k, l)$.
In addition, let $\mf{E}$ be a vector field on $\mc{D}$, and let $E = \mc{P} \mf{E}$.
\footnote{See Definition \ref{DefinitionVa} for the definition of $\mc{P}$.}
Then,
\begin{equation} \label{VerticalCurvatureEstimates}
| \Rm_{ N E } [ \ms{A} ] |_{ \mf{h} } \lesssim (k+l) \rho^2 f \, | \mf{E} |_{ \mf{h} } | \ms{A} |_{ \mf{h} } \text{.}
\end{equation}
\end{lemma}

\begin{proof}
Fix a compact coordinate system $( U, \varphi )$ on $\mc{I}$.
The main observations are the bounds
\begin{equation} \label{VerticalCurvaturePre}
| \Rm_{ab} [ \ms{A} ] |_{ 0, \varphi } \lesssim_\varphi ( k + l ) \, | \ms{A} |_{ 0, \varphi} \text{,} \qquad	| \Rm_{\rho a} [ \ms{A} ] |_{ 0, \varphi } \lesssim_\varphi (k+l) \rho \, | \ms{A} |_{ 0, \varphi } \text{,}
\end{equation}
stated in terms of $\varphi$-coordinates; for their derivations, see the proof of \cite[Lemma 5.24]{Arick3}, which is directly applicable to our current setting.
\footnote{The essential formulas for expanding $\Rm [ \ms{A} ]$ can be found in \cite[Proposition 2.32]{Arick3}.}
Now, by \eqref{Normal} and \eqref{Va},
\[
\Rm_{N E} [ \ms{A} ] = \Rm [ \ms{A} ] ( \mc{O} ( \rho ) \, \partial_\rho + \mc{O} ( \rho f ) \, \ms{D}^\sharp \eta, \, \rho f ( \mf{E} \eta ) \, \partial_\rho + \rho \mf{E} ) \text{.}
\]
The bound \eqref{VerticalCurvatureEstimates} now follows from \eqref{VerticalCurvaturePre} and the above.
\end{proof}

\begin{lemma}[Asymptotics for $\mf{h}$] \label{EstimatesVerticalRiemannianMetrich}
Let $\mf{E}$ be a vector field on $\mc{D}$, and let $E = \mc{P} \mf{E}$.
Then,
\begin{align} \label{Estimates_h}
\nablam_\rho \mf{h} = \mc{O} ( \rho ) \text{,} &\qquad \nablam_N \mf{h} = \mc{O} ( \rho f ) \text{,} \qquad \nablam_E \mf{h} = \mc{O} ( \rho ) ( \mf{E} ) \text{,} \qquad \Boxm \mf{h} = \mc{O} ( \rho^2 ) \text{.}
\end{align}
\end{lemma}

\begin{proof}
First, we clearly have that
\begin{equation} \label{eql.Estimates_h_1}
\mf{h} = \mc{O} (1) \text{,} \qquad \Dm \mf{h} = \mc{O} (1) \text{,} \qquad \Dm^2 \mf{h} = \mc{O} (1) \text{.}
\end{equation}
For convenience, in the upcoming computations, we will state quantities with respect to a compact coordinate system $( U, \varphi )$ on $\mc{I}$.
Since $\mc{L}_\rho \mf{h} = 0$ by definition, we derive that
\begin{equation} \label{eql.Estimates_h_2}
\nablam_\rho \mf{h}_{ab} = - \frac{1}{2} \ms{g}^{dc} \mc{L}_\rho \ms{g}_{ac} \mf{h}_{db} - \frac{1}{2} \ms{g}^{dc} \mc{L}_\rho \ms{g}_{bc} \mf{h}_{ad} \text{.}
\end{equation}
The first part of \eqref{Estimates_h} follows from Lemma \ref{LemmaMetric} and the above.
Next, by \eqref{Normal} and \eqref{Va},
\[
\nablam_N \mf{h} = \mc{O} ( \rho ) \, \nablam_\rho \mf{h} + \mc{O} ( \rho f ) \, \ms{D}_{ \ms{D}^\sharp \eta } \mf{h} \text{,} \qquad \nablam_E \mf{h} = \rho f ( \mf{E} \eta ) \, \nablam_\rho \mf{h} + \rho \ms{D}_{ \mf{E} } \mf{h} \text{.}
\]
The second and third parts of \eqref{Estimates_h} now follow from the first part and from the above.

Now, we unwind the definitions of $\smash{\nablam}^2$ and $\ms{D}^2$ (see Definitions \ref{MixedTensor} and \ref{MixedTensorOps}) to obtain
\begin{align} \label{eql.Estimates_h_3}
\nablam^2_{ab} \mf{h} &= \ms{D}_{ab} \mf{h} - \Dm_{ \nabla_a \partial_b - \ms{D}_a \partial_b } \mf{h} \\
\notag &= \ms{D}_{ab} \mf{h} + \left( \frac{1}{2} \mc{L}_\rho \ms{g}_{ab} - \rho^{-1} \ms{g}_{ab} \right) \Dm_\rho \mf{h} \\
\notag &= \mc{O} (1)_{ab} \text{,}
\end{align}
where we also used the first part of \eqref{Estimates_h} and \eqref{eql.Estimates_h_1}.
Moreover, differentiating \eqref{eql.Estimates_h_2} yields
\begin{equation} \label{eql.Estimates_h_4}
\Dm_\rho ( \Dm_\rho \mf{h} )_{ a b } = - \frac{1}{2} \ms{g}^{dc} \, \Dm_\rho ( \mc{L}_\rho \ms{g}_{ac} \mf{h}_{db} + \mc{L}_\rho \ms{g}_{bc} \mf{h}_{ad} ) = \mc{O} (1)_{ a b } \text{,}
\end{equation}
where we again recalled Lemma \ref{LemmaMetric} and the first part of \eqref{Estimates_h}.
Unwinding the definitions of $\smash{\nablam}^2$ and $( \smash{\Dm}_\rho )^2$, and then applying the first part of \eqref{Estimates_h} and \eqref{eql.Estimates_h_4}, we obtain
\begin{align} \label{eql.Estimates_h_5}
\nablam^2_{ \rho \rho } \mf{h} &= \Dm_\rho ( \Dm_\rho \mf{h} ) - \nablam_{ \nabla_\rho \partial_\rho } \mf{h} \\
\notag &= \Dm_\rho ( \Dm_\rho \mf{h} ) + \rho^{-1} \nablam_\rho \mf{h} \\
\notag &= \mc{O} (1)
\end{align}
Finally, combining Lemma \ref{LemmaMetric}, \eqref{eql.Estimates_h_3}, and \eqref{eql.Estimates_h_5} yields
\[
\Boxm \mf{h} = g^{ \rho \rho } \nablam_{ \rho \rho } \mf{h} + g^{ a b } \nablam_{ a b } \mf{h} = \mc{O} ( \rho^2 ) \text{,}
\]
which is precisely the last part of \eqref{Estimates_h}.
\end{proof}

\begin{remark}
Lemma \ref{EstimatesVerticalRiemannianMetrich} represents one part of the proof of Theorem \ref{TheoremCarlemanEstimate} that differs from \cite{Arick3}.
This is due to our use of a different Riemannian metric to measure vertical tensor fields.
However, our metric $\mf{h}$ satisfies the same asymptotic estimates as its counterpart $\ms{h}$ in \cite{Arick3}, so this does not significantly affect other portions of the proof compared to \cite{Arick3}.
\end{remark}

\begin{lemma}[Expansion of $\mc{L}$] \label{LemmaEstimatesConjugatedWaveOperator}
The following formula holds for $\mc{L}$:
\begin{align}\label{DefinitionMathcalLBetter}
\mc{L} &= \Boxm + 2 F' f^{-(n-3)} \nablam_S + \mc{A} \text{,} \\
\notag \mc{A} &= [ ( F' )^2 + F'' ] \, g ( \nabla^\sharp f, \nabla^\sharp f ) + F' \Box_g f + \sigma \text{.}
\end{align}
In addition, we have the following asymptotic formulas:
\begin{align} \label{Asymptotics_A}
\mc{A} &= ( \kappa^2 - n \kappa + \sigma ) + ( 2 \kappa - n + p ) \lambda f^p + \lambda^2 f^{2p} + \lambda^2 \mc{O} ( f^2 ) \text{,} \\
\notag - \frac{1}{2} \div_g ( \mc{A} S ) &= ( \kappa^2 - n \kappa + \sigma ) f^{n-2} + \left( 1 - \frac{p}{2} \right) ( 2 \kappa - n + p ) \lambda f^{ n - 2 + p } \\
\notag &\qquad + (1-p) \lambda^2 f^{ n - 2 + 2p } + \lambda^2 \mc{O} ( f^n ) \text{.}
\end{align}
\end{lemma}

\begin{proof}
The formulas \eqref{DefinitionMathcalLBetter} follow from direct computations and the definitions \eqref{DefinitionS}--\eqref{DefinitionOperators}.
The asymptotic relations \eqref{Asymptotics_A} are also direct computations; these formulas are the same as \cite[Lemma 5.26]{Arick3}, and their derivations are identical, hence we omit the details here.
\end{proof}

\subsubsection{The GNCC and Pseudoconvexity} \label{RolePseudoConvexity}

Next, we establish the crucial (and new compared to \cite{Arick3}) estimate in the proof of Theorem \ref{TheoremCarlemanEstimate}---that the GNCC implies the level sets of $f$ are pseudoconvex.
This is captured in the following lemma, which shows that the modified deformation tensor $\pi_\zeta$ is positive-definite in the directions tangent to the level sets of $f$:

\begin{lemma}[Lower bound for $\pi$] \label{LemmaLowerBoundModifiedDeformationTensor}
Given a spacetime tangent vector $Z \in T \mc{M}$, we write
\footnote{Note that $\mc{P}$ is invertible, so we can define $\mf{E} := \mc{P}^{-1} E$ in the above.}
\begin{equation} \label{Z_pseudoconvex}
Z := Z^N N + E \text{,} \qquad E := \mc{P} \mf{E} \text{,}
\end{equation}
where $E$ is normal to $N$ (and hence tangent to level sets of $f$).\footnote{Observe that $N$ is spacelike when $f$ is sufficiently small.}
Then, there exists a constant $K > 0$ such that the following inequality holds on $\mc{D} ( f_\star )$:
\begin{equation} \label{Carleman_pseudoconvex}
\pi (Z, Z) \geq - [ (n-1) f^{n-2} + \mc{O} ( f^n ) ] \, ( Z^N )^2 + K \rho^2 f^{n-2} \, \mf{h} ( \mf{E}, \mf{E} ) \text{.}
\end{equation}
\end{lemma}

\begin{proof}
First, we recall Lemma \ref{Lemmapi} and \eqref{Z_pseudoconvex}, which imply
\begin{align} \label{Z_pseudoconvex_1}
\pi ( Z, Z ) &= ( Z^N )^2 \, \pi ( N, N ) + 2 Z^N \, \pi ( N, E ) + \pi ( E, E ) \\
\notag &= - [ ( n - 1 ) f^{n-2} + \mc{O} ( f^n ) ] ( Z^N )^2 + \mc{O} ( \rho f^n ) ( \mf{E} ) \, Z^N \\
\notag &\qquad + \rho f^{n-1} ( \mf{D}^2 \eta - \eta \bar{\mf{g}} - \zeta \mf{g} ) ( \mf{E}, \mf{E} ) + \mc{O} ( \rho^2 f^{n-1} ) ( \mf{E}, \mf{E} ) \text{.}
\end{align}
We apply \eqref{zeta} to treat the third term on the right-hand side of \eqref{Z_pseudoconvex_1}:
\begin{align} \label{Z_pseudoconvex_2}
\rho f^{n-1} ( \mf{D}^2 \eta - \eta \bar{\mf{g}} - \zeta \mf{g} ) ( \mf{E}, \mf{E} ) &\geq c \rho f^{n-1} \eta \, \mf{h} ( \mf{E}, \mf{E} ) \\
\notag &= c \rho^2 f^{n-2} \, \mf{h} ( \mf{E}, \mf{E} ) \text{.}
\end{align}
Moreover, for the second term, the Cauchy inequality implies
\[
\mc{O} ( \rho f^n ) ( \mf{E} ) \, Z^N \gtrsim \mc{O} ( f^n ) \, ( Z^N )^2 + \mc{O} ( \rho^2 f^n ) ( \mf{E}, \mf{E} ) \text{.}
\]
Combining \eqref{Z_pseudoconvex_1} and \eqref{Z_pseudoconvex_2} with the above, we obtain
\[
\pi ( Z, Z ) \geq - [ ( n - 1 ) f^{n-2} + \mc{O} ( f^n ) ] ( Z^N )^2 + c \rho^2 f^{n-2} \, \mf{h} ( \mf{E}, \mf{E} ) + \mc{O} ( \rho^2 f^{n-1} ) ( \mf{E}, \mf{E} ) \text{,}
\]
and the desired \eqref{Carleman_pseudoconvex} follows from the above by taking $f_\ast$ small enough.
\end{proof}

\subsubsection{Pointwise Estimate for $\ms{\Psi}$}

We now establish a pointwise estimate for the conjugated operator $\mc{L}^\dagger$.
This will serve as a precursor for our main Carleman estimate \eqref{CarlemanEstimate}.
The first step is to derive an estimate comparing the operators $\mc{L}$ and $\mc{L}^\dagger$:

\begin{lemma}[Relation between $\mc{L}$ and $\mc{L}^\dagger$] \label{LemmaEstimatesCurrentJ}
Let $\Psi$ be as in \eqref{DefinitionauxiliaryF}, and define
\begin{equation} \label{DefinitionCurrentJ}
J ( \ms{\Psi} ) := \frac{1}{2} \langle \mc{L} ( \mf{h}^\star \ms{\Psi} ), \breve{S}_\zeta ( \ms{\Psi} ) \rangle + \frac{1}{2} \langle \mc{L} ( \ms{\Psi} ),\breve{S}_\zeta ( \mf{h}^\star \ms{\Psi} ) \rangle \text{.}
\end{equation}
Then, there exists $\mc{C}_0 > 0$ (depending only on $\ms{g}$, $\mc{D}$, $X$, $k$, $l$) such that on $\mc{D} ( f_\star )$,
\begin{align} \label{EquationCurrentJ}
| J ( \ms{\Psi} ) | &\leq \lambda^{-1} f^{n-2-p} | \mc{L}^\dagger \ms{\Psi} |_{ \mf{h} }^2 + f^{n-2} | \nablam_N \ms{\Psi} |_{ \mf{h} }^2 \left( \mc{C}_0 + \frac{1}{2} \lambda f^p \right) \\
\notag &\qquad + \frac{1}{4} K \rho^2 f^{n-2} \sum_{A=1}^n | \nablam_{ E_A } \ms{\Psi} |_{ \mf{h} }^2 + \lambda \, \mc{O} ( f^{n-1} ) ( | \nablam_N \ms{\Psi} |_{ \mf{h} }^2 + | \ms{\Psi} |_{ \mf{h} }^2 ) \text{,}
\end{align}
where the constant $K > 0$ is as in the statement of Lemma \ref{LemmaLowerBoundModifiedDeformationTensor}, and where $E_A := \mc{P} \mf{E}_A$ ($1 \leq A \leq n$) is a local orthonormal frame on the level sets of $f$ satisfying
\[
| \mf{E}_A |_{ \mf{h} } \simeq 1 \text{,} \qquad 1 \leq A \leq n \text{.}
\]
Furthermore, if $X = 0$ and $k = l = 0$, then \eqref{EquationCurrentJ} holds with $\mc{C}_0 = 0$.
\end{lemma}

\begin{proof}
First, we obtain from \eqref{DefinitionOperators} the relation
\begin{equation} \label{DefinitionMathcalModifiedLBetter}
\mc{L}^\dagger = \mc{L} + \nablam_Y + F' (Y f) \text{,} \qquad Y := \rho^2 X \text{.}
\end{equation}
Expanding \eqref{DefinitionCurrentJ} using the product rule and the above, we obtain
\footnote{The notation $\smash{\hm}$ in $J_4$ was defined in \eqref{vertical_rmetric}. Note that we can think of $\smash{\hm}$ as a rank $( 2l, 2k )$ vertical tensor field, hence we can make sense of the mixed derivatives of $\hm$ in \eqref{DefinitionJBetter}.}
\begin{align} \label{DefinitionJBetter}
J ( \ms{\Psi} ) &= \hm ( \mc{L}^\dagger \ms{\Psi}, \breve{S}_\zeta \ms{\Psi} ) - F' (Yf) \, \hm ( \ms{\Psi}, \breve{S}_\zeta \ms{\Psi} ) - \hm ( \nablam_Y \ms{\Psi}, \breve{S}_\zeta \ms{\Psi} ) \\
\notag &\qquad + g^{ \mu \nu } \nablam_\mu \hm ( \nablam_\nu \ms{\Psi}, \breve{S}_\zeta \ms{\Psi} ) + F' f^{-(n-3)} \, \nablam_S \hm ( \ms{\Psi}, \breve{S}_\zeta \ms{\Psi} ) + \frac{1}{2} \Boxm \hm ( \ms{\Psi}, \breve{S}_\zeta \ms{\Psi} ) \\
\notag &\qquad + \frac{1}{2} \nablam_S \hm ( \mc{L}^\dagger \ms{\Psi}, \ms{\Psi} ) - \frac{1}{2} \nablam_S \hm ( \ms{\Psi}, \nablam_Y \ms{\Psi} ) - \frac{1}{2} F' (Yf) \, \nablam_S \hm ( \ms{\Psi}, \ms{\Psi} ) \\
\notag &:= J_1 + \dots + J_9 \text{,}
\end{align}
where the indices $\mu$, $\nu$ in $J_4$ are with respect to an arbitrary coordinate system on $\mc{M}$. 

The next step is to estimate each of the $J_i$'s in \eqref{DefinitionJBetter}, using Lemmas \ref{LemmaAsymptoticsS} and \ref{EstimatesVerticalRiemannianMetrich}.
While this is a lengthy process, the estimates obtained are analogous to those found in the proof of \cite[Lemma 5.27]{Arick3}.
\footnote{Strictly speaking, the estimates here and in \cite{Arick3} are not quite identical, as different Riemannian metrics are used. However, the estimates for $\mf{h}$ in Lemma \ref{EstimatesVerticalRiemannianMetrich} are identical to those in \cite[Lemma 5.25]{Arick3}, so the same proof holds in both settings. Moreover, note our frames $\{ E_A \}$ correspond to the frames $\{ V, E_A \}$ in \cite{Arick3}.}
Hence, here we omit the details and list the resulting bounds:
\begin{align*}
| J_1 | &\leq \lambda^{-1} \left[ \frac{1}{2} f^{n-2-p} + \mc{O} (f^n) \right] | \mc{L}^\dagger \ms{\Psi} |_{ \mf{h} }^2 + \frac{1}{2} \lambda f^{n-2+p} | \nablam_N \ms{\Psi} |_{ \mf{h} }^2 \\
&\qquad + \lambda \, \mc{O} (f^n) \, ( | \nablam_N \ms{\Psi} |_{ \mf{h} }^2 + | \ms{\Psi} |_{ \mf{h} }^2 ) \text{,} \\
| J_2 | &\leq \lambda \, \mc{O} ( \rho f^{n-2} ) \, ( | \nablam_N \ms{\Psi} |_{ \mf{h} }^2 + | \ms{\Psi} |_{ \mf{h} }^2 ) \text{,} \\
| J_3 | &\leq ( | \mf{X}^\rho |^2 + | \mf{X} |_{ \mf{h} }^2 ) \, \mc{O} ( f^{n-2} ) | \nablam_N \ms{\Psi} |_{ \mf{h} }^2 + \frac{1}{16} K \rho^2 f^{n-2} \sum_{ A = 1 }^n | \nablam_{ E_A } \ms{\Psi} |_{ \mf{h} }^2 \\
&\qquad + \mc{O} ( f^{n-1} ) ( | \nablam_N \ms{\Psi} |_{ \mf{h} }^2 + | \ms{\Psi} |_{ \mf{h} }^2 ) \text{,} \\
| J_4 | &\leq (k+l)^2 \mc{O} ( f^{n-2} ) \, | \nablam_N \ms{\Psi} |_{ \mf{h} }^2 + \frac{1}{16} K \rho^2 f^{n-2} \sum_{ A = 1 }^n | \nablam_{ E_A } \ms{\Psi} |_{ \mf{h} }^2 \\
&\qquad + \mc{O} ( f^n ) \, ( | \nablam_N \ms{\Psi} |_{ \mf{h} }^2 + | \ms{\Psi} |_{ \mf{h} }^2 ) \text{,} \\
| J_5 | &\leq \lambda \, \mc{O} ( \rho f^{n-1} ) \, ( | \nablam_N \ms{\Psi} |_{ \mf{h} }^2 + | \ms{\Psi} |_{ \mf{h} }^2 ) \text{,} \\
| J_6 | &\leq \mc{O} ( \rho^2 f^{n-2} ) \, ( | \nablam_N \ms{\Psi} |_{ \mf{h} }^2 + | \ms{\Psi} |_{ \mf{h} }^2 ) \text{,} \\
| J_7 | &\leq \mc{O} ( \rho f^{n-1} ) \, ( | \mc{L}^\dagger \ms{\Psi} |_{ \mf{h} }^2 + | \ms{\Psi} |_{ \mf{h} }^2 ) \text{,} \\
| J_8 | &\leq \mc{O} (\rho^3 f^{n-1}) \left( \| \nablam_N \ms{\Psi} |_{ \mf{h} }^2 + \sum_{ A = 1 }^n | \nablam_{ E_A } \ms{\Psi} |_{ \mf{h} }^2 \right) + \mc{O} ( \rho f^{n-1} ) \, | \ms{\Psi} |_{ \mf{h} }^2 \text{,} \\
| J_9 | &\leq \lambda \mc{O} ( \rho^2 f^{n-1} ) \, | \ms{\Psi} |_{ \mf{h} }^2 \text{.}
\end{align*}
Combining \eqref{DefinitionJBetter} and the above, while taking $f_\star$ sufficiently small, yields \eqref{EquationCurrentJ}, with
\[
\mc{C}_0 \lesssim_{ \ms{g}, \mc{D}, X, k, l } ( k + l )^2 + | \mf{X}^\rho |^2 + | \mf{X} |_{ \mf{h} }^2 \text{.}
\]
In particular, $\mc{C}_0 = 0$ whenever $X = 0$ and $k = l = 0$.
\end{proof}

From now on, we set $\mc{C}_0$ as in Lemma \ref{LemmaEstimatesCurrentJ}, so that \eqref{carleman_ass_kappa} holds with this choice of $\mc{C}_0$.

\begin{lemma}[Pointwise estimate for $\mc{L}^\dagger$] \label{LemmaPointwiseEstimateModifiedConjugatedWaveOperatorPsi}
There exists a constant $\mc{C} > 0$ (depending on $\ms{g}$, $\mc{D}$, $X$, $k$, $l$) such that the following bound holds on $\mc{D} ( f_\star )$:
\begin{align} \label{Pointwise_Psi}
\lambda^{-1} f^{n-2-p} | \mc{L}^\dagger \ms{\Psi} |_{ \mf{h} }^2 &\geq \mc{C} \rho^2 f^{n-2} \left( | \nablam_N \ms{\Psi} |_{ \mf{h} }^2 + \sum_{A=1}^n | \nablam_{ E_A } \ms{\Psi} |_{ \mf{h} }^2 \right) \\
\notag &\qquad + \frac{1}{4} \lambda^2 f^{n-2+2p} | \ms{\Psi} |_{ \mf{h} }^2 + \div_g P \text{,}
\end{align}
where $E_A$, $1 \leq A \leq n$, denotes the same local orthonormal frames as in Lemma \ref{LemmaEstimatesCurrentJ}, and where $P$ is the spacetime 1-form given, with respect to any coordinate system on $\mc{M}$, by
\begin{align} \label{Pointwise_P}
P_\beta &:= \frac{1}{2} [ \langle \nablam_S ( \mf{h}^\star \ms{\Psi} ), \nablam_\beta \ms{\Psi} \rangle + \langle \nablam_S \ms{\Psi}, \nablam_\beta ( \mf{h}^\star \ms{\Psi} ) \rangle - g_{ \alpha \beta } g^{ \mu \nu } S^\beta \langle \nablam_\mu ( \mf{h}^\star \ms{\Psi} ), \nablam_\nu \mathsf{\Psi} \rangle ] \\
\notag &\qquad + \frac{1}{2} v_\zeta [ \langle \nablam_\beta ( \mf{h}^\star \mathsf{\Psi} ), \ms{\Psi} \rangle + \langle \nablam_\beta \ms{\Psi}, \mf{h}^\star \ms{\Psi} \rangle ] - \frac{1}{2} \nabla_\beta v_\zeta \, |\mathsf{\Psi}|_{\breve{\mathsf{h}}(\rho)}^2 + \frac{1}{2} g_{ \alpha \beta } S^\beta \mc{A} \, | \ms{\Psi} |_{ \mf{h} }^2 \\
\notag &\qquad + ( 2 \kappa - n + 1 - \mc{C}_0 ) f^{n-3} \nabla_{\beta} f \, | \ms{\Psi} |_{ \mf{h} }^2 + \frac{1}{2} (2-p) \lambda f^{n-3+p} \nabla_{\beta} f \, | \ms{\Psi} |_{ \mf{h} }^2 \text{.}
\end{align}
Furthermore, there exists $\mc{C}_b > 0$ (depending on $\ms{g}$, $\mathcal{D}$, $X$, $k$, $l$) such that
\begin{equation} \label{Pointwise_Prho}
P ( \rho \partial_\rho ) \leq \mc{C}_b f^{n-2} \rho^2 \, ( | \Dm_\rho \ms{\Psi} |_{ \mf{h} }^2 + | \Dm \ms{\Psi} |_{ \mf{h} }^2 ) + \mc{C}_b \lambda^2 f^{n-2} | \ms{\Psi} |_{ \mf{h} }^2 \text{.}
\end{equation}
\end{lemma}

\begin{proof}
Let $J ( \Psi )$ be as in \eqref{DefinitionCurrentJ}.
We then expand $J ( \Psi )$ using \eqref{DefinitionS} and Lemma \ref{LemmaEstimatesConjugatedWaveOperator}.
After some extensive computations (we omit the details here, but this mirrors the process shown within the proof of \cite[Lemma 5.28]{Arick3}), we obtain the following identity for $\Psi$,
\begin{align} \label{ExpressionJForLemma517}
J ( \ms{\Psi} ) &= \div_g ( P^S + P^Q ) + 2 F' f^{ -(n-3) } \, | \nablam_S \ms{\Psi} |_{ \mf{h} }^2 + g^{ \alpha \mu } g^{ \beta \nu } ( \pi_\zeta )_{ \alpha \beta } \, \hm ( \nablam_\mu \ms{\Psi}, \nablam_\nu \ms{\Psi} ) \\
\notag &\qquad + \left[ \mc{A} v_\zeta - \frac{1}{2} \div_g ( \mc{A} S ) + \frac{1}{2} \Boxm v_\zeta \right] \, | \ms{\Psi} |_{ \mf{h} }^2 + I_S + I_\zeta + I_R + I_\pi \text{,}
\end{align}
where the (spacetime) $1$-forms $P^S$, $P^Q$ and the scalar quantities $I_S$, $I_\zeta$, $I_R$, $I_\pi$ are given by
\begin{align} \label{ExpressionJ_error}
P_\nu^S &:= \frac{1}{2} g_{ \nu \mu } S^\mu \mc{A} \, | \ms{\Psi} |_{ \mf{h} }^2 \text{,} \\
\notag P_\nu^Q &:= Q_{ \mu \nu } S^\mu + \frac{1}{2} v_{\zeta} [ \langle \nablam_\nu ( \mf{h}^\star \ms{\Psi} ), \ms{\Psi} \rangle + \langle \nablam_\nu \ms{\Psi}, \mf{h}^\star \ms{\Psi} \rangle ] - \frac{1}{2} \nabla_\nu v_\zeta \, | \ms{\Psi} |_{ \mf{h} }^2 \text{,} \\
\notag Q_{ \mu \nu } &:= \frac{1}{2} [ \langle \nablam_\mu ( \mf{h}^\star \ms{\Psi} ), \nablam_\nu \ms{\Psi} \rangle + \langle \nablam_\mu \ms{\Psi}, \nablam_\nu ( \mf{h}^\star \ms{\Psi} ) \rangle - g_{ \mu \nu } g^{ \alpha \beta } \langle \nabla_\alpha ( \mf{h}^\star \ms{\Psi} ), \nablam_\beta \ms{\Psi} \rangle ] \text{,} \\
\notag I_S &:= 2 F' f^{ -(n-3) } \, \nablam_S \hm ( \ms{\Psi}, \nablam_S \ms{\Psi} ) \text{,} \\
\notag I_\zeta &:= F f^{ -(n-3) } v_\zeta \, [ \langle \nablam_S ( \mf{h}^\star \ms{\Psi} ), \ms{\Psi} \rangle + \langle \nablam_S \ms{\Psi}, \mf{h}^\star \ms{\Psi} \rangle ] \text{,} \\
\notag I_R &:= \frac{1}{2} g^{ \alpha \beta } S^\nu \, [ \langle \Rm_{\nu \alpha} [ \mf{h}^\star \ms{\Psi} ], \nablam_\beta \ms{\Psi} \rangle + \langle \Rm_{ \nu \alpha } [ \ms{\Psi} ], \nablam_\beta ( \mf{h}^\star \ms{\Psi} ) ] \text{,} \\
\notag I_\pi &:= g^{ \alpha \mu } g^{ \beta \nu } ( \pi_\zeta )_{ \alpha \beta } \, \nablam_\mu \hm ( \ms{\Psi}, \nablam_\nu \ms{\Psi}).
\end{align}

Next, we obtain asymptotic formulas for various terms in \eqref{ExpressionJForLemma517} and \eqref{ExpressionJ_error}.
Again, this process mirrors that in the proof of \cite[Lemma 5.28]{Arick3}, so the reader is referred there for further details.
\footnote{Again, while we use a different Riemannian metric than in \cite{Arick3}, both metrics satisfy the same estimates.
Furthermore, our frames $\{ E_A \}$ here again correspond to the frames $\{ V, E_A \}$ in \cite{Arick3}.}
For the second term in \eqref{ExpressionJForLemma517}, we use \eqref{DefinitionauxiliaryF} and \eqref{AsymptoticsS} to infer
\begin{equation} \label{ExpressionJ_1}
2 F' f^{ -(n-3) } | \nablam_S \ms{\Psi} |_{ \mf{h} }^2 = [ 2 \kappa f^{n-2} + 2 \lambda f^{n-2+p} + \mc{O} ( f^n ) ] \, | \nablam_N \ms{\Psi} |_{ \mf{h} }^2 \text{.}
\end{equation}
For the fourth term in \eqref{ExpressionJForLemma517}, we use \eqref{Asymptotics_vzeta} and \eqref{Asymptotics_A} to obtain
\begin{align} \label{ExpressionJ_2}
\mc{A} v_\zeta - \frac{1}{2} \div_g ( \mc{A} S ) + \frac{1}{2} \Boxm v_\zeta &= 
( \kappa^2 - n \kappa + \sigma ) f^{n-2} + \left( 1 - \frac{p}{2} \right) ( 2 \kappa - n  + p) \lambda f^{n-2+p} \\
\notag &\qquad + (1-p) \lambda^2 f^{n-2+2p} + \lambda^2 \, \mc{O} ( f^n ) \text{.}
\end{align}

For the third term in \eqref{ExpressionJForLemma517}, we expand the $g$-contractions in terms of the frame $\{ N, E_A \}$:
\[
g^{ \alpha \mu } g^{ \beta \nu } ( \pi_\zeta )_{ \alpha \beta } \, \hm ( \nablam_\mu \ms{\Psi}, \nablam_\nu \ms{\Psi} ) = \sum_{ Z_1, Z_2 \in \{ N, E_A \} } \mu_{ Z_1 } \mu_{ Z_2 } \, \pi_\zeta ( Z_1, Z_2 ) \, \hm ( \nablam_{ Z_1 } \ms{\Psi}, \nablam_{ Z_2 } \ms{\Psi} ) \text{,}
\]
where each $\mu_{ Z_i } = \pm 1$, depending on whether $Z_i$ is timelike or spacelike.
From the above, we fully expand the factors $\smash{ \hm ( \nablam_{ Z_1 } \ms{\Psi}, \nablam_{ Z_2 } \ms{\Psi} ) }$ using an $\mf{h}$-orthonormal frame, we apply Lemma \ref{LemmaLowerBoundModifiedDeformationTensor}, and we then recall the assumptions $| \mf{E}_A |_{ \mf{h} } \simeq 1$ in order to deduce
\begin{align} \label{ExpressionJ_3}
g^{ \alpha \mu } g^{ \beta \nu } ( \pi_\zeta )_{ \alpha \beta } \, \hm ( \nablam_\mu \ms{\Psi}, \nablam_\nu \ms{\Psi} ) &\geq - [ (n-1) f^{n-2} + \mc{O} ( f^n ) ] \, | \nablam_N \ms{\Psi} |_{ \mf{h} }^2 \\
\notag &\qquad + K \rho^2 f^{n-2} \sum_{ A=1 }^n | \nablam_{ E_A } \ms{\Psi} |_{ \mf{h} }^2 \text{,}
\end{align}
where $K > 0$ is as in the statement of Lemma \ref{LemmaLowerBoundModifiedDeformationTensor}.

The error terms $I_S$, $I_\zeta$, $I_R$ and $I_\pi$ contain no leading-order terms and are controlled using Lemmas \ref{LemmaAsymptoticsS}--\ref{EstimatesVerticalRiemannianMetrich} by following the process in the proof of \cite[Lemma 5.28]{Arick3}:
\footnote{There are no $N$-derivatives of $\ms{\Psi}$ in the estimate for $I_R$, since $\smash{\Rm}_{ N N } [ \ms{\Psi} ] = 0$.}
\begin{align} \label{ExpressionJ_4}
| I_S | &\leq \lambda \, \mc{O} ( \rho f^{n-1} ) \, ( | \nablam_N \ms{\Psi} |_{ \mf{h} }^2 + | \ms{\Psi} |_{ \mf{h} }^2 ) \text{,} \\
\notag | I_\zeta | &\leq \lambda \, \mc{O} ( f^n ) \, ( | \nablam_N \ms{\Psi} |_{ \mf{h} }^2 + | \ms{\Psi} |_{ \mf{h} }^2 ) \text{,} \\
\notag | I_R | &\leq \mc{O} ( \rho^3 f^{n-1} ) \sum_{ A=1 }^n | \nablam_{ E_A } \ms{\Psi} |_{ \mf{h} }^2 + \mc{O} ( \rho f^{n-1} ) \, | \ms{\Psi} |_{ \mf{h} }^2 \text{,} \\
\notag | I_\pi | &\leq \mc{O} ( \rho f^{n-1} ) \, ( | \nablam_N \ms{\Psi} |_{ \mf{h} }^2 + | \ms{\Psi} |_{ \mf{h} }^2 ) + \mc{O} ( \rho^3 f^{n-1} ) \sum_{ A=1 }^n | \nablam_{ E_A } \ms{\Psi} |_{ \mf{h} }^2 \text{.}
\end{align}
Combining \eqref{ExpressionJForLemma517}--\eqref{ExpressionJ_4} with the estimate \eqref{EquationCurrentJ} for $J ( \ms{\Psi} )$, we obtain
\begin{align} \label{ExpressionJ_10}
\lambda^{-1} f^{n-2-p} | \mc{L}^\dagger \ms{\Psi} |_{ \mf{h} }^2 &\geq \left[ ( 2 \kappa - n + 1 - \mc{C}_0 ) f^{n-2} + \frac{3}{2} \lambda f^{n-2+p} + \lambda \, \mc{O} ( f^{n-1} ) \right] \, | \nablam_N \ms{\Psi} |_{ \mf{h} }^2 \\
\notag &\qquad + \frac{1}{4} K \rho^2 f^{n-2} \sum_{ A=1 }^n | \nablam_{ E_A } \ms{\Psi} |_{ \mf{h} }^2 + ( \kappa^2 -n \kappa + \sigma ) f^{n-2} \, | \ms{\Psi} |_{ \mf{h} }^2 \\
\notag &\qquad + \left[ \left( 1 - \frac{p}{2} \right) ( 2 \kappa - n + p ) \lambda f^{n-2+p} + (1-p) \lambda^2 f^{n-2+2p} \right] \, | \ms{\Psi} |_{ \mf{h} }^2 \\
\notag &\qquad + \lambda^2 \mathcal{O}(f^{n-1}) \, | \ms{\Psi} |_{ \mf{h} }^2 + \div_g ( P^S + P^Q ) \text{.}
\end{align}

Now, for any $q \in \R$, we expand the inequality
\[
0 \leq f^{ q - 2 } \left| \nablam_{ \nabla^\sharp f } \ms{\Psi} + \frac{1}{2} ( q - n ) f \, \ms{\Psi} \right|_{ \mf{h} }^2 \text{,}
\]
and we recall the asymptotics from \eqref{AsympGradientf}, \eqref{AsympWaveOperatorf}, and Lemma \ref{EstimatesVerticalRiemannianMetrich} to obtain
\footnote{See the derivation of \cite[Equations (5.78)--(5.79)]{Arick3} for details.}
\begin{align} \label{Carleman_Hardy}
f^q \, | \nablam_N \ms{\Psi} |_{ \mf{h} }^2 &\geq \frac{1}{4} (q-n)^2 f^q \, | \ms{\Psi} |_{ \mf{h} }^2 + ( 1 + q^2 ) \, \mc{O} ( f^{q+2} ) \, ( | \nablam_N \ms{\Psi} |_{ \mf{h} }^2 + | \ms{\Psi} |_{ \mf{h} }^2 ) \\
\notag &\qquad + (q-n) \, \mc{O} ( \rho f^{q+1} ) \, | \ms{\Psi} |_{ \mf{h} }^2 + \div_g P^{H,q} \text{,}
\end{align}
where $P^{H,q}$ is defined as
\begin{equation} \label{Carleman_Hardy_terms}
P^{H,q} := - \frac{1}{2} (q-n) f^{q-1} \nabla^\sharp f \, | \ms{\Psi} |_{ \mf{h} }^2 \text{.}
\end{equation}

We then apply the inequality \eqref{Carleman_Hardy} twice, with $q := n-2$ and $q := n-2+p$, to the estimate \eqref{ExpressionJ_10}.
Noting in particular that $2\kappa - n + 1 - \mc{C}_0 \geq 0$ by \eqref{carleman_ass_kappa}, we then obtain
\begin{align} \label{ExpressionJ_20}
\lambda^{-1} f^{n-2-p} | \mc{L}^\dagger \ms{\Psi} |_{ \mf{h} }^2 &\geq \left[ \frac{1}{2} \lambda f^{n-2+p} + \lambda \, \mc{O} (f^{n-1}) \right] | \nablam_N \ms{\Psi} |_{ \mf{h} }^2 + \frac{1}{4} K \rho^2 f^{n-2} \sum_{A=1}^n | \nablam_{ E_A } \ms{\Psi} |_{ \mf{h} }^2 \\
\notag &\qquad + ( \kappa^2 - (n-2) \kappa + \sigma - (n-1) - \mc{C}_0 ) f^{n-2} \, | \ms{\Psi} |_{ \mf{h} }^2 \\
\notag &\qquad + \left( 1 - \frac{p}{2} \right) ( 2 \kappa - n + 1 + \frac{p}{2} ) \lambda f^{n-2+p} \, | \ms{\Psi} |_{ \mf{h} }^2 \\
\notag &\qquad + [ (1-p) \lambda^2 f^{n-2+2p} + \lambda^2 \, \mc{O} ( f^{n-1} ) ] \, | \ms{\Psi} |_{ \mf{h} }^2 + \div_g P \text{,}
\end{align}
where we have set
\begin{equation} \label{ExpressionJ_P}
P := P^S + P^Q + \lambda \, P^{H, n-2+p} + ( 2 \kappa - n + 1 - \mc{C}_0 ) \, P^{H, n-2} \text{.}
\end{equation}
Observe that \eqref{ExpressionJ_20} immediately implies the desired \eqref{Pointwise_Psi}, once we set $f_\star$ sufficiently small and $\lambda$ sufficiently large as in \eqref{carleman_ass_params}.
Furthermore, by the definitions \eqref{ExpressionJ_error} and \eqref{Carleman_Hardy_terms}, we see that the quantity $P$ defined in \eqref{ExpressionJ_P} precisely matches that of \eqref{Pointwise_P}.

It remains to prove the bound \eqref{Pointwise_Prho} for $P ( \rho \partial_\rho )$; as before, the steps mirror those in the end of the proof of \cite[Lemma 5.28]{Arick3}.
First, by Lemma \ref{LemmaDerivativesf}, \eqref{AsympGradientf}, \eqref{DefinitionS}, and the crude estimate $\mc{A} = \lambda^2 \, \mc{O} (1)$ from Lemma \ref{LemmaEstimatesConjugatedWaveOperator}, we obtain
\begin{align} \label{Prho_1}
| P^S ( \rho \partial_\rho ) | &\leq \lambda^2 \, \mc{O} ( f^{n-2} ) \, | \ms{\Psi} |_{ \mf{h} }^2 \text{,} \\
\notag | \lambda \, P^{H, n-2+p} ( \rho \partial_\rho ) + ( 2 \kappa - n + 1 - \mc{C}_0 ) \, P^{H, n-2} ( \rho \partial_\rho ) | &\leq \lambda \, \mc{O} ( f^{n-2} ) \, | \ms{\Psi} |_{ \mf{h} }^2 \text{.}
\end{align}

The bound $P^Q$ is similar, but there are more terms involved.
We begin by expanding
\begin{align} \label{Prho_2}
P^Q ( \rho \partial_\rho ) &= \frac{1}{2} \rho [ 2 \, \hm ( \nablam_\rho \ms{\Psi}, \nablam_S \ms{\Psi} ) + \nablam_\rho \hm  ( \ms{\Psi}, \nablam_S \ms{\Psi} ) + \nablam_S \hm ( \nablam_\rho \ms{\Psi}, \ms{\Psi}) ] \\ 
\notag &\qquad - \frac{1}{2} \rho g ( \partial_\rho, S ) \, g^{ \mu \nu } [ \hm ( \nablam_\mu \ms{\Psi}, \nablam_\nu \ms{\Psi} ) + \nablam_\mu \hm ( \ms{\Psi}, \nablam_\nu \ms{\Psi} ) ] \\
\notag &\qquad + \frac{1}{2} \rho v_{\zeta} [ 2 \, \hm ( \nablam_\rho \ms{\Psi}, \ms{\Psi} ) + \nablam_\rho \hm ( \ms{\Psi}, \ms{\Psi} ) ] - \frac{1}{2} \rho \nabla_\rho v_\zeta \, | \ms{\Psi} |_{ \mf{h} }^2 \\
\notag &:= B_1 + B_2 + B_3 + B_4 \text{.}
\end{align}
Following the proof of \cite[Lemma 5.28]{Arick3}, by applying \eqref{AsympGradientf} and Lemma \ref{EstimatesVerticalRiemannianMetrich}, we estimate
\begin{align} \label{Prho_3}
| B_1 | &\leq \mc{O} ( \rho^2 f^{n-2} ) \, | \Dm_\rho \ms{\Psi} |_{ \mf{h} }^2 + \mc{O} ( \rho^2 f^n ) \, | \ms{D} \ms{\Psi} |_{ \mf{h} }^2 + \mc{O} ( \rho^2 f^n ) \, | \ms{\Psi} |_{ \mf{h} }^2 \text{,} \\
\notag | B_3 + B_4 | &\leq \mc{O} ( \rho^2 f^n ) \, | \Dm_\rho \ms{\Psi} |_{ \mf{h} }^2 + \mc{O} ( f^n ) \, | \ms{\Psi} |_{ \mf{h} }^2 \text{.}
\end{align}
For $B_2$, we apply Lemma \ref{LemmaMetric}, \eqref{Derivativesf}, and Lemma \ref{EstimatesVerticalRiemannianMetrich}, which yields
\footnote{In contrast to \cite[Lemma 5.28]{Arick3}, it is easier to expand in terms of coordinates rather than frames.}
\begin{align} \label{Prho_4}
B_2 &\leq \frac{1}{2} f^{n-2} ( - \rho^2 | \Dm_\rho \ms{\Psi} |_{ \mf{h} }^2 + \rho^2 | \ms{D} \ms{\Psi} |_{ \mf{h}^2 } ) + \mc{O} ( f^{n-2} ) \, ( \rho^3 \, | \nablam_\rho \ms{\Psi} |_{ \mf{h} } + \rho^2 \, | \ms{D} \ms{\Psi} |_{ \mf{h} } ) \, | \ms{\Psi} |_{ \mf{h} } \\
\notag &\leq \mc{O} ( \rho^2 f^n ) \, | \Dm_\rho \ms{\Psi} |_{ \mf{h} }^2 + \mc{O} ( \rho^2 f^{n-2} ) \, | \ms{D} \ms{\Psi} |_{ \mf{h} }^2 + \mc{O} ( \rho^2 f^{n-2} ) \, | \ms{\Psi} |_{ \mf{h} }^2 \text{.}
\end{align}
Finally, combining \eqref{Prho_1}--\eqref{Prho_4} yields the desired estimate \eqref{Pointwise_Prho} for $P ( \rho \partial_\rho )$.
\end{proof}

\subsubsection{Pointwise Estimate for $\ms{\Phi}$}

We now convert Lemma \ref{LemmaPointwiseEstimateModifiedConjugatedWaveOperatorPsi} into a pointwise estimate for the wave operator $\smash{\Boxm + \sigma + \rho^2 \nablam_X}$ and for the original unknown $\ms{\Phi}$.

\begin{lemma}[Main pointwise estimate] \label{LemmaPointwiseEstimateBreveWaveOperator}
There exists a constant $\mc{C} > 0$ (depending on $\ms{g}$, $\mc{D}$, $X$, $k$, $l$) such that the following inequality holds on $\mc{D} ( f_\star )$,
\begin{align} \label{Pointwise_Phi}
\lambda^{-1} \mathcal{E} f^{-p} \, | ( \Boxm + \sigma + \rho^2 \nablam_X ) \ms{\Phi} |_{ \mf{h} }^2 &\geq \mc{C} \mc{E} \rho^4 \, ( | \Dm_\rho \ms{\Phi} |_{ \mf{h} }^2 + | \ms{D} \ms{\Phi} |_{ \mf{h} }^2 ) \\
\notag &\qquad + \frac{1}{8} \lambda^2 \mc{E} f^{2p} \, | \ms{\Phi} |_{ \mf{h} }^2 + \div_g P \text{,}
\end{align}
where $P$ is as in Lemma \ref{LemmaPointwiseEstimateModifiedConjugatedWaveOperatorPsi}, and where the weight $\mc{E}$ is given by
\begin{equation} \label{Pointwise_E}
\mc{E} := e^{-2F} f^{n-2} = e^{ -2 \lambda p^{-1} f^p } f^{ n - 2 - 2 \kappa } \text{.}
\end{equation}
Furthermore, there exists $\mc{C}_b> 0$ (depending on $\ms{g}$, $\mc{D}$, $X$, $k$, $l$) such that
\begin{equation} \label{Pointwise_PPrho}
\rho^{-n} \, P ( \rho \partial_\rho ) \leq \mc{C}_b \, [ | \Dm_\rho ( \rho^{-\kappa} \ms{\Phi} ) |_{ \mf{h} }^2 + | \ms{D} ( \rho^{-\kappa} \ms{\Phi} ) |_{ \mf{h} }^2 + \lambda^2 \, | \rho^{ -\kappa - 1 } \ms{\Phi} |_{ \mf{h} }^2 ] \text{.}
\end{equation}
\end{lemma}

\begin{proof}
The proof is analogous to that of \cite[Lemma 5.29]{Arick3}, so we again omit many of the common details.
First, applying the definitions \eqref{DefinitionauxiliaryF} and \eqref{DefinitionOperators} for $\Psi$ and $\mc{L}^\dagger$ to Lemma \ref{LemmaPointwiseEstimateModifiedConjugatedWaveOperatorPsi}, we see that there exists $\mc{C}> 0$ (depending on $\ms{g}$, $\mc{D}$, $X$, $k$, $l$) with
\begin{align} \label{Pointwise_Phi_1}
\lambda^{-1} \mc{E} f^{-p} \, | ( \Boxm + \sigma + \rho^2 \nablam_X ) \ms{\Phi} |_{ \mf{h} }^2 &\geq \mc{C} \rho^2 f^{n-2} \left( | \nablam_N \ms{\Psi} |_{ \mf{h} }^2 + \sum_{A=1}^n | \nablam_{ E_A } \ms{\Psi} |_{ \mf{h} }^2 \right) \\
\notag &\qquad + \frac{1}{4} \lambda^2 \mc{E} f^{2p} \, | \ms{\Phi} |_{ \mf{h} }^2 + \div_g P \text{,}
\end{align}
where the frames $\{ N, E_A \}$ are as in Lemma \ref{LemmaPointwiseEstimateModifiedConjugatedWaveOperatorPsi}.

From \eqref{Normal}, \eqref{Va}, and the definition of the $E_A$'s, we can estimate
\begin{align} \label{Pointwise_Phi_2}
\rho^4 | \Dm_\rho \ms{\Phi} |_{ \mf{h} }^2 &\leq \mc{O} ( \rho^2 ) \, | \nablam_N \ms{\Phi} |_{ \mf{h} }^2 + \mc{O} ( f^2 \rho^2 ) \sum_{ A = 1 }^n | \nablam_{ E_A } \ms{\Phi} |_{ \mf{h} }^2 \text{,} \\
\notag \rho^4 | \ms{D} \ms{\Phi} |_{ \mf{h} }^2 &\leq \mc{O} ( f^2 \rho^2 ) \, | \nablam_N \ms{\Phi} |_{ \mf{h} }^2 + \mc{O} ( \rho^2 ) \sum_{ A = 1 }^n | \nablam_{ E_A } \ms{\Phi} |_{ \mf{h} }^2 \text{.}
\end{align}
Moreover, by \eqref{AsympGradientf}, \eqref{Normal}, Lemma \ref{LemmaVa}, \eqref{carleman_ass_kappa}, and \eqref{DefinitionauxiliaryF},
\begin{equation} \label{Pointwise_Phi_3}
e^{-2F} | \nablam_N \ms{\Phi} |_{ \mf{h} }^2 \leq | \nablam_N \ms{\Psi} |_{ \mf{h} }^2 + \lambda^2 \, \mc{O} (1) \, | \ms{\Psi} |_{ \mf{h} }^2 \text{,} \qquad e^{-2F} | \nablam_{ E_A } \ms{\Phi} |_{ \mf{h} }^2 = | \nablam_{ E_A } \ms{\Psi} |_{ \mf{h} }^2 \text{.}
\end{equation}
From \eqref{Pointwise_Phi_2} and \eqref{Pointwise_Phi_3}, we then conclude
\[
\mc{E} \rho^4 ( | \Dm_\rho \ms{\Phi} |_{ \mf{h} }^2 + | \ms{D} \ms{\Phi} |_{ \mf{h} }^2 ) \leq \mc{O} ( \rho^2 f^{n-2} ) \, \left( | \nablam_N \ms{\Psi} |_{ \mf{h} }^2 + \sum_{A=1}^n | \nablam_{ E_A } \ms{\Psi} |_{ \mf{h} }^2 \right) + \mc{O} (f^2) \, \lambda^2 \mc{E} | \ms{\Phi} |_{ \mf{h} }^2 \text{.}
\]
Combining \eqref{Pointwise_Phi_1} and the above results in the estimate \eqref{Pointwise_Phi}.

It remains to show \eqref{Pointwise_PPrho}.
For this, we apply \eqref{Derivativesf} and \eqref{DefinitionauxiliaryF}, and recall that $\rho \leq f$ (by \eqref{Definitionf}) and that $2 \kappa \geq n - 1$ (by \eqref{carleman_ass_kappa}).
Putting all this together, we obtain
\begin{align*}
f^{n-2} \rho^{-n} | \ms{\Psi} |_{ \mf{h} }^2 &\leq | \rho^{-\kappa-1} \ms{\Phi} |_{ \mf{h} }^2 \text{,} \\
f^{n-2} \rho^{-n+2} | \Dm_\rho \ms{\Psi} |_{ \mf{h} }^2 &\leq | \Dm_\rho ( \rho^{-\kappa} \ms{\Phi} ) |_{ \mf{h} }^2 + \lambda^2 \, \mc{O} (1) \, | \rho^{-\kappa-1} \ms{\Phi} |_{ \mf{h} }^2 \text{,} \\
f^{n-2} \rho^{-n+2} | \ms{D} \ms{\Psi} |_{ \mf{h} }^2 &\leq | \ms{D} ( \rho^{-\kappa} \ms{\Phi} ) |_{ \mf{h} }^2 + \lambda^2 \, \mc{O} (1) \, | \rho^{-\kappa-1} \ms{\Phi} |_{ \mf{h} }^2 \text{.}
\end{align*}
The bound \eqref{Pointwise_PPrho} now follows immediately from \eqref{Pointwise_Prho} and the above.
\end{proof}

\subsubsection{Completion of the Proof} \label{CompletionProof}

The final step is to integrate our pointwise inequality over $\mc{D} ( f_\star )$.
\footnote{The steps are analogous to \cite[Section 5.3.4]{Arick3}; the current situation is in fact easier, since our version of $f$ is smooth, while \cite{Arick3} also had to deal with discontinuities in higher derivatives of $f$.}
Since $\mc{D} ( f_\star )$ has infinite volume, we first apply an additional cutoff to $\mc{D} ( f_\star )$:
\begin{equation} \label{Dfrho}
\mc{D} ( f_\star, \rho_\star ) := \mc{D} ( f_\star ) \cap \{ \rho > \rho_\star \} \text{,} \qquad 0 < \rho_\star \ll f_\star \text{.}
\end{equation}
Integrating \eqref{Pointwise_Phi} over $\mc{D} ( f_\star, \rho_\star )$ then yields
\begin{align} \label{Integral_Phi_1}
&\int_{ \mc{D} ( f_\star, \rho_\star ) } \mc{E} f^{-p} | ( \Boxm + \sigma + \rho^2 \nablam_X ) \ms{\Phi} |_{ \mf{h} }^2 \, d \mu_g - \lambda \int_{ \mc{D} ( f_\star, \rho_\star ) } \div_g P \, d \mu_g \\		
\notag &\quad \geq \lambda \mc{C} \int_{ \mc{D} ( f_\star, \rho_\star ) }	 \mc{E} \rho^4 ( | \Dm_\rho \ms{\Phi} |_{ \mf{h} }^2 + | \ms{D} \ms{\Phi} |_{ \mf{h} }^2 ) \, d \mu_g + \frac{1}{8} \lambda^3 \int_{ \mc{D} ( f_\star, \rho_\star ) } \mc{E} f^{2p} | \ms{\Phi} |_{ \mf{h} }^2 \, d \mu_g \text{,}
\end{align}
where $\mc{C}$ and $\mc{E}$ are as in the statement of Lemma \ref{LemmaPointwiseEstimateBreveWaveOperator}.

Observe that $\partial \mc{D} ( f_\star, \rho_\star )$ consists two components:
\begin{itemize}
\item The first is $\{ f = f_\star \}$, on which both $\ms{\Phi}$ and $\nablam \ms{\Phi}$ are assumed to vanish.

\item The second is $\mc{D} ( f_\star ) \cap \{ \rho = \rho_\star \}$; observe that $\rho^{-n} \, d \mu_{ \ms{g} }$ is the induced volume measure, and that $- \rho \partial_\rho$ is the outward-facing unit normal.
\end{itemize}
Thus, applying the divergence theorem and recalling \eqref{Pointwise_PPrho} yields
\begin{align*}
- \int_{ \mc{D} ( f_\star, \rho_\star ) }	\div_g P \, d \mu_g &= \int_{ \mc{D} ( f_\star ) \cap \{ \rho = \rho_\star \} } \rho^{-n} \, P ( \rho \partial_\rho ) \, d \mu_{ \ms{g} } \\
\notag &\quad \leq \mc{C}_b \lambda^2 \int_{ \mc{D} ( f_\star ) \cap \{ \rho = \rho_\star \} }	[ | \Dm_\rho ( \rho^{-\kappa} \ms{\Phi} ) |_{ \mf{h} }^2 + | \ms{D} ( \rho^{-\kappa} \ms{\Phi} ) |_{ \mf{h} }^2 + | \rho^{-\kappa-1} \ms{\Phi} |_{ \mf{h} }^2 ] \, d \mu_{ \ms{g} } \text{.}
\end{align*}
Finally, the Carleman estimate \eqref{CarlemanEstimate} follows from combining \eqref{Integral_Phi_1} with the above, taking $\rho_\star \searrow 0$, and then applying the monotone convergence theorem.

\raggedright

\bibliographystyle{plain}
\bibliography{ads}

\end{document}